\documentclass{article}
\usepackage{authblk} 

\usepackage[margin=1in]{geometry}
\usepackage{graphicx,url} 
\usepackage{amsmath,amssymb,amsthm,amsfonts}
\newcommand{\R}{{\mathbb R}}\newcommand{\N}{{\mathbb N}}
\newcommand{\Z}{{\mathbb Z}}

\newcommand{\sech}{\mathrm{sech}}

\DeclareMathOperator{\const}{\mathrm{const.}}
\DeclareMathOperator{\sinc}{\mathrm{sinc}}

\usepackage{times,fourier}
\usepackage{color}
\usepackage{cancel}
\usepackage[normalem]{ulem}
\usepackage{todonotes}
\usepackage[colorlinks=true,citecolor=blue]{hyperref}
\usepackage{mathtools}

\def\dn{\mathop{\rm dn}\nolimits}

\DeclareMathOperator{\range}{Rg}
\numberwithin{equation}{section}
\def\@#1{\mathbf{#1}}

\newtheorem{theorem}{Theorem}[section]

\newtheorem{lemma}[theorem]{Lemma} 

\newtheorem{remark}[theorem]{Remark}

\def\bse{\begin{subequations}}
\def\ese{\end{subequations}}

\title{Traveling and Dispersive Shock Waves in a Two-Dimensional Fermi-Pasta-Ulam-Tsingou Lattice}

\author[1]{Christopher Chong}

\author[2,3,4]{Panayotis G. Kevrekidis}

\author[5]{Gino Biondini}

\author[6]{Wolfgang Reichel}

\affil[1]{Department of Mathematics, Bowdoin College, Brunswick, Maine 04011}
\affil[2]{Department of Mathematics and Statistics, University of Massachusetts, Amherst, Massachusetts 01003}
\affil[3]{Department of Physics, University of Massachusetts Amherst, Massachusetts 01003,USA}
\affil[4]{Department of Mechanical Engineering, Seoul National University,
1 Gwanak-ro, Gwanak-gu, Seoul 08826, South Korea}

\affil[5]{Department of Mathematics, State University of New York at Buffalo, Buffalo, New York 14260}
\affil[6]{Institute for Analysis, Karlsruhe Institute of Technology (KIT), D-76128 Karlsruhe, Germany.}
\date{\small\today}

\begin{document}

\maketitle

\begin{abstract}
In the present work we analyze traveling and dispersive shock waves of a two-dimensional Fermi-Pasta-Ulam-Tsingou lattice. 
In the first part of the paper, using variational techniques we prove the existence of both periodic and solitary traveling waves for convex potentials. 
In the case of unimodal profiles we are able to remove the assumption of convexity. 
The variational formulation also provides a natural algorithm for the numerical computation of traveling waves, which we use to explore both solitary and periodic traveling waves.
The numerical computations are compared with analytical approximations based on the derivation of the KdV equation for quasi-one-dimensional propagation.
In the  second  part of the paper, we focus on dispersive shock waves (DSWs), which are expanding modulated waves that connect states of different amplitude.
In particular, we focus on line DSWs, which are constant along one direction and propagate in the direction orthogonal to which  it is constant. 
Such solutions form when subject to  quasi-one-dimensional jump initial data.
We find that while the shape of the DSW depends on the direction of travel, properties such as the speed and amplitude do not.
The systematic numerical study of the line~DSWs is then compared to those predicted by the KdV equation along the line of propagation.
Key characteristics of the DSWs, such as the speeds of the trailing and leading edges, are investigated for various jump heights,  yielding good agreement between simulation and KdV approximation in the limit of vanishing jump height.
Finally, we apply the DSW fitting method to study the trailing
and leading edge characteristics of the DSW, finding even better
agreement to the numerics when compared to the KdV prediction.
The KdV prediction and DSW fitting predictions agree in the limit of small jump height.

\end{abstract}

\section{Introduction}

A dispersive shock wave (DSW) is a type of nonlinear, expanding wave structure that arises in dispersive media, where a sharp gradient in initial conditions evolves into a region of rapid, modulated oscillations due to the interplay of nonlinearity and dispersion. 
The study of DSWs has gained considerable recent attention in a variety of settings, including, but not limited to, mechanical metamaterials, superfluids and atomic
gases, nonlinear optics, water waves, and plasmas, as summarized, e.g., in~\cite{scholar,Mark2016,Whitham74}. 
Much of the research on DSWs has focused on continuous systems,
although discrete systems have also gained some attention. For example, granular media offered some of the early experimental realizations of {\it discrete}
DSWs, e.g., in the works of~\cite{Herbold07}  and~\cite{Molinari2009}.
More recent examples of discrete DSWs include those
in chains of hollow elliptic cylinders~\cite{HEC_DSW},
as well as in tunable magnetic lattices~\cite{talcohen}.

Dispersive shock waves in two-dimensional (2D) systems are far less explored, even in continuous media. 
Some relevant examples include the study of DSWs in the  Kadomtsev-Petviashvili (KP)  equation, in which models
like the cylindrical  Korteweg-de Vries (KdV)  equation \cite{PHYSD333p84}
or  modified KdV  equation \cite{2DmKdV_mod2023} were derived 
to analyze the resulting DSWs.
Full modulation theory in higher-dimensional models has also been explored \cite{modulation2D}, although not explicitly for the study of DSWs. 

Recently, a general approach to modulation theory for systems in multiple spatial dimensions was formulated in \cite{Ablowitz2017},
where the Whitham modulation equations for the KP equation were studied.
The resulting equations were then used in \cite{Ablowitz2017} to study the transverse stability of the elliptic traveling wave solutions of the KP equation.
Suitable reductions of the modulation equations were also used in \cite{JFM2021v909pA24,NLTY2021v34p3583,KPII_Mach2022} to study 
a variety of dynamical problems for the KP equation, including the Mach expansion of solitons and the oblique interactions between solitons, DSWs and rarefaction waves. 
A similar methodology was then used to derive modulation equations and study the stability of periodic waves 
in the 2D Benjamin-Ono equation \cite{PRE2017v96p032225}, the 2D NLS equation \cite{2DNLS_mod2022,2DNLS_mod2023} and the 2D Zakharov–Kuznetsov equation \cite{2DZK_mod2023}
 among others.   

 On the other hand, 
while periodic waves \cite{Bak2011,Feckan2007} and solitons \cite{Bak2011, Friesecke2003,Wattis94, jpa2010,Herrmann2018} have been
explored in various 2D lattice settings,
to the best of the authors' knowledge, there are no studies of DSWs in 2D lattices. 
In the present paper, we explore DSWs in a scalar  2D  analog of
the classical Fermi-Pasta-Ulam-Tsingou (FPUT) equation. 
Since a DSW is a periodic wave that is modulated in such a way as to connect states of different amplitude, it is crucial to first have an understanding of traveling solitary and periodic waves of the underlying model.
Indeed, the former appear at the front or rear of the DSWs
(depending on the DSW orientation), while the latter are modulated
throughout the DSW.
For this reason, the first part of our paper is dedicated to traveling waves. 
 Specifically, in section~\ref{sec:TWs} 
we prove the existence of periodic and solitary waves using a variational approach and develop a numerical fixed-point algorithm built upon that. 
We then derive a  KdV equation approximation  to analytically describe such solutions 
in the limit of small amplitude. 
The second part of the paper is dedicated to the study of DSWs using systematic numerical simulations
and an analytical approximation using the KdV equation. 
 In section~\ref{sec:DSWs}, we explore  
the spatial profiles and leading and trailing edge characteristics of the DSWs  for various  values of the  parameters, showing good agreement between simulation and asymptotic approximation. 
This being an early study of 2D lattice DSWs, we restrict our attention to ''line" solutions, namely ones that travel along a single propagation direction along the lattice and are constant in the direction orthogonal to the propagation. While properties such as the amplitude and speed do not depend on the angle of propagation, we find that the wavenumber and frequency do. 
 The remainder of this section contains the detailed problem statement and model equations,
while 
section~\ref{s:conclusions} offers some concluding remarks and possible future directions for further work. 
Various technical results are relegated to an appendix.  

\subsection*{Problem Statement and Model Equations}

\begin{figure}[t!] %
    \centerline{
   \begin{tabular}{@{}p{0.42\linewidth}@{}p{0.42\linewidth}@{}}
     \rlap{\hspace*{5pt}\raisebox{\dimexpr\ht1-.1\baselineskip}{\bf (a)}}
 \includegraphics[height=5cm]{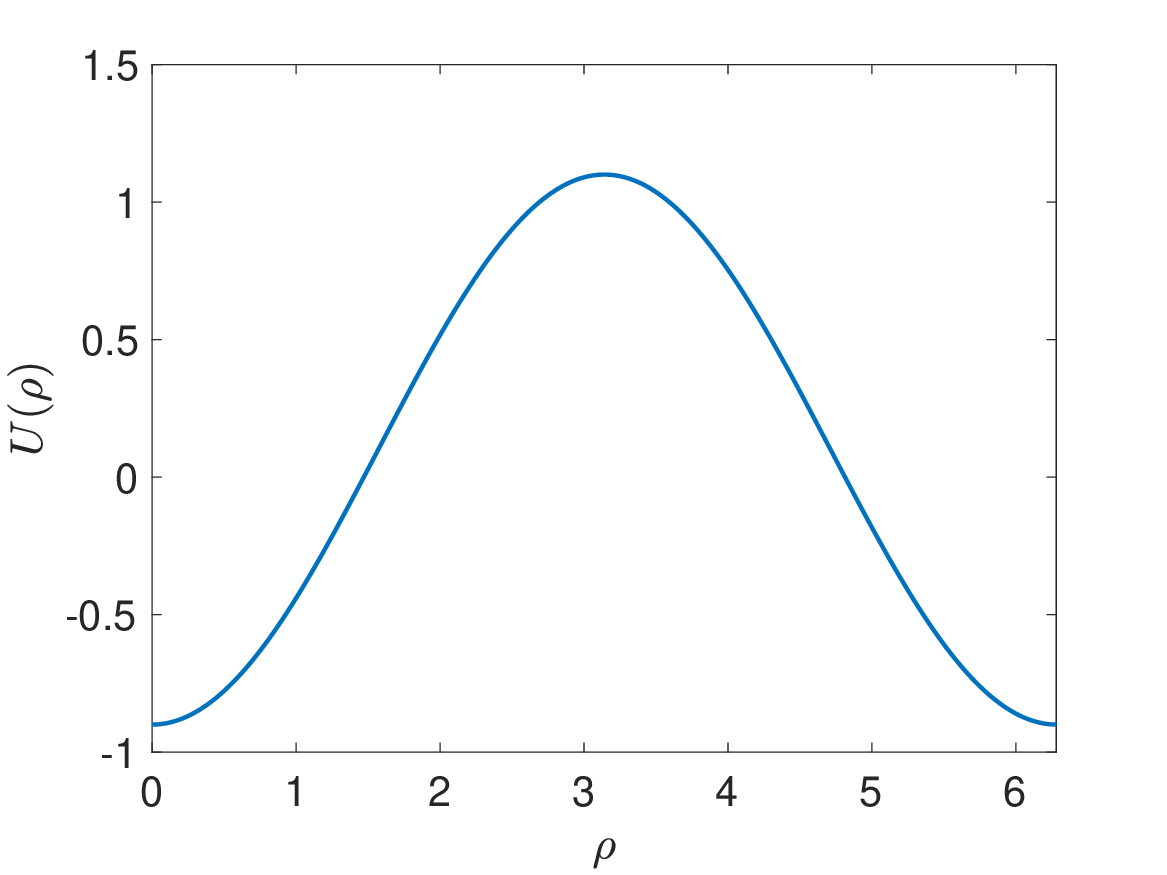}  &
   \rlap{\hspace*{5pt}\raisebox{\dimexpr\ht1-.1\baselineskip}{\bf (b)}}
 \includegraphics[height=5cm]{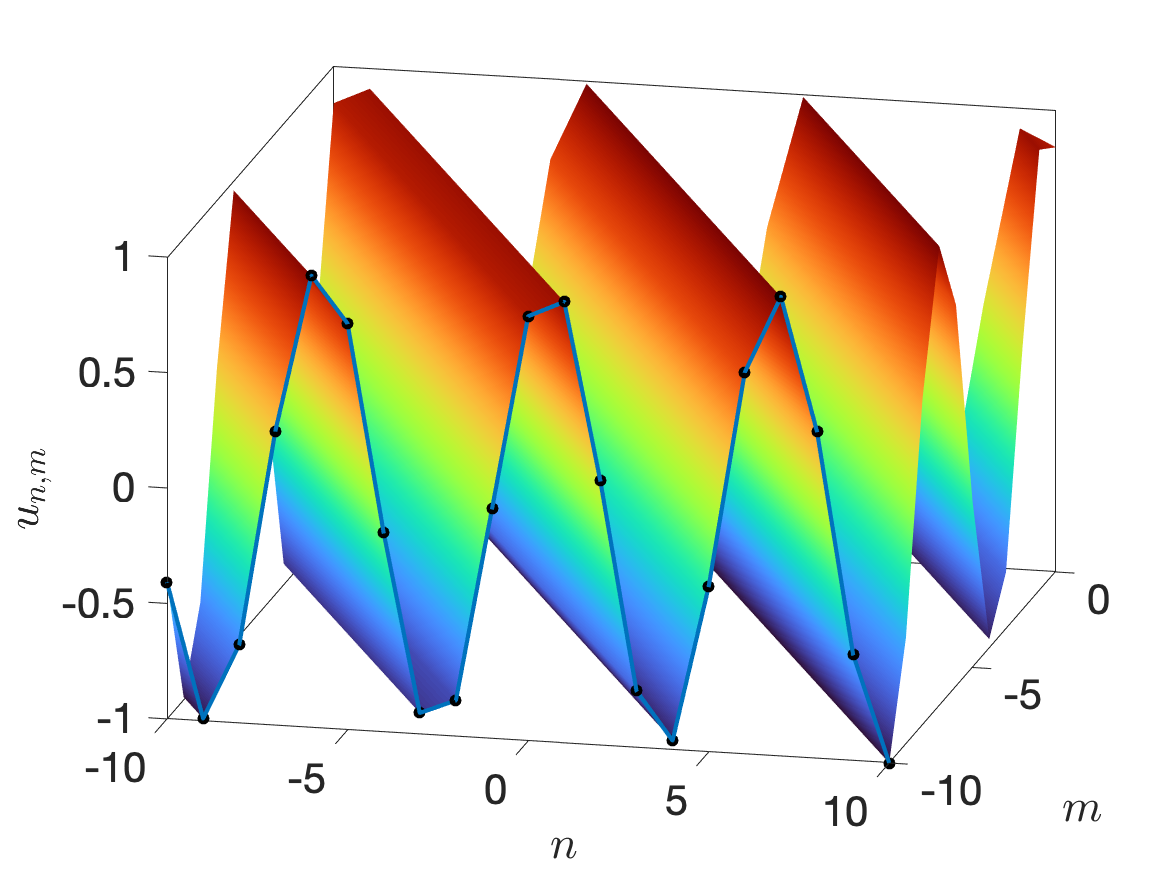} 
  \end{tabular}
  }
   \caption{\textbf{(a)} Example of a unimodal periodic traveling wave profile
plotted against the coordinate $\rho$.  Both 1D lattice and 2D lattice
travelings are fully described by the profile as a function of just
$\rho.$ \textbf{(b)} 
   The traveling wave shown in (a) in a 2D representation with
   the wavevector chosen such that $r=1$ and $k_2/k_1=1$.
   }
   \label{fig:TWexamples}
\end{figure}

The 1D FPUT equation \cite{FPU55} written in the difference form is 
\begin{equation} \label{eq:model1D}
    \ddot{u}_{n} = \Phi'(u_{n-1}) - 2 \Phi'(u_n) + \Phi'(u_{n+1}), \quad  n\in \Z.
\end{equation}
The classical situation is where the potential function ${\Phi: \R\to\R}$ has the form
\begin{equation}  \label{eq:def_phi}
\Phi(u) = \frac{a_1}{2} u^2 + \frac{a_2}{3} u^3 + \frac{a_3}{4} u^4. 
\end{equation}
This model is relevant in a host of settings, including
phononic, electrical, and biological systems (among others) \cite{Berman2005}.
 The linearized
equations have plane wave solutions of the form $u_n(t) = e^{i(kn-\omega t)}$
where the dispersion relation is given by $$\omega^2(k) = 4a_1\sin^2(k/2),$$
where $k$ is the wavenumber, $\omega$ is the frequency and $a_1$ is
the linear coefficient of $\Phi'(u)$. Note that the sound speed
is $\sqrt{a_1}$. More interesting is the plethora of nonlinear behaviors that
can be found in the FPUT model. This includes the classical 
recurrence phenomenon \cite{Berman2005,FPUreview}, breathers \cite{Flach2007},
traveling waves \cite{VAINCHTEIN2022,pegof2,pegogf1,pegof3,pegof4} and
DSWs \cite{Holian81,ruffo}. Of the extensive list of such behavior,
the latter two are most relevant for the present paper. We
briefly touch on those topics now.

Traveling waves of Eq.~\eqref{eq:model1D}
have the form 
\begin{equation}\label{eq:ansatz1D}
u_n(t) = U(\rho) + \bar{u}, \qquad \rho = kn - \omega t
\end{equation}
where $k$ is the wavenumber, $\omega$ is the frequency, and $\bar{u}$ is the mean.
The wave profile
$U$ satisfies the advance delay equation
\begin{equation} 
\label{eq:advdelay1D}
     \omega^2 \frac{d^2U}{d\rho^2}(\rho)  = \Phi'(\bar{u}+
   U(\rho-k))-2\Phi'(\bar{u}+
   U(\rho)) + \Phi'(\bar u
   +U(\rho+k)), 
\end{equation} 
which is obtained upon substituting Eq.~\eqref{eq:ansatz1D} into Eq.~\eqref{eq:model1D}.
Thus, while traveling waves depend on space and time, they can be fully described
by a single dependent variable by virtue of Eq.~\eqref{eq:advdelay1D}. 
See Fig.~\ref{fig:TWexamples}(a) for example.
The existence of traveling waves has been proven for a large class
of lattice models, including Eq.~\eqref{eq:model1D}, see for example \cite{Wattis,MacKay99,Venakides99,DHM06,Stefanov,PankovFPU}.

Dispersive shock waves, on the other hand, will form in Eq.~\eqref{eq:model1D}
when initialized with Riemann (step) initial data. For example,
\begin{equation}\label{eq:step1}
u_n(0) = \begin{cases}
    \delta{\sqrt{a_1}}/{a_2}, \quad & n \leq 0,\\
    0, \quad & n > 0,
    \end{cases}
\end{equation}
where $\delta {\sqrt{a_1}}/{a_2}$ is the jump height and the initial velocity is chosen appropriately;
$\dot{u}_n=0$ is one natural choice.
The factor ${\sqrt{a_1}}/{a_2}$ will 
simplify notation later on.
Figure~\ref{fig:DSWexamples}(a) shows an example
of an initial condition that leads to the formation of a DSW, shown in 
Fig.~\ref{fig:DSWexamples}(b). The oscillatory region of the DSW
is often referred to as the ``core" of the DSW. Where the oscillatory
region meets the constant state at the  back (around $u=0.2$ 
in the example) is called the trailing edge, while
the peak at the front of the DSW (located near lattice site $n=120$
in the example) is called the leading edge, which is
actually well described by a solitary wave solution
of the lattice. In particular, it is a traveling wave with zero wavenumber.
Indeed, in each small window of time and space (not just at the front)
the DSW shown in Fig.~\ref{fig:DSWexamples}(b) appears to be 
a traveling wave of the form shown in Eq.~\eqref{eq:ansatz1D}, but with
 values of the underlying parameters (e.g., $k,\omega,\bar{u}$)
depending on the spatio-temporal window. Hence,
to describe a DSW, one often seeks equations describing
how those wave parameters vary across the core of the DSW. The resulting
equations are called Whitham modulation equations \cite{Whitham74}.
While they are well-known for the 1D model \cite{Venakides99,DHM06,blochkodama} the equations
themselves can be cumbersome to work with. Despite this,
there has been much recent progress on the study of DSWs in 1D
lattices. For example, recent work has leveraged integrable models such as the Korteweg-de Vries
(KdV) equation and the Toda lattice~\cite{Ari2024,Gino2024}
to approximate 1D DSWs. Asymptotic techniques such
as the DSW fitting method of~\cite{El2005} have also been used
to quantitatively characterize problems in such 1D discrete~\cite{Sprenger2024} 
or/and metamaterial settings~\cite{yang2024regularizedcontinuummodeltraveling}.
Data-driven methods including dimension-reduction have
also been brought to bear to produce effective,
as well as quasi-continuum variants of the relevant
lattice models~\cite{CHONG2022133533}.

\begin{figure}[t!] 
    \centerline{
   \begin{tabular}{@{}p{0.42\linewidth}@{}p{0.42\linewidth}@{}}
     \rlap{\hspace*{5pt}\raisebox{\dimexpr\ht1-.1\baselineskip}{\bf (a)}}
 \includegraphics[height=5cm]{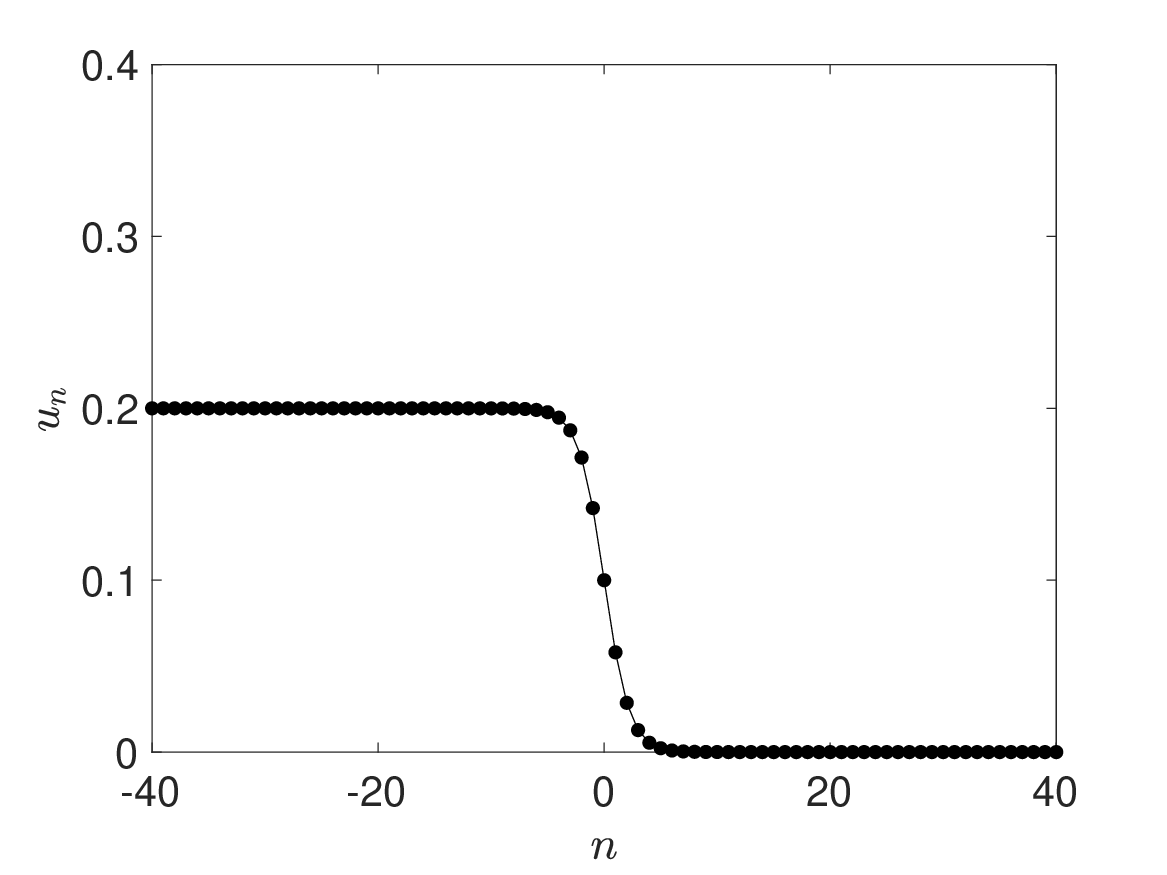}  &
   \rlap{\hspace*{5pt}\raisebox{\dimexpr\ht1-.1\baselineskip}{\bf (b)}}
 \includegraphics[height=5cm]{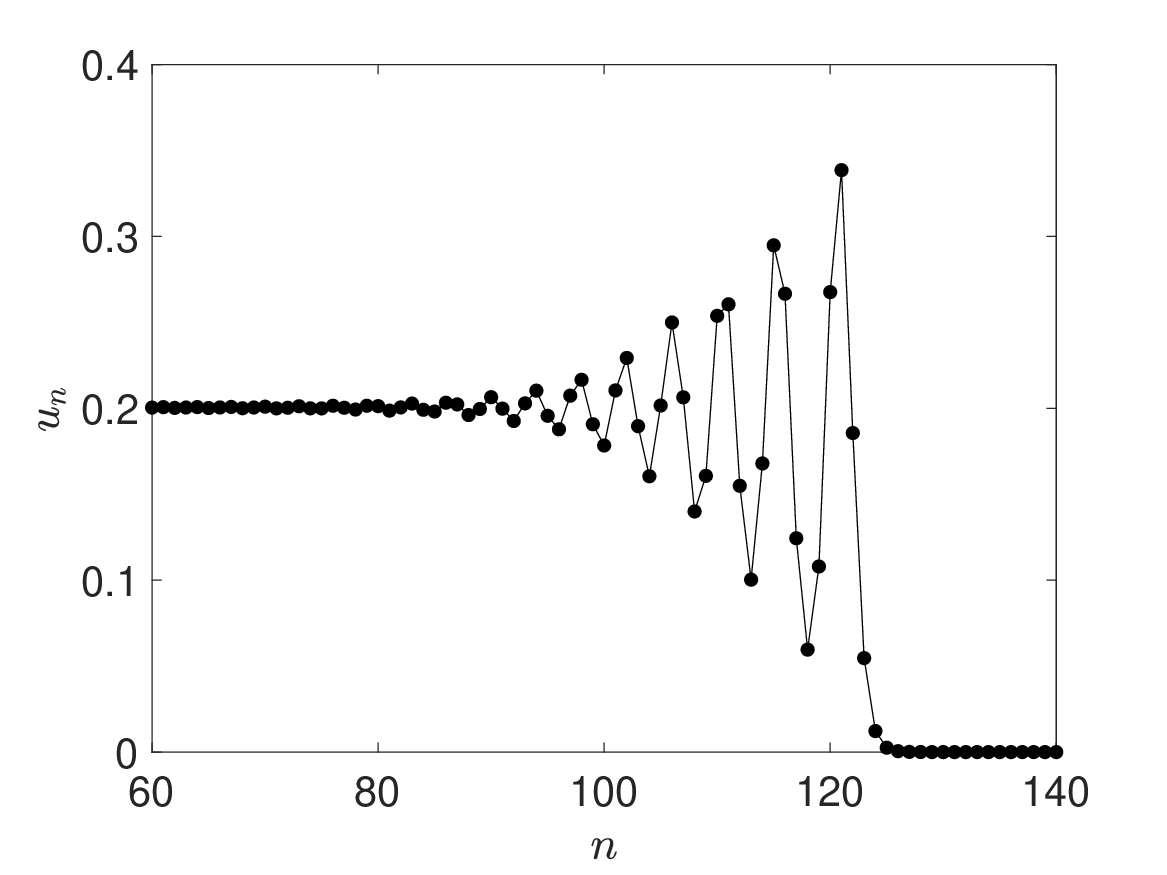} 
  \end{tabular}
  }
      \centerline{
   \begin{tabular}{@{}p{0.42\linewidth}@{}p{0.42\linewidth}@{}}
     \rlap{\hspace*{5pt}\raisebox{\dimexpr\ht1-.1\baselineskip}{\bf (c)}}
 \includegraphics[height=5cm]{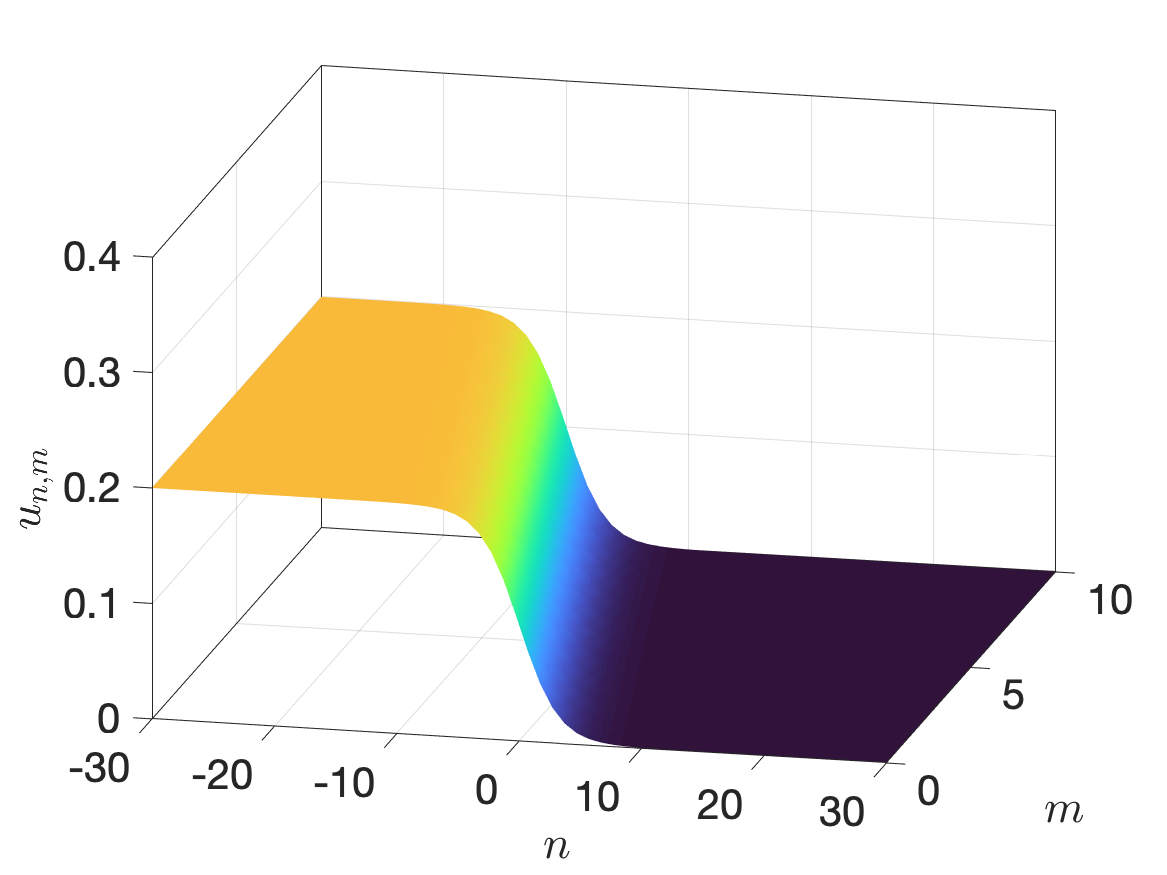}  &
   \rlap{\hspace*{5pt}\raisebox{\dimexpr\ht1-.1\baselineskip}{\bf (d)}}
 \includegraphics[height=5cm]{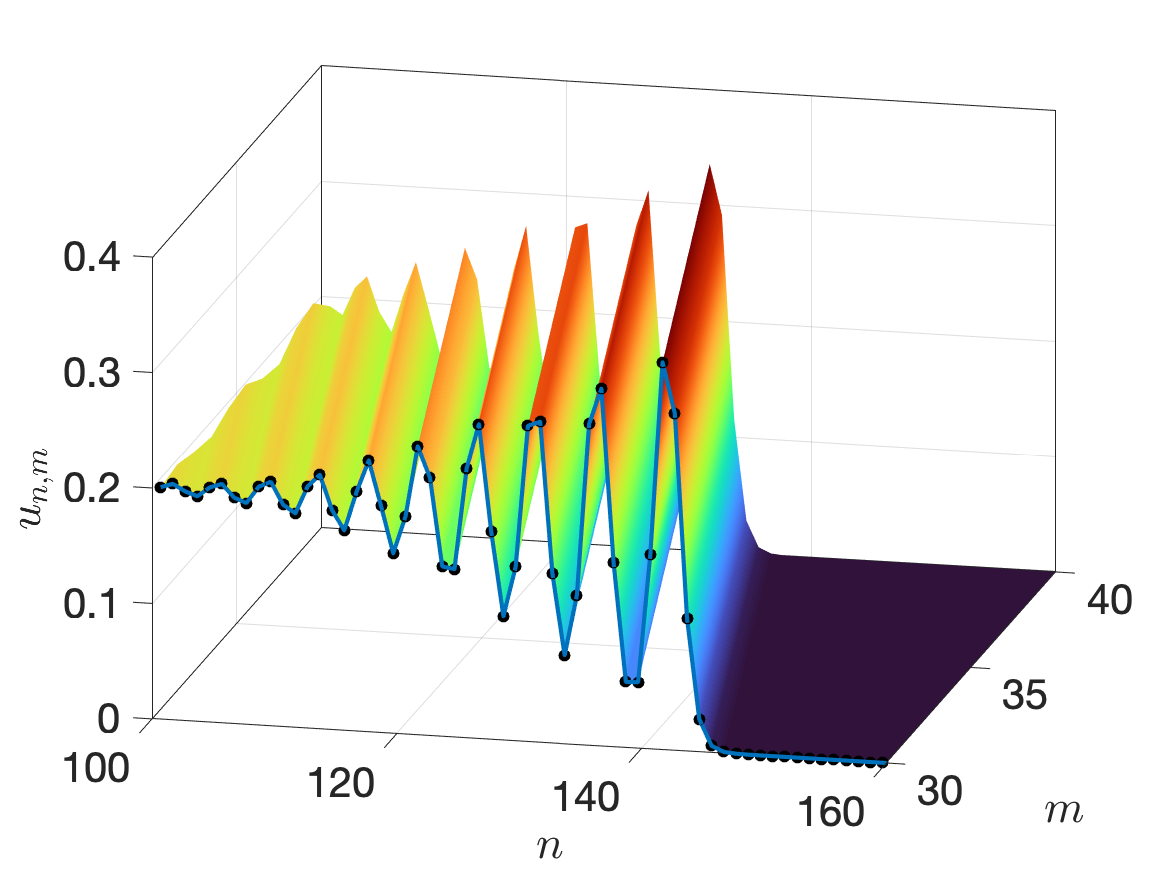} 
  \end{tabular}
  }
   \caption{\textbf{(a)} Possible initial condition that leads to the formation of
   a DSW in the 1D lattice with jump height $\delta=0.2$. The example shown is a smoothed variant of the Riemann data in Eq.~\eqref{eq:step1}. Such smoothing is necessary for asymptotic approximations, to be described later in the text. \textbf{(b)}
   The lattice DSW that forms from the initial data given in panel (a). 
\textbf{(c)} Possible initial condition that leads to the formation of
   a DSW in the 2D lattice with jump height $\delta=0.2$. Here the slope of the wavevector is $k_2/k_1=1$. The example shown is a smoothed variant of the Riemann data in Eq.~\eqref{eq:step2}. \textbf{(d)}
   The lattice DSW that forms from the initial data given in panel (c).
   The solution propagates along the lattice in the direction given by
   $(k_1,k_2)$. 
   }  
    \label{fig:DSWexamples}
\end{figure}

In the present paper, our focus is the study of traveling and dispersive shock
waves ${  u:\Z^2\times\R\to\R}$ in a scalar two-dimensional (2D) analog of
the FPUT model, namely
\begin{equation} \label{eq:model}
    \ddot{u}_{n,m} = \Delta_{d} \Phi'(u), \qquad 
    (n,m) \in \Z^2,
\end{equation}
where $\Delta_{d} u = u_{n+1,m}+u_{n-1,m}+u_{n,m+1}+u_{n,m-1} - 4 u_{n,m} $.
Besides being a model for a 2D (conservative variant of a) transmission line \cite{Butt_2006},
Eq.~\eqref{eq:model} has the additional benefit of having a Hamiltonian structure. 
The dispersion
relationship for Eq.~\eqref{eq:model} is obtained
by substituting the plane wave $u_{n,m}(t)=e^{i(k_1n + k_2m - \omega t)}$ into the linearized equations of motion, leading to
\begin{equation} \label{eq:disp}
   \omega^2(\mathbf{k}) = 4a_1( \sin^2(k_1/2) + \sin^2(k_2/2)) 
\end{equation}
where once again $\omega$ is the frequency, $a_1$ is the linear coefficient of $\Phi'$
and now the single wavenumber has been replaced by a wave vector $\mathbf{k} = (k_1,k_2)$.
It will be convenient to define $k_1 = r \cos(\theta)$, $k_2 = r \sin(\theta)$ 
such that $r^2=k_1^2 + k_2^2$. Then, fixing $\theta$, we have from Eq.~\eqref{eq:disp} that
$$ \lim_{r \rightarrow 0} \omega'(r) = \sqrt{a_1} $$
such that the sound speed is $c_s = \sqrt{a_1}$. Similar to the 1D model,
the linear situation is fairly straightforward. The presence of nonlinearity allows
for more interesting dynamics. For example, traveling waves of Eq.~\eqref{eq:model} have the form
\begin{equation} \label{eq:tw}
    u_{n,m}(t) = u(\rho) = U(\rho) + \bar{u}
\end{equation}
where the phase variable is now $\rho = k_1 n + k_2 m - \omega t$.
For periodic traveling waves we have
$U(\rho) = U(\rho + 2\pi)$ and $(k_1,k_2) \in (0,\pi) \times(0,\pi)$.
Solitary waves are obtained by replacing the $2\pi$ periodicity 
property with $\displaystyle\lim_{|\rho|\rightarrow \infty} U(\rho) = 0$, where
$\omega/r$ represents the solitary wave speed. It is assumed that
$$\bar{U}={  \frac{1}{2\pi}}\int_0^{2\pi} U(\rho) d\rho= 0$$
such that $\bar{u}$ represents the mean of the wave. Upon substitution of Eq.~\eqref{eq:tw} into Eq.~\eqref{eq:model}
we find that
\begin{equation} \label{eq:advdelay}
    \omega^2 \frac{d^2 U}{d\rho^2} = \Delta_\mathbf{k} \Phi'(\bar{u}+U)
\end{equation}
where 
$$  \Delta_\mathbf{k} U{  (\rho)} = U(\rho + k_1) + U(\rho - k_1) + U(\rho + k_2) + U(\rho - k_2) -4 U(\rho). $$
Upon obtaining such a wave profile $U$, one returns to Eq.~\eqref{eq:tw} to obtain the 2D lattice
solution. Fig.~\ref{fig:TWexamples}(a) shows an example of a traveling wave profile $U(\rho)$ and its
2D representation in Fig.~\ref{fig:TWexamples}(b).  By construction, the 2D solution
(shown in (b)) will be constant along vectors orthogonal to the wavevector $(k_1,k_2)$.
Hence, if one were to cut the lattice along the wavevector $(k_1,k_2)$
the resulting profile would look like a 1D traveling wave. For this reason,
the solutions can be considered line-solutions.
Such traveling waves will be the first main focus of the paper, see Sec.~\ref{sec:TWs}.

With regards to the second focal point of this paper, namely DSWs, a natural starting point would be to consider the 2D analog of the initial data given in Eq.~\eqref{eq:step1}, such as  
\begin{equation} \label{eq:step2}
 u_{n,m}(0) = \begin{cases}
    \delta{\sqrt{a_1}}/{a_2}, \quad & k_1 n + k_2 m \leq 0 \\
    0, \quad & k_1 n + k_2 m > 0,
\end{cases}
\end{equation}
see Fig.~\ref{fig:DSWexamples}(c) for example. Here the slope of the wavevector $(k_1,k_2)$ defines
a natural parameter that could, in principle, yield DSWs with properties depending on that slope. One may expect  the DSW to propagate in the direction of $(k_1,k_2)$,
(namely along the so-called observation direction). A simulation of the 2D lattice
is shown in Fig.~\ref{fig:DSWexamples}(d), which is an example of a 2D lattice DSW. If one
were to slice this 2D DSW along the observation direction, the resulting spatial profile
would resemble the 1D DSW shown in panel (b). These  ``line~DSWs"  are arguably the most
natural starting point for the study of lattice 2D DSWs,  given their
quasi-1D nature and will be examined in
section~\ref{sec:DSWs}.

\section{Traveling Waves} \label{sec:TWs}

We first give an existence proof of traveling waves for two cases.
In one, convexity of the potential $\Phi$ is assumed and in the
other the convexity assumption is lifted for the case of unimodal waves. The key point is to rewrite Eq.~\eqref{eq:advdelay} as a fixed point problem which has an additional variational structure. The variational structure allows us to solve it via constraint maximization where the frequency $\omega$ is found as a Lagrange parameter. The fixed point formulation is the basis for our subsequent discussion of a numerical algorithm to solve Eq.~\eqref{eq:advdelay}. Finally, using the KdV approximation of solitary and periodic waves, we demonstrate that they compare favorably to the numerically obtained ones.

\medskip

{
In 1D settings, variational approaches have been successfully used both for periodic and solitary waves. For FPUT chains with displacements as unknowns, monotone waves were found in \cite{Wattis} by constraint minimization and in \cite{  PankovFPU, Smets_Willem, Schwetlick_Zimmer} by use of the mountain pass method. Relative displacements as unknowns opened another way to existence using a fixed-point formulation and constraint maximization, cf. \cite{Venakides99, pego1, Herrmann10b, herrmann_peridynamics}, from which also relevant numerical algorithms were derived, including in the far
earlier work of~\cite{Hochstrasse}. For 2D settings, we are only aware of non-variational methods in the slightly supersonic regime \cite{Friesecke2003}, formal approximations \cite{Wattis94, Butt_2006}  or variational methods for chains of oscillators with only local forces \cite{Bak2011,Feckan2007}. In our approach for traveling waves in 2D we first find the suitable fixed point formulation. Similar to \cite{Venakides99, Herrmann10b, herrmann_peridynamics} we then solve the fixed point equation by constraint maximization. Unlike for periodic waves, maximization in the context of solitary waves is obstructed by a lack of weak continuity of the functional. In our approach this is overcome by using symmetrically decreasing rearrangements and the general Riesz rearrangement inequality as a new tool, cf. Lemma~\ref{lem:rearrange} below for details. While methodologically different, our main result for unimodal solitary waves (cf. Theorem~\ref{thm:ex2} below) compares with that in the 1D setting from \cite{Herrmann10b}. In particular the issue of strict superquadraticity of the potential $\Phi$ is overcome similarly as in \cite{Herrmann10b,herrmann_peridynamics}.
}

\subsection{Existence of periodic waves} \label{sec:periodic_waves}

{  Clearly $U=0$ trivially solves \eqref{eq:advdelay}. Our main result in this section is the following existence result for non-trivial solutions with non-vanishing wave speed.

\begin{theorem} \label{thm:ex1} Suppose that $\Phi:\R\to\R$ is strictly convex. For every fixed $R>0$ and $\bar{u}\in \R$ there is a non-trivial $2\pi$-periodic traveling wave $u_{n,m}(t)=U(k_1 n + k_2m-\omega t)+\bar u$ of Eq.~\eqref{eq:model} where the profile $U$ satisfies Eq.~\eqref{eq:advdelay} and $\|U\|_\infty \leq \sqrt{\max\{k_1,k_2\}}R, \|U\|_{L^2}\leq \max\{k_1,k_2\}R$. The a-priori unknown parameter $\omega^2 >0$ occurs as a Lagrange multiplier where $\omega$ is the frequency. 
\end{theorem}
}

{  The proof of Theorem~\ref{thm:ex1} is a consequence of Lemma~\ref{lem:ex} (for the existence) and Lemma~\ref{lem:prop_op}(iii) (for the estimates on $U$). We first formulate Eq.~\eqref{eq:advdelay} as a fixed point problem. This will be the source both for the rigorous establishment of traveling waves and for the numerical algorithm which provides approximations of them.} We begin with the observation that 
$$
\Delta_\mathbf{k} = -D^\ast D
$$
with the operator 
$$
Du(\rho) = \begin{pmatrix} u(\rho+\frac{k_1}{2}) - u(\rho-\frac{k_1}{2}) \\ u(\rho+\frac{k_2}{2}) - u(\rho-\frac{k_2}{2})
\end{pmatrix}
$$
and its formal $L^2$-adjoint 
$$
D^\ast w(\rho) = -\Bigl(w_1(\rho+\frac{k_1}{2}) - w_1(\rho-\frac{k_1}{2}) + w_2(\rho+\frac{k_2}{2}) - w_2(\rho-\frac{k_2}{2})\Bigr), 
$$
where $u:\R\to\R$ is scalar-valued and $w:\R\to\R^2$ is vector-valued. {  One can think of $D, D^\ast$ as discrete versions of grad and div. Note that $\langle Du,w\rangle_{L^2([0,2\pi]; \R^2)} =\langle u,D^\ast w\rangle_{L^2([0,2\pi];\R)}$ for $u\in L^2([0,2\pi];\R)$ and $w\in L^2([0,2\pi]; \R^2)$ and that $(D^\ast)^\ast=D$.} If we assume that $U=-D^\ast \partial_\rho^{-1} z$ for a {  $2\pi$-}periodic, mean-zero function $z\in L^2([0,2\pi]; \R^2)$, then \eqref{eq:advdelay} can be written as 
$$
\omega^2 D^\ast \partial_\rho z = D^\ast D \Phi'(\bar u- D^\ast \partial_\rho^{-1} z).
$$
Removing $D^\ast$ and inverting $\partial_\rho$ yields the fixed point problem
\begin{equation} \label{eq:fixed_point}
z  =\frac{1}{\omega^2} \partial_\rho^{-1} D\Phi'(\bar u -D^\ast \partial_\rho^{-1} z)
\end{equation}
for the vector-valued mean-free {  $2\pi$-periodic} function $z:\R\to \R^2$. 
{  We note that $\partial_\rho^{-1}$ can be understood as Fourier-multiplier $-\frac{i}{k}$ which is defined for mean-free $2\pi$-periodic functions.} 
{  Since} $D\Phi'(\ldots)$ is mean-free by the definition of $D$, {  Eq.~\eqref{eq:fixed_point} is a well-defined fixed point problem in the space of mean-free $2\pi$-periodic functions.}  We can extend this approach to functions $z$ with non-zero mean by defining the projection $Pz := z - \bar z$ with $\bar z := \frac{1}{2\pi}\int_0^{2\pi} z(s)\,ds$ for an arbitrary function $z\in L^2([0,2\pi])$. Then we can seek solutions $z\in L^2([0,2\pi])$ of 
\begin{equation} \label{eq:fixed_point2}
z  =\frac{1}{\omega^2} \partial_\rho^{-1} D\Phi'(\bar u -D^\ast \partial_\rho^{-1}P z)
\end{equation}
{  and since the right-hand side is already mean-free, any solution $z$ of Eq.\eqref{eq:fixed_point2} will be mean-free and thus it will solve the original fixed-point problem Eq.~\eqref{eq:fixed_point}.}

\medskip

Next we determine the precise form of $D^\ast\partial_\rho^{-1}P$ and $\partial_\rho^{-1} D$. With $L^2_0([0,2\pi])$ we denote the subspace of $L^2([0,2\pi])$ of functions with mean zero. 

\begin{lemma} \label{lem:prop_op} The operator $D^\ast \partial_\rho^{-1}P: L^2([0,2\pi])\times L^2([0,2\pi])\to L^2([0,2\pi])$ is bounded, compact and given by 
\begin{equation} \label{eq:def_D_drhoinv}
D^\ast \partial_\rho^{-1} Pz(\rho) = -\left(\int_{\rho-\frac{k_1}{2}}^{\rho+\frac{k_1}{2}} z_1(s)\,ds -k_1 \bar z_1+ \int_{\rho-\frac{k_2}{2}}^{\rho+\frac{k_2}{2}} z_2(s)\,ds -k_2 \bar z_2\right).
\end{equation}
The operator $\partial_\rho^{-1} D: L^2([0,2\pi])\to L^2([0,2\pi))\times L^2([0,2\pi])$ is bounded, compact and given by \begin{equation}
\label{eq:def_drhoinv_D}
\partial_\rho^{-1} D u(\rho) = \begin{pmatrix} \int_{\rho-\frac{k_1}{2}}^{\rho+\frac{k_1}{2}} u(s)\,ds-k_1 \bar u  \\
\int_{\rho-\frac{k_2}{2}}^{\rho+\frac{k_2}{2}} u(s)\,ds -k_2 \bar u 
\end{pmatrix}.
\end{equation}
With $k=\max\{k_1,k_2\}$ the operators have the following properties:
\begin{itemize}
\item[(i)] $\range \partial_\rho^{-1}D \subset L^2_0([0,2\pi])\times L^2_0([0,2\pi])$ and $\range D^\ast \partial_\rho^{-1}P \subset L^2_0([0,2\pi])$.
\item[(ii)] $(D^\ast\partial_\rho^{-1}P)^\ast = -\partial_\rho^{-1}D$ on $L^2([0,2\pi])$ and $(D^\ast\partial_\rho^{-1})^\ast = -\partial_\rho^{-1}D$ on $L^2_0([0,2\pi])$. 
\item[(iii)] $\|D^\ast\partial_\rho^{-1}P z\|_2 \leq k\|z\|_2$, $\|D^\ast\partial_\rho^{-1}P z\|_\infty \leq \sqrt{k} \|z\|_2$ for all $z\in L^2([0,2\pi])\times L^2([0,2\pi])$.
\item[(iv)] $\|\partial_\rho^{-1} Du\|_2 \leq k\|u\|_2$, $\|\partial_\rho^{-1} Du\|_\infty \leq \sqrt{k} \|u\|_2$ for all $u\in L^2([0,2\pi])$. 
\end{itemize}
\end{lemma}

\begin{proof}
\textit{Verification of Eq.~\eqref{eq:def_D_drhoinv} and Eq.~\eqref{eq:def_drhoinv_D}:} Since $D^\ast$ eliminates integration constants we obtain 
$$
D^\ast \partial_\rho^{-1} Pz(\rho) = D^\ast \int_0^\rho \begin{pmatrix} z_1(s)-\bar z_1 \\ z_2(s)-\bar z_2 \end{pmatrix} \,ds = -\left(\int_{\rho-\frac{k_1}{2}}^{\rho+\frac{k_1}{2}} z_1(s)\,ds -k_1 \bar z_1+ \int_{\rho-\frac{k_2}{2}}^{\rho+\frac{k_2}{2}} z_2(s)\,ds -k_2 \bar z_2\right)
$$
and hence we have verified Eq.~\eqref{eq:def_D_drhoinv}. Next we compute $\partial_\rho^{-1} u$ for a mean-zero function $u$. Obviously, $\partial_\rho^{-1}u(\rho) = \int_0^\rho u(s)\,ds - c$ where the constant $c$ needs to satisfy 
$$
c = \frac{1}{2\pi}\int_0^{2\pi} \int_0^\rho u(s)\,ds\,d\rho = \frac{1}{2\pi}\int_0^{2\pi} \int_s^{2\pi} u(s)\,d\rho \,ds = \frac{1}{2\pi}\int_0^{2\pi}(2\pi-s)u(s)\,ds = -\frac{1}{2\pi}\int_0^{2\pi}su(s)\,ds.
$$
Thus, 
$$
\partial_\rho^{-1} u(\rho) = \int_0^\rho u(s)\,ds+\frac{1}{2\pi}\int_0^{2\pi}su(s)\,ds. 
$$
Now we need to apply $\partial_\rho^{-1}$ to the two components of $Du(\rho)$. It is enough to do this for the first component. Since $\partial_\rho^{-1}$ consists of two terms, let's start with the first:
$$
\int_0^\rho u(s+\frac{k_1}{2}) - u(s-\frac{k_1}{2}) \,ds = \int_{\frac{k_1}{2}}^{\rho+\frac{k_1}{2}} u(s)\,ds - \int_{-\frac{k_1}{2}}^{\rho-\frac{k_1}{2}} u(s)\,ds = \int_{\rho-\frac{k_1}{2}}^{\rho+\frac{k_1}{2}} u(s)\,ds - \int_{-\frac{k_1}{2}}^\frac{k_1}{2} u(s)\,ds.
$$
Next let us check the effect of the second term in $\partial_\rho^{-1}$ applied to $Du$. This can be seen from 
\begin{align*}
\int_0^{2\pi} s\left(u(s+\frac{k_1}{2})-u(s-\frac{k_1}{2})\right)\,ds &= \int_{\frac{k_1}{2}}^{2\pi+\frac{k_1}{2}} (s-\frac{k_1}{2}) u(s)\,ds - \int_{-\frac{k_1}{2}}^{2\pi-\frac{k_1}{2}} (s+\frac{k_1}{2}) u(s)\,ds \\
&= -\int_{-\frac{k_1}{2}}^\frac{k_1}{2} su(s)\,ds + \int_{2\pi-\frac{k_1}{2}}^{2\pi+\frac{k_1}{2}} su(s)\,ds -k_1 \int_0^{2\pi} u(s)\,ds\\
& = -\int_{-\frac{k_1}{2}}^\frac{k_1}{2} su(s)\,ds + \int_{-\frac{k_1}{2}}^{\frac{k_1}{2}} (s+2\pi) u(s)\,ds -k_1 \int_0^{2\pi} u(s)\,ds\\
& =  2\pi\int_{-\frac{k_1}{2}}^{\frac{k_1}{2}} u(s)\,ds -k_1 \int_0^{2\pi} u(s)\,ds.
\end{align*}
Adding this (divided by $2\pi$) to the first term of $\partial_\rho^{-1}$ applied to $Du(\rho)$ we get the result of Eq.~\eqref{eq:def_drhoinv_D}. {  Boundedness and compactness properties of both operators will be proven below.}

\medskip

{ 
\noindent
(i): The range of both operators in Eq.~\eqref{eq:def_D_drhoinv} and Eq.~\eqref{eq:def_drhoinv_D} consists of mean-free functions. 

\medskip

\noindent
(ii): We have already seen that $D^\ast$ is the $L^2$-adjoint of $D$ and that $(D^\ast)^\ast=D$. Moreover, seeing $\partial_\rho^{-1}$ as the Fourier-multiplier $-\frac{i}{k}$ directly explains $(\partial_\rho^{-1})^\ast=-\partial_\rho^{-1}$.

\medskip

\noindent
(iii) \textit{boundedness and compactness:} the $L^2$-$L^2$ and $L^\infty$-$L^2$ bounds for the two operators $D^\ast\partial_\rho^{-1}P$ and $\partial_\rho^{-1}D$ are very similar. For a $2\pi$-periodic function $u\in L^2([0,2\pi])$ we see that 
\begin{align*}
\left\|\int_{\cdot+a}^{\cdot+b} u\,ds\right\|_{L^2}^2 &= \int_0^{2\pi}\left(\int_{\rho+a}^{\rho+b} u(s)\,ds\right)^2\,d\rho \leq \int_0^{2\pi} (b-a)\int_{\rho+a}^{\rho+b} u^2(s)\,ds \,d\rho \\
& = \int_0^{2\pi} (b-a) \int_{s+a}^{s+b} u^2(s)\,d\rho\,ds = (b-a)^2 \|u\|_{L^2}^2.
\end{align*}
Moreover, if we insert $Pu=u-\bar u$ into the above estimate and use that $u\mapsto \bar u$ is an orthogonal projection we get 
$$
\left\|\int_{\cdot+a}^{\cdot+b} (u-\bar u)\,ds\right\|_{L^2} \leq (b-a) \|u-\bar u\|_{L^2} \leq (b-a) \|u\|_{L^2}.
$$
This explains the $L^2$-$L^2$ estimates for both operators.  The $L^2$-$L^\infty$ bounds follow from 
$$
\left\|\int_{\cdot+a}^{\cdot+b} (u-\bar u)\,ds\right\|_{L^\infty}  \leq \sqrt{b-a} \|u-\bar u\|_{L^2} \leq \sqrt{b-a} \|u\|_{L^2}.
$$
It remains to show the compactness of $D^\ast\partial_\rho^{-1}$ and $\partial_\rho^{-1}D$ in the $L^2$-$L^2$ setting. We apply the Fr\'{e}chet-Kolmogorov-Riesz theorem, cf. \cite{adams_fournier}, on the bounded domain $[0,2\pi]$ to show the compactness of the sequence $(Tu_j)_{j\in\N}$ for a bounded sequence $(u_j)_{j\in\N}$ in $L^2([0,2\pi])$. Here $(Tu)(\rho)= \int_{\rho+a}^{\rho+b} u(s)\,ds$ and $u\in L^2([0,2\pi])$ is $2\pi$-periodically extended to $\R$. The equi-integrability of $(Tu_j)_{j\in\N}$ for $h\in \R$ follows from the estimate 
\begin{align}
    \| Tu(\cdot+h)-Tu\|_{L^2}^2 &= \left\| \int_{\cdot+a+h}^{\cdot+b+h} u(s)\,ds - \int_{\cdot+a}^{\cdot+b} u(s)\,ds\right\|_{L^2}^2 \\
    &= \left\| -\int_{\cdot+a}^{\cdot+a+h} u(s)\,ds + \int_{\cdot+b}^{\cdot+b+h} u(s)\,ds\right\|_{L^2}^2 \\
    &\leq  2\int_0^{2\pi}\left(\left( \int_{\rho+a}^{\rho+a+h} u(s)\,ds\right)^2 + \left( \int_{\rho+b}^{\rho+b+h} u(s)\,ds\right)^2\right) \,d\rho \\
    & \leq 8\pi h \|u\|_{L^2}^2.
\qedhere
\end{align}
}
\end{proof}

The variational structure of Eq.~\eqref{eq:fixed_point2} becomes visible if we re-write it as 
\begin{equation} \label{eq:variational}
z = \frac{-1}{\omega^2} (\partial_\rho^{-1})^\ast D \Phi'(\bar u - D^\ast \partial_\rho^{-1} Pz) 
\end{equation}
and 
{  recall} that $(D^\ast)^\ast=D$. 
As we shall see in the next lemma, solutions of Eq.~\eqref{eq:variational} are critical points of the functional 
$$
J(z) = \int_0^{2\pi}\Phi(\bar u - D^\ast\partial_\rho^{-1} Pz) \,d\rho
$$
on the constraint $\|z\|_{2}=R$ with $\omega^2$ being the Lagrange multiplier. The following lemma then concludes the proof of Theorem~\ref{thm:ex1} { by showing the existence of a maximizer of $J$ on an $L^2$ norm-ball of fixed radius $R$. Note that apart from a global maximizer the functional $J$ also has many other constraint critical points. For instance, by the same method one can find a $2\pi/l$-periodic mean-zero solution of \eqref{eq:variational} for any $l\in \N$ with prescribed $L^2$-norm equal to $R/\sqrt{l}$. Considered as a $2\pi$-periodic function with $L^2$-norm equal to $R$ it is thus a constraint critical point of $J$. Since this works for any $l\in \N$, infinitely many of these $2\pi/l$-periodic constraint critical points of $J$ must be different.}

\begin{lemma} \label{lem:ex} Let $R>0$ and consider the problem of maximizing the functional $z\mapsto J(z)$ over the constraint $C:= \{z\in L^2([0,2\pi])\times L^2([0,2\pi]): \|z\|_2 \leq R\}$. If $\Phi:\R\to\R$ is convex, maximizers exist and any maximizer $z$ will be a mean-zero solution of \eqref{eq:variational}. Via $U \coloneqq\bar u -D^\ast \partial_\rho^{-1} z$ we recover a solution to Eq.~\eqref{eq:advdelay} with prescribed mean $\bar u$ { and parameter $\omega^2$ as a Lagrange multiplier, where $\omega$ is the frequency. If $\Phi$ is strictly convex then $u$ is non-constant and $\omega\not =0$.}
\end{lemma}

\begin{proof}
The functional $J:L^2([0,2\pi])\times L^2([0,2\pi])\to \R$ is well-defined by Lemma~\ref{lem:prop_op}(iii) and it is standard to show its continuous Fr\'{e}chet differentiability. Moreover, the constraint $C$ is convex, closed and hence weakly closed, and $J$ is bounded on $C$ by (iv) from Lemma~\ref{lem:prop_op}. A maximizing sequence $(z_j)_{j\in\N}$ has a weakly convergent subsequence with limit $z\in C$ and by compactness of $D^\ast \partial_\rho^{-1}P$ we get that $z$ will be a maximizer. Since $J$ is non-constant and convex, such a maximizer must lie on $\partial C$. For the derivative of $J$ in direction $\varphi$ we obtain
\begin{align*}
J'(z)\varphi & = -\int_\R \Phi'(\bar u - D^\ast \partial_\rho^{-1} Pz) D^\ast \partial_\rho^{-1} P\varphi \,d\rho \\
& = -\int_\R (\partial_\rho^{-1})^\ast D \Phi'(\bar u - D^\ast\partial_\rho^{-1}Pz) P\varphi \,d\rho \\
& =\int_\R \partial_\rho^{-1} D \Phi'(\bar u - D^\ast\partial_\rho^{-1}Pz) P\varphi \,d\rho \\
& =\int_\R \partial_\rho^{-1} D \Phi'(\bar u - D^\ast\partial_\rho^{-1}Pz) \varphi \,d\rho
\end{align*}
{  where in the last step we were allowed to replace $P\varphi$ by $\varphi$ since $\partial_\rho^{-1} D \Phi'(\bar u - D^\ast\partial_\rho^{-1}Pz)$ is mean-free.} The maximizer $z$ of $J$ on $C$ with Lagrange multiplier $\omega^2$ satisfies $\omega^2 z= \partial_\rho^{-1} D \Phi'(\bar u - D^\ast\partial_\rho^{-1}Pz)$ and hence $z$ will be a mean-free function. 

\medskip

{  Next we want to rule out that $U=0$ or that $\omega=0$. First we note that by Jensen's inequality 
$$
\frac{1}{2\pi} J(\tilde z) = \frac{1}{2\pi} \int_0^{2\pi} \Phi(\bar u - D^\ast\partial_\rho^{-1} P\tilde z)\,d\rho \geq \Phi\left(\frac{1}{2\pi}\int_0^{2\pi} (\bar u - D^\ast\partial_\rho^{-1} P\tilde z) \,d\rho\right) = \Phi(\bar u)= \frac{1}{2\pi}J(0)
$$
for every $\tilde z\in C$. Therefore, if we assume for contradiction that for the maximizer $z$ we have $D^\ast\partial_\rho^{-1}z=0$ then $J(\tilde z)=J(0)$ for all $\tilde z\in C$. By strict convexity of $J$ the equality case of Jensen's inequality implies that $D^\ast\partial_\rho^{-1} P\tilde z=0$ for all $\tilde z\in C$ and hence for all $\tilde z\in L^2([0,2\pi])\times L^2([0,2\pi])$. This is a contradiction since the Fourier-transform of $D^\ast\partial_\rho^{-1} P$ given by 
$$
(\widehat{D^\ast \partial_\rho^{-1}P\tilde{z}})_j = -\frac{2}{j}\left(\sin(\frac{jk_1}{2})\widehat{\tilde z}^1_j + \sin(\frac{j k_2}{2})\widehat{\tilde{z}}^2_j\right), \quad j\in \Z
$$
shows that the linear operator $D^\ast\partial_\rho^{-1} P$ does not vanish on $L^2([0,2\pi])\times L^2([0,2\pi])$. Next let us verify that also the Lagrange multiplier $\omega^2$ does not vanish. If, for contradiction, we assume $\omega=0$ then $0=J'(z)z\geq J(z)-J(0)\geq 0$ shows that we are again in the equality case of Jensen's inequality first for the maximizer $z$ and then for all $\tilde z\in C$. The contradiction is then as above. }
\end{proof}

\subsection{Existence of unimodal periodic and solitary waves -- another way of looking at it}

In this section we are looking for solutions {  $U$} of \eqref{eq:advdelay} where $\bar u=0$ {  but where we give up the requirement of $U$ being mean-free.} 
Our setup will be such that we consider \eqref{eq:advdelay} on $I=[-\pi,\pi]$ with periodicity or on $I=\R$ with decay to $0$ at $\pm \infty$. We are looking for unimodal solutions, i.e., solutions with $U(x)\geq 0$, $U(-x)=U(x)$ and $x\to U(x)$ non-decreasing 
 for$x>0$. {  Note that unimodal $L^2$-functions have the property that $0\leq U(\rho)|\rho| \leq \int_0^{|\rho|} U(s)\,ds \leq \sqrt{|\rho|}\|U\|_2$ implies the decay estimate $0 \leq U(\rho)\leq \frac{1}{\sqrt{|\rho|}}\|U\|_2$. Our main result now is the following.}

{ 
\begin{theorem} \label{thm:ex2} Let $\Phi:\R\to\R$ be as in \eqref{eq:def_phi} with $a_1, a_3 \geq 0$, $a_2\in\R$, and either $a_1=0$, $(a_2, a_3)\not =(0,0)$ or $a_1>0, a_2\not =0$. For every fixed $R>0$ there is a non-trivial $2\pi$-periodic traveling wave $u_{n,m}(t)=U(k_1 n + k_2m-\omega t)$ of Eq.~\eqref{eq:model} where the profile $U$ satisfies Eq.~\eqref{eq:advdelay} with $\bar u=0$ and $\|U\|_\infty \leq \sqrt{\max\{k_1,k_2\}}R, \|U\|_{L^2}\leq \max\{k_1,k_2\}R$. The a-priori unknown parameter $\omega^2>0$ occurs as a Lagrange multiplier where $\omega$ is the frequency. If $a_2\geq 0$ the profile $U$ is unimodal and if $a_2\leq 0$ then $-U$ is unimodal. 
\end{theorem}
}

{  The proof of the main result follows from Lemma~\ref{lem:ex_2} below.}
First, we define two operators $A, A^\ast$ which take the role of $D^\ast \partial_\rho^{-1}, \partial_\rho^{-1} D$ {  from the previous section} and analyze their properties.

\begin{lemma} \label{lem:prop_op_2} 
{  The} two operators $A: L^2(I)\times L^2(I)\to L^2(I)$ and $B: L^2(I)\to L^2(I)\times L^2(I)$ by
\begin{align*}
A z(\rho) = -\left(\int_{\rho-\frac{k_1}{2}}^{\rho+\frac{k_1}{2}} z_1(s)\,ds + \int_{\rho-\frac{k_2}{2}}^{\rho+\frac{k_2}{2}} z_2(s)\,ds\right), \qquad B u(\rho) = \begin{pmatrix} \int_{\rho-\frac{k_1}{2}}^{\rho+\frac{k_1}{2}} u(s)\,ds \\
\int_{\rho-\frac{k_2}{2}}^{\rho+\frac{k_2}{2}} u(s)\,ds  
\end{pmatrix},
\end{align*}
are bounded with $A^\ast=-B$. If $I=[0,2\pi]$ they are {  both} compact. With $k=\max\{k_1,k_2\}$ the operators have the following properties:
\begin{itemize}
\item[(i)] 
$\partial_\rho B = D$, $\partial_\rho A=D^\ast$.
\item[(ii)] 
{  $\|Az\|_p \leq k^\frac{2+p}{2p}\|z\|_2$ for any $p\in [2,\infty].$}
\item[(iii)] 
{  $\|Bz\|_p \leq k^\frac{2+p}{2p}\|z\|_2$ for any $p\in [2,\infty].$} 
\item[(iv)] $-A$, $B$ map 
unimodal $L^2$-functions to unimodal $L^p$-functions for all $p\in [2,\infty]$. Restricted to unimodal functions, both maps are compact for $p\in [2,\infty]$ if $I=[-\pi,\pi]$ and for $p\in (2,\infty]$ if $I=\R$. 
\end{itemize}
\end{lemma}

\begin{proof}
    (i) follows directly by differentiation. (ii) and (iii) generalize the estimates from Lemma~\ref{lem:prop_op}. This can be seen from the mapping property of the map $T:u \mapsto \int_{\cdot+a}^{\cdot+b} u(s)\,ds$ for $a>b$ and any $L^2$-function $u:I\to \R$. Since $Tu= \chi_{[a,b]}\ast u$ where $\chi_{[a,b]}$ is the indicator function of the interval $[a,b]$, Young's inequality yields $\|Tu\|_p \leq \|\chi_{[a,b]}\|_q \|u\|_2=(b-a)^\frac{1}{q}\|u\|_2$ for $q= \frac{2p}{2+p} \in [1,2]$ when $p\in [2,\infty]$. To see (iv) one first checks that $T$ maps unimodal functions to unimodal functions which then implies that $-A$ maps pairs of unimodal functions to a single unimodal function and that $B$ maps a single unimodal function to a pair of unimodal functions. It remains to see the compactness which we again check for $T$ by the Fr\'{e}chet-Kolmogorov-Riesz theorem \cite{adams_fournier}. It follows from the equi-integrability estimate 
    \begin{align*}
        \|(Tu)(\cdot+h)-Tu\|_p = \|(\chi_{[a-h,b+h]}-\chi_{[a,b]})\ast u\|_p \leq \|\chi_{[a-h,b+h]}-\chi_{[a,b]}\|_q \|u\|_2 = (2h)^{1/q}\|u\|_2
    \end{align*}
    and from the property of unimodal functions that $0 \leq (Tu)(\rho) \leq \frac{1}{\sqrt{|\rho|}}\|Tu\|_2 \leq \frac{(b-a)}{\sqrt{|\rho|}} \|u\|_2$ which leads to the uniform tightness estimate for $X>0$
    \begin{align*}
        \|Tu\|_{L^p(\R\setminus [-X,X])} \leq \const \|u\|_{L^2(\R)} X^\frac{2-p}{2p}
    \end{align*}
where the constant only depends on $a,b$ and $p$.
\end{proof}

\begin{remark} The last estimate in the proof shows why in (iv) compactness does not hold if $I=\R$ and $p=2$ since in this case the decay of unimodal functions is too weak for the uniform tightness estimate to hold outside intervals $[-X,X]$.  
\end{remark}

This time we are looking for a unimodal solution $z$ of 
\begin{equation} \label{eq:variational_2}
\omega^2 z = B \Phi'(-Az).  
\end{equation}

\begin{lemma} Any solution $z \in L^2(I)$ of \eqref{eq:variational_2} leads via $U=-Az$ to a solution of \eqref{eq:advdelay} {  where $\bar u=0$}.
\end{lemma}

\begin{proof} We first apply $-A$ and then $\partial_\rho^2$ to \eqref{eq:variational_2}. If we use Lemma~\ref{lem:prop_op_2}(i) and take into account that $\partial_\rho$ commutes with both $A$ and $B$ this leads to
$$
\omega^2 \partial^2_\rho U = -\partial^2_\rho AB \Phi'(U) = -\partial_\rho D^\ast B \Phi'(U) = -D^\ast \partial_\rho B \Phi'(U)= -D^\ast D\Phi'(U)=\Delta_\mathbf{k} \Phi'(U).
$$
\end{proof}

We will obtain a solution to \eqref{eq:variational_2} as the weak limit of a suitable maximizing sequence for the functional 
$$
J(z)= \int_I \Phi(-Az)\,d\rho
$$
on the constraint 

\begin{equation} \label{eq:Rdef}
    \|z\|_2=R
\end{equation}
with $\omega^2$ as Lagrange multiplier. { As pointed out in Section~\ref{sec:periodic_waves} the functional $J$ will have many other constraint critical point different from the global maximizer.}

\begin{lemma} \label{lem:ex_2} Let $R>0$ and consider the problem of maximizing the functional $z\mapsto J(z)$ over the constraint $C:= \{z\in L^2(I)\times L^2(I): \|z\|_2 \leq R\}$. If $\Phi:\R\to\R$ is as in \eqref{eq:def_phi} with $a_1, a_3 \geq 0$, $a_2\in\R$, $(a_2,a_3)\not = (0,0)$, then a maximizer $z$ exists, which is a solution of \eqref{eq:variational_2}. Via $U = -A z$ we recover a solution to \eqref{eq:advdelay}. If $a_2\geq 0$ we find a unimodal solution $U$ and if $a_2\leq 0$ then we find a solution such that $-U$ is unimodal.
\end{lemma}

Before we prove this lemma, we need more details on the structure of the functional $J$. 

\begin{lemma} Let $\Phi:\R\to\R$ be as in \eqref{eq:def_phi} and $z{  = (z_1,z_2)}\in L^2(I)\times L^2(I)$, where $z$ is periodically extended to $\R$ if $I=[-\pi,\pi]$. Then $J(z)$ is a finite linear combination of terms of the form
$$
\int_\R \ldots \int_\R \chi_I(\rho) \prod_{i=1}^{j-l} \chi_{k_1}(\rho-s_i)z_1(s_i) \prod_{i'=1}^l \chi_{k_2}(\rho-s_{i'}')z_2(s_{i'}') \, ds_1\ldots ds_{j-l} ds_1'\ldots d s_l' d\rho
$$    
where $j\in \N$, $l\in \{0,\ldots, j\}$ and $\chi_{k_i} = \chi_{[-\frac{k_i}{2},\frac{k_i}{2}]}$ for $i=1,2$.
\end{lemma}

\begin{proof}
We see that $-Az= \chi_{k_1}\ast z_1+ \chi_{k_2}\ast z_2$, where we use that $z_1, z_2$ are periodically extended to $\R$ if $I=[-\pi,\pi]$. Since the potential $\Phi$ is a finite linear combination of monomials the claim follows since we have by the binomial theorem that for any $j\in \N$
\begin{align*}
(-Az)^j(\rho) & = \sum_{l=0}^j \binom{j}{l} (\chi_{k_1}\ast z_1)^{j-l}(\rho) (\chi_{k_2}\ast z_2)^l (\rho)\\
& = \sum_{l=0}^j \binom{j}{l}\int_\R\ldots\int_\R \prod_{i=1}^{j-l} \chi_{k_1}(\rho-s_i)z_1(s_i) \prod_{i'=1}^l \chi_{k_2}(\rho-s_{i'}')z_2(s_{i'}') \, ds_1\ldots ds_{j-l} ds_1'\ldots d s_l'.
\end{align*}
\end{proof}

In the next lemma we show that the functional $J$ increases if its argument $z$ is turned into a vector of suitable unimodal functions. The procedure is called "symmetrically decreasing rearrangement". It is defined as follows, cf. \cite[Chapter~3]{lieb_loss}: for a Lebesgue-measurable set $S\subset \R$ let $S^\ast$ be the interval $(-\frac{|S|}{2}, \frac{|S|}{2})$. For $f\in L^2(I)$ with $I=[-\pi,\pi]$ or $I=\R$ we define 
$$
f^\ast(x) = \int_0^\infty \chi_{\{|f|>t\}^\ast}(x)\,dt, \quad x\in I.
$$
Among many other features, $\|f\|_2= \|f^\ast\|_2$ holds.

\begin{lemma} \label{lem:rearrange} Let $\Phi:\R\to\R$ be as in \eqref{eq:def_phi} with $a_1, a_2, a_3\geq 0$. If $z\in L^2(I)\times L^2(I)$ and $z^\ast\in L^2(I)\times L^2(I)$ is its symmetrically decreasing rearrangement then $J(z)\leq J(z^ \ast)$. 
\end{lemma}

\begin{proof}
    According to the general Riesz rearrangement inequality, cf. \cite[Chapter~3]{lieb_loss}, we have that 
    \begin{eqnarray*}
    \lefteqn{\int_\R \ldots \int_\R \chi_I(\rho) \prod_{i=1}^{j-l} \chi_{k_1}(\rho-s_i)z_1(s_i) \prod_{i'=1}^l \chi_{k_2}(\rho-s'_{i'})z_2(s'_{i'}) \, ds_1\ldots ds_{j-l} ds_1'\ldots d s_l' d\rho} \\
    & \leq & \int_\R \ldots \int_\R \chi_I(\rho) \prod_{i=1}^{j-l} \chi_{k_1}(\rho-s_i)|z_1|(s_i) \prod_{i'=1}^l \chi_{k_2}(\rho-s'_{i'})|z_2|(s'_{i'}) \, ds_1\ldots ds_{j-l} ds_1'\ldots d s_l' d\rho \\
    & \leq & \int_\R \ldots \int_\R \chi_I^\ast(\rho) \prod_{i=1}^{j-l} \chi^\ast_{k_1}(\rho-s_i)|z_1|^\ast(s_i) \prod_{i'=1}^l \chi^\ast_{k_2}(\rho-s'_{i'})|z_2|^\ast(s'_{i'}) \, ds_1\ldots ds_{j-l} ds_1'\ldots d s_l' d\rho.
    \end{eqnarray*}
    Since the convolution kernels satisfy $\chi_I=\chi_I^\ast, \chi_{k_1}=\chi_{k_1}^\ast, \chi_{k_2}=\chi_{k_2}^\ast$ they are already symmetrically decreasing functions and since $J(z)$ is a finite sum of non-negative multiples of the above form, the claim follows.
\end{proof}

{
\begin{remark}
    The results of \cite{Satomi} and \cite{Wang_Madiman, burchard_thesis} indicate that Lemma~\ref{lem:rearrange} might also hold for general convex potentials $\Phi:[0,\infty)\to\R$ with $\Phi(0)=0$. In these references the inequality $\int_\R \Phi(w_1\ast w_2)\,d\rho \leq \int_\R \Phi(w_1^\ast\ast w_2^\ast)\,d\rho$ is established whenever $w_1, w_2\geq 0$. For our purpose we would need an extension to a sum of two convolutions inside $\Phi$. 
\end{remark}
}

\begin{proof}[Proof of Lemma~\ref{lem:ex_2}] We only consider the case $a_2\geq 0$, where we show the existence of a positive unimodal solution. For $a_2\leq 0$ we obtain a negative unimodal solution by replacing $\Phi(u)$ with $\Phi(-u)$. Since $\Phi(u)\leq Mu^2(1+ |u|^2)$ for all $u\in \R$ with a suitable $M>0$ for all $u\in\R$ we see from Lemma~\ref{lem:prop_op_2}(iii) that $J(z) \leq M k^2\|z\|_2^2 (1+k\|z\|_2^2)$. Therefore $J$ is bounded from above on the norm ball $C$ so that we can find a maximizing sequence $(z_j)_{j\in\N}$ in $C$. If we pass from $(z_j)_{j\in \N}$ to the sequence $(w_j)_{j\in\N}$ with $w_j := |z_j|^\ast$ then the new sequence is still in $C$ and by Lemma~\ref{lem:rearrange} it is also maximizing. By extracting a subsequence and denoting it again by $(w_j)_{j\in \N}$ we have that $w_j\rightharpoonup w\in C$ as $j\to \infty$. {  Since the set of unimodal $L^2$-functions is closed and convex it is also weakly closed which means} that $w$ is unimodal. 
{  Now the proof differs according to whether $I=[-\pi,\pi]$ or $I=\R$. In the former case, we can use the compactness of $A$ from Lemma~\ref{lem:prop_op_2} to see that $w$ is a maximizer of $J|_C$. In the latter case $-A$ compactly maps pairs of unimodal $L^2$-functions into $L^3(\R)$ and $L^4(\R)$ which yields $J(w_j)\stackrel{j\to\infty}{\to} J(w)$ provided $a_1=0$, i.e., if the quadratic part in $\Phi$ is absent. For the case $a_1>0$ the same result still holds provided also $a_2>0$ but its proof is more involved, cf. Appendix~\ref{sec:appendix}.}

The strict convexity of $J$ on non-negative functions implies that $w\in \partial C$. Since $J'(w)\varphi = -\int_I \Phi'(-Aw) A\varphi \,d\rho = \int_I -A^\ast \Phi'(-Aw)\varphi \,d\rho = \int_I B\Phi'(-Aw)\varphi\,d\rho$ we see that $w$ is a non-trivial, unimodal, non-negative solution of \eqref{eq:variational_2}. Hence, {  by Lemma~\ref{lem:prop_op_2}(iv),} $U:=-Aw$ is a non-trivial, unimodal, non-negative solution of \eqref{eq:advdelay}.
\end{proof}

\subsection{Numerical Computation of Traveling Waves}

{
To find periodic waves of Eq.~\eqref{eq:advdelay} numerically, we use the fixed point formulation from \eqref{eq:fixed_point}. 
It takes the explicit form
\begin{equation} \label{eq:integral2}
\begin{split}
\omega^2 z_1 &= \hat{A}_{k_1} \Phi'(\bar{u} + \hat{A}_{k_1} z_1 +  \hat{A}_{k_2} z_2   ) \\ 
\omega^2 z_2  &= \hat{A}_{k_2} \Phi'(\bar{u} + \hat{A}_{k_1} z_1 +  \hat{A}_{k_2} z_2) 
\end{split}
\end{equation}
where the operators $\hat{A}_{k_i}$, $i=1,2$ are defined by 
$$ 
(\hat{A}_{k_i} z_i)(\rho) \coloneqq \int_{\rho - k_i/2}^{\rho + k_i/2} z_i(\tilde{\rho})d\tilde{\rho} - \frac{k_i}{2\pi}\int_0^{2\pi} z_i(\rho)d\rho.
$$
In accordance with the approach in Section~\ref{sec:periodic_waves} we make use of the following amplitude type constraint
\begin{equation} \label{eq:amp_constraint}
C=\left\{(z_1,z_2)\in L^2([0,2\pi])\times L^2([0,2\pi]): \int_0^{2\pi} z_1(\rho)^2 + z_2(\rho)^2 d\rho = R^2 \right\}
\end{equation}
for some chosen value $R>0$. 

\medskip

In Section~\ref{sec:periodic_waves} this fixed-point problem was solved as a constrained maximization problem with Lagrange multiplier $\omega^2$, { and we pointed out that the fixed-point problem has many other solutions different from the constraint global maximizer.}
{ Following the idea of solving the fixed-point problem via global maximization}, one can consider the positive gradient flow 
$$
\begin{pmatrix} \dot z_1 \\ \dot z_2\end{pmatrix} = P_{T_z C} \begin{pmatrix} \hat{A}_{k_1} \Phi'(\bar{u} + \hat{A}_{k_1} z_1 +  \hat{A}_{k_2} z_2   ) \\ \hat{A}_{k_2} \Phi'(\bar{u} + \hat{A}_{k_1} z_1 +  \hat{A}_{k_2} z_2   )
\end{pmatrix}
$$
with $P_{T_z C}$ being the projection onto the tangent space of the constraint set $C$ given by 
$$
P_{T_z C} \begin{pmatrix} h_1 \\ h_2 \end{pmatrix} = \begin{pmatrix} h_1 - \frac{\langle z,h\rangle_{L^2\times L^2}}{R^2} z_1 \\ 
h_2 - \frac{\langle z,h\rangle_{L^2\times L^2}}{R^2} z_2 
\end{pmatrix}, \quad h=(h_1, h_2)\in L^2([0,2\pi])\times L^2([0,2\pi])
$$
and $\langle \cdot,\cdot \rangle_{L^2\times L^2}$ being the standard inner product on $L^2([0,2\pi])\times L^2([0,2\pi])$. { Depending on the choice of the initial condition the gradient flow might not end up in a global maximizer but rather in a local maximizer which, however, is still a solution of the fixed-point problem.} In its time-discretized form with $\Delta_t$ as time-step, the gradient flow becomes 
$$
\begin{pmatrix} w_1^{j} \\ w_2^{j} \end{pmatrix} = \begin{pmatrix}
    z_1^j \\ z_2^j
\end{pmatrix} + \Delta_t P_{T_{z^j} C}
\begin{pmatrix} \hat{A}_{k_1} \Phi'(\bar{u} + \hat{A}_{k_1} z_1^j +  \hat{A}_{k_2} z_2^j   ) \\ \hat{A}_{k_2} \Phi'(\bar{u} + \hat{A}_{k_1} z_1^j +  \hat{A}_{k_2} z_2^j   ) \end{pmatrix}, \quad f^j := \sqrt{\frac{\|w^j\|_{L^2}}{R}}, \quad \begin{pmatrix}
    z_1^{j+1} \\ z_2^{j+1}
\end{pmatrix} = \frac{1}{f^j} \begin{pmatrix} w_1^{j} \\ w_2^j \end{pmatrix}.
$$
If we set the time-step to $\Delta_t = R^2 (\sum_{i=1}^2 \langle z^j_i, \hat A_{k_i}\Phi'(\bar u +\hat A_{k_1}z^j_1+ \hat A_{k_2}z^j_2)\rangle_{L^2})^{-1}$ the above scheme takes the very simple form 
$$
\begin{pmatrix} \tilde w_1^{j} \\ \tilde w_2^{j} \end{pmatrix} = 
\begin{pmatrix} \hat{A}_{k_1} \Phi'(\bar{u} + \hat{A}_{k_1} z_1^j +  \hat{A}_{k_2} z_2^j   ) \\ \hat{A}_{k_2} \Phi'(\bar{u} + \hat{A}_{k_1} z_1^j +  \hat{A}_{k_2} z_2^j   ) \end{pmatrix}, \quad f^j := \sqrt{\frac{\|\tilde w^j\|_{L^2}}{R}}, \quad \begin{pmatrix}
    z_1^{j+1} \\ z_2^{j+1}
\end{pmatrix} = \frac{1}{f^j} \begin{pmatrix} \tilde w_1^{j} \\ \tilde w_2^j \end{pmatrix}
$$
and the frequency update is given by $\omega^{j+1} = \sqrt{f^j}$. Upon convergence of $z^j=(z_1^j,z_2^j)$ and $\omega^j$ as $j\to \infty$, we obtain a solution of Eq.~\eqref{eq:advdelay} with frequency
$\omega$ and profile $U =\hat{A}_{k_1} z_1 +  \hat{A}_{k_2} z_2 $ which,
due to the formulation
of the $\hat{A}$ operators, will have a zero mean. 
This scheme is similar for one-dimensional lattices, see \cite{DHM06,Hochstrasse,pego1} for example.
In \cite{Yasuda2016} traveling periodic waves are found by discretizing in $\rho$ and
performing Newton iterations subject to the constraint that $\bar{U} = 0$. Solitary waves can also be found by replacing the $2\pi$ periodicity 
property with $\lim_{|\rho|\rightarrow \infty} U(\rho) = 0$,
which is implemented by enforcing Dirichlet boundary conditions. In this case, $\omega/r$ represents the solitary wave speed.
Examples of solitary and periodic waves are shown in Sec.~\ref{sec:KdV},
see for example Fig.~\ref{fig:solitary} and Fig.~\ref{fig:periodic}.
}

\subsection{KdV Description of Traveling Waves} \label{sec:KdV}

\begin{figure}[t!] %
    \centerline{
   \begin{tabular}{@{}p{0.34\linewidth}@{}p{0.34\linewidth}@{}p{0.34\linewidth}@{}}
     \rlap{\hspace*{5pt}\raisebox{\dimexpr\ht1-.1\baselineskip}{\bf (a)}}
 \includegraphics[height=4.5cm]{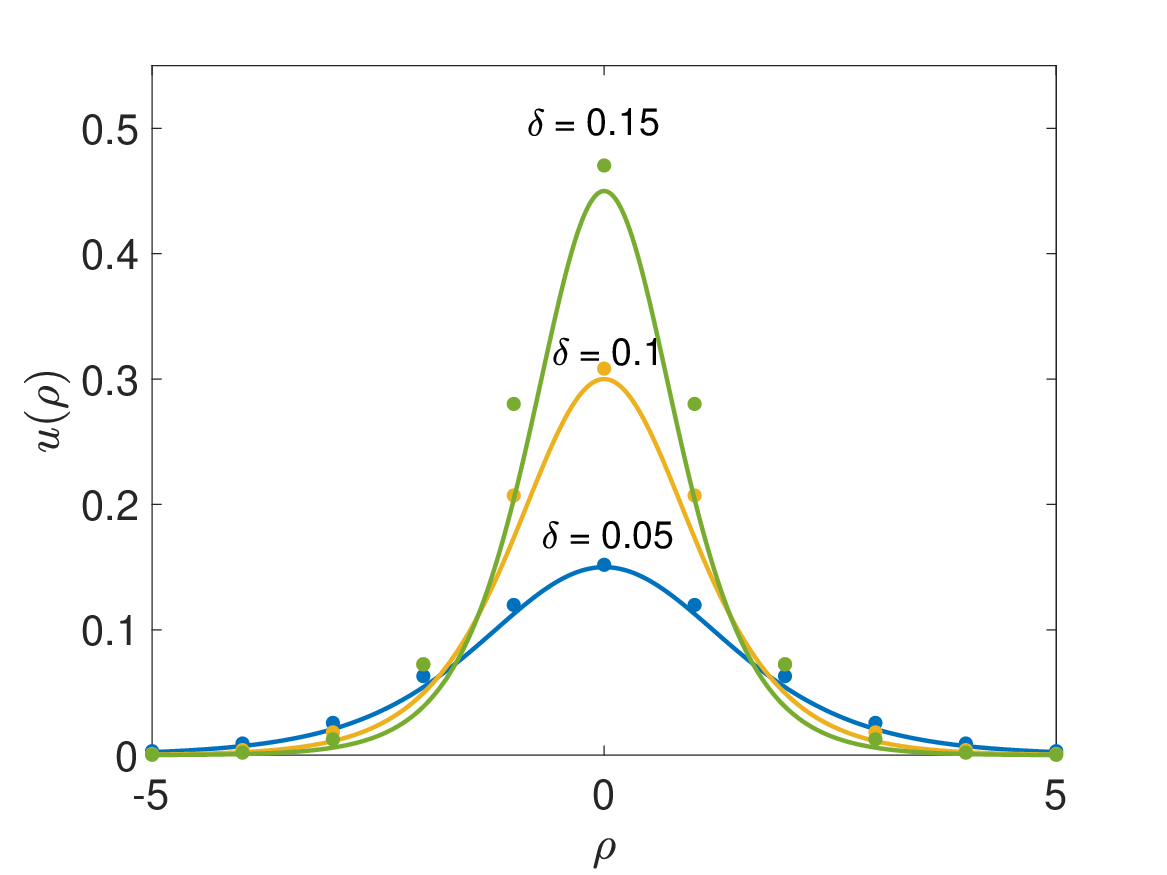}  &
   \rlap{\hspace*{5pt}\raisebox{\dimexpr\ht1-.1\baselineskip}{\bf (b)}}
 \includegraphics[height=4.5cm]{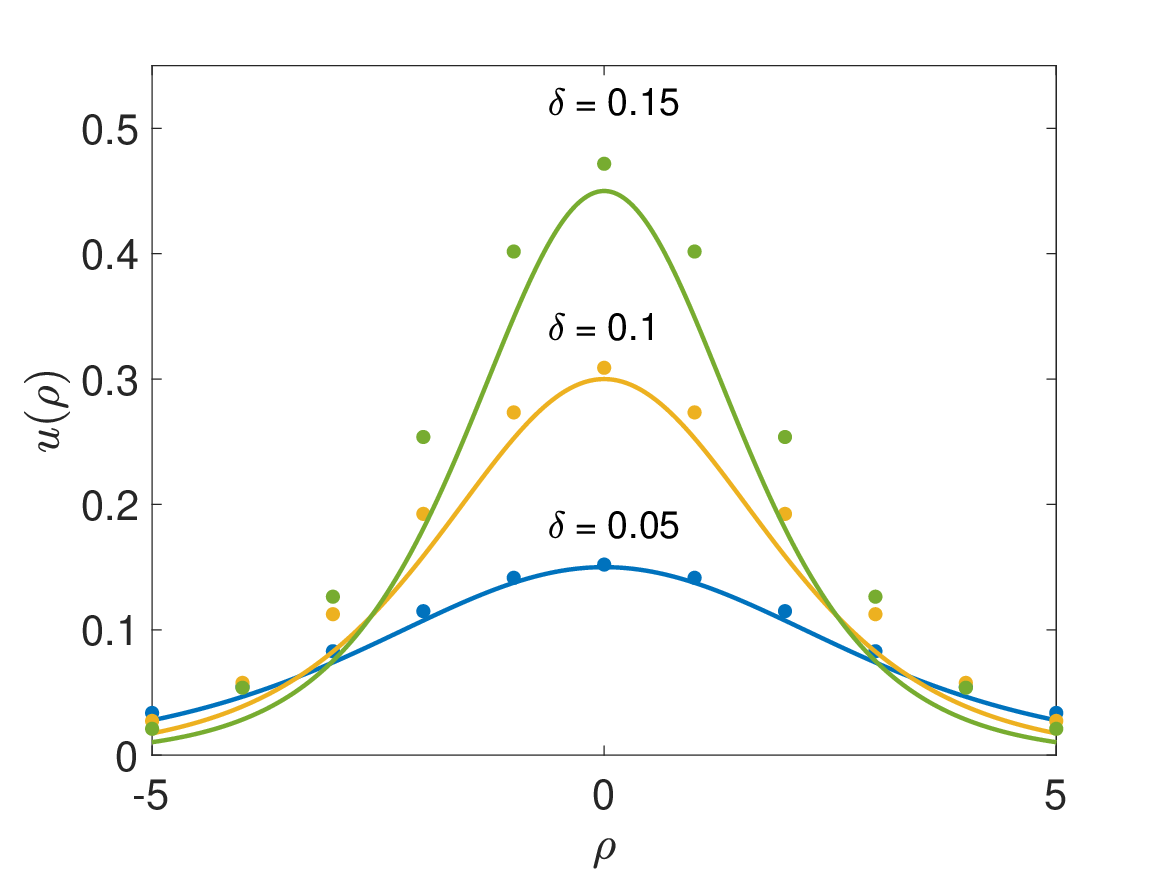} &
    \rlap{\hspace*{5pt}\raisebox{\dimexpr\ht1-.1\baselineskip}{\bf (c)}}
 \includegraphics[height=4.5cm]{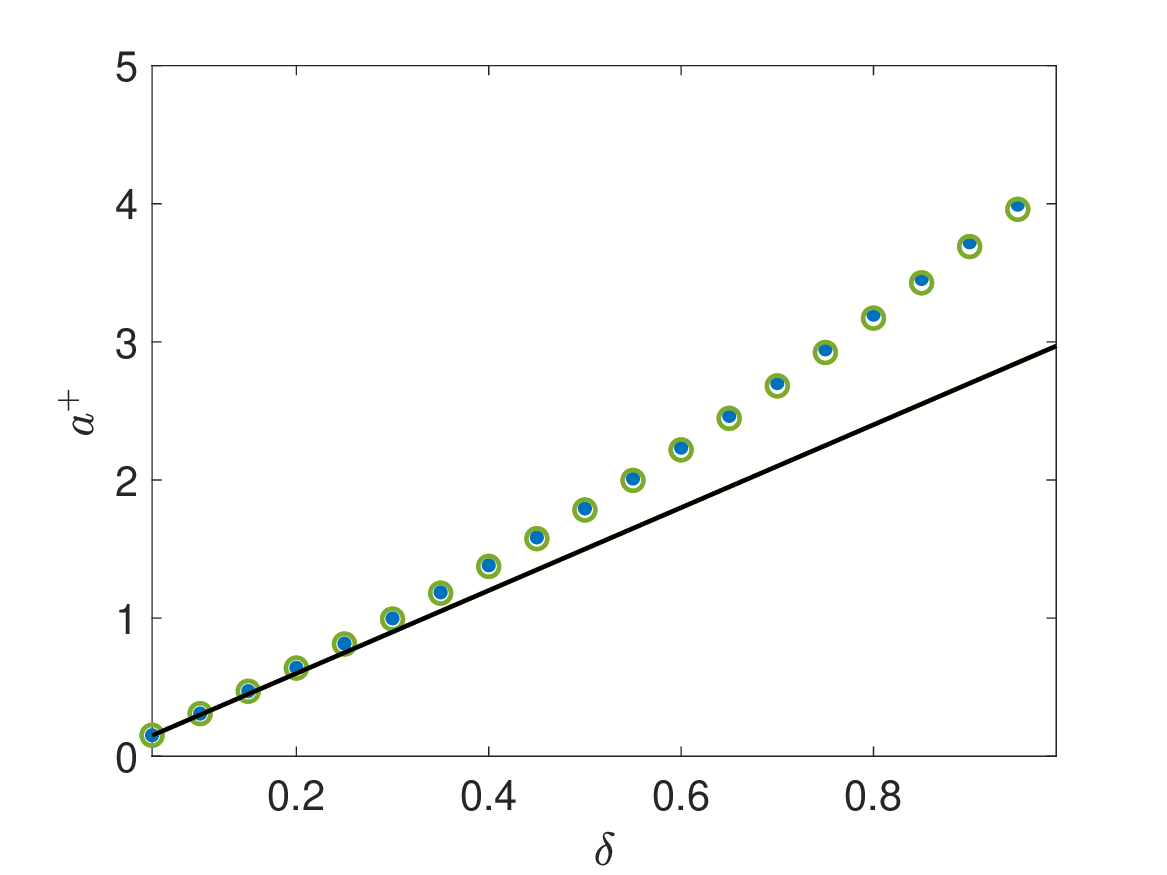} 
  \end{tabular}
  }
   \caption{Comparison of numerical lattice solitary waves and the 
   KdV predictions with $a_1 = a_2=1$ and $a_3=0$. \textbf{(a)} The wave numbers satisfy $k_2/k_1=1$.
   The solid line represents the KdV prediction and the markers
   represent the numerical lattice solitary wave. Three examples
   are shown with deviation in sound speed $\delta = 0.05$
   (bottom) $\delta = 0.1$ (middle) and $\delta = 0.15$ (top).
   \textbf{(b)} Same as panel (a) but with $k_2/k_1=2$.
   \textbf{(c)} Amplitude vs $\delta$ for $k_2/k_1=1$ (green open circles)
   and $k_2/k_1=2$ (blue filled circles)
   of the numerical solution and the KdV prediction (black line),
   which is independent of the wave-vector.
   }
\label{fig:solitary}
\end{figure}

We seek to approximate analytically the traveling waves discussed previously. We employ
the KdV scaling in a co-moving frame consistent with that of the traveling waves of the 2D lattice along a particular lattice direction defined below. In particular,
we use the ansatz 
\begin{equation} \label{eq:ansatz}
    u(n,m,t) = \epsilon^2 A(X,T), \qquad X=\epsilon( k_1n + k_2m - c r t), \quad T = \epsilon^3 t.
    \end{equation}
Here, $cr$ plays the role of the $\omega$ and $c$ is a velocity. Recall
that $k_1=r\cos(\theta), k_2 = r\sin(\theta)$.
With this ansatz, the derivation of KdV is similar to the one-dimensional
situation \cite{Zabusky}. Upon substitution of the ansatz \eqref{eq:ansatz} into Eq.~\eqref{eq:model} we obtain
\begin{align*}
\mathrm{Residual} =
-\epsilon^2&( (cr)^2 \epsilon^2 A_{XX} -2cr A_{XT} \epsilon^4   + A_{TT} \epsilon^6)     ) \\
&
{
+ \sum_{j=1}^3 a_j \epsilon^{2j}
\Big(
A(X-\epsilon k_1)^j
+
A(X-\epsilon k_2)^j
+
A(X+\epsilon k_1)^j
+
A(X+\epsilon k_2)^j
-
4A(X)^j
\Big).
}
\end{align*}
We now expand the shifted terms using the Taylor series:
\begin{align*}
A(X \pm \varepsilon k, T) &= \left[A \pm \varepsilon k A_X + \frac{(\varepsilon k)^2}{2} A_{XX} \pm \frac{(\varepsilon k)^3}{6} A_{XXX} + \frac{(\varepsilon k)^4}{24} A_{XXXX} + \cdots\right].
\end{align*}
With the goal of making the residual small, we collect terms
according to powers of $\epsilon$. Eliminating the $\mathcal{O}(\epsilon^4)$ terms yields the relation
\begin{equation}
    (rc)^2 = a_1 (k_1^2 + k_2^2) 
\end{equation}
such that $c = \sqrt{a_1}$ is the sound speed. At $\mathcal{O}(\epsilon^6)$ we obtain {  (upon one integration w.r.t. $X$)} the KdV equation,
\begin{equation} \label{eq:KdV}
0 =  A_T + \frac{\sqrt{a_1} }{24r}(k_1^4 + k_2^4) A_{XXX}+ \frac{a_2 r}{\sqrt{a_1}} A A_X 
\end{equation}
If $k_1=0$ or $k_2=0$ one recovers the KdV description for a one-dimensional FPUT lattice.

\subsubsection{Solitary waves}
The KdV equation, Eq.~\eqref{eq:KdV}, has the solitary wave solution
\begin{equation} \label{eq:KdVsoliton}
    A(X,T) = \frac{3 \sqrt{a_1} \sigma }{a_2 r } \sech^2\left( \sqrt{ \frac{6 \sigma r}{ \sqrt{a_1 }   (k_1^4 + k_2^4)   }  } (X - \sigma T)   \right)  
\end{equation} 
where $\sigma>0$ is a free parameter.
In terms of the original 2D lattice variables we have
$$u_{n,m}(t)  =  \frac{3 \sqrt{a_1} \sigma \epsilon^2 }{a_2 r } \sech^2\left( \sqrt{ \frac{6 r ^3 \sigma}{ \sqrt{a_1} (k_1^4 + k_2^4)}}\epsilon\left(    \frac{k_1 n + k_2m}{r}  - (\sqrt{a_1} + \frac{\sigma \epsilon^2}{r})t          \right )    \right). $$
The speed of the above solitary wave is 
$$s = \sqrt{a_1} + \frac{\sigma \epsilon^2}{r}. $$
The above expression naturally defines  a small parameter
$$\delta = \frac{\sigma \epsilon^2}{r} > 0$$
which is the deviation from the sound speed $c_s=\sqrt{a_1}$.
%
The solitary wave can also be expressed in terms of the small parameter $\delta$,
\begin{equation} \label{eq:lattice_solitary}
    u_{n,m}(t) = u(\rho)  =  \frac{3 \sqrt{a_1} \delta }{a_2 } \sech^2\left( \sqrt{ \frac{6 \delta}{ r^2 \sqrt{a_1} (\cos(\theta)^4 + \sin(\theta)^4)}} \,\rho            \right) 
\end{equation}
where $\rho = k_1 n + k_2m - \omega t$ with $\omega = r(\sqrt{a_1} +  \delta) $. We have the following relation for the amplitude {  $a$} of the solitary wave
in terms of the deviation from the sound speed
$$a = \frac{3 \sqrt{a_1}}{a_2} \delta.  $$
An important observation is 
that the solitary wave amplitude and speed do not depend on
the choice of wave vector, but the shape (or rather its width) does.

Figure~\ref{fig:solitary} shows comparisons between
the KdV prediction, Eq.~\eqref{eq:lattice_solitary}, 
and numerical solutions of the fixed-point problem shown in 
Eq.~\eqref{eq:integral2} with Dirichlet boundary conditions. Panel (a) shows an example
with $k_2/k_1=1$ for three different speed parameters
$\delta = 0.05,0.1$ and $0.15$. Panel (b) is similar
with $k_2/k_1 =2$. There is no noticeable difference in amplitudes 
when comparing both panels (and by construction their wave speeds
are identical for both numerical and KdV prediction), but those in panel (b) are wider, both in
the KdV prediction and for the numerical solution,
as expected.
Panel (c) shows the speed-amplitude relationship
for the KdV prediction (black line) and the numerical
solution for both the $k_2/k_1=1$ and $k_2/k_1 =2$ cases.
While the speed-amplitude relationship predicted by
the KdV equation is independent of the wavevector,
there is a small observed difference in the numerical solutions
between the $k_2/k_1=1$ and $k_2/k_1=2$ cases (since the blue dots do not lie exactly in the center of
the green circles).
{{This suggests anisotropy effects are possible
in the 2D lattice, albeit minimal for the solitary waves studied
here. Indeed, anisotropic effects in solitary waves are
known to exist in  other 2D FPUT lattices, such
as those with piecewise linear interaction terms \cite{VAINCHTEIN2018}.
While anisotropy is an interesting aspect of 2D lattice solitary wave propagation, 
it lies outside the scope of the present work. Later, for DSWs,
we do observe and study anisotropy, in particular in connection
to wavenumber dependence on the angle $\theta$.
}

\subsubsection{Periodic waves}
\begin{figure} 
    \centerline{
   \begin{tabular}{@{}p{0.45\linewidth}@{}p{0.45\linewidth}}
     \rlap{\hspace*{5pt}\raisebox{\dimexpr\ht1-.1\baselineskip}{\bf (a)}}
 \includegraphics[height=6cm]{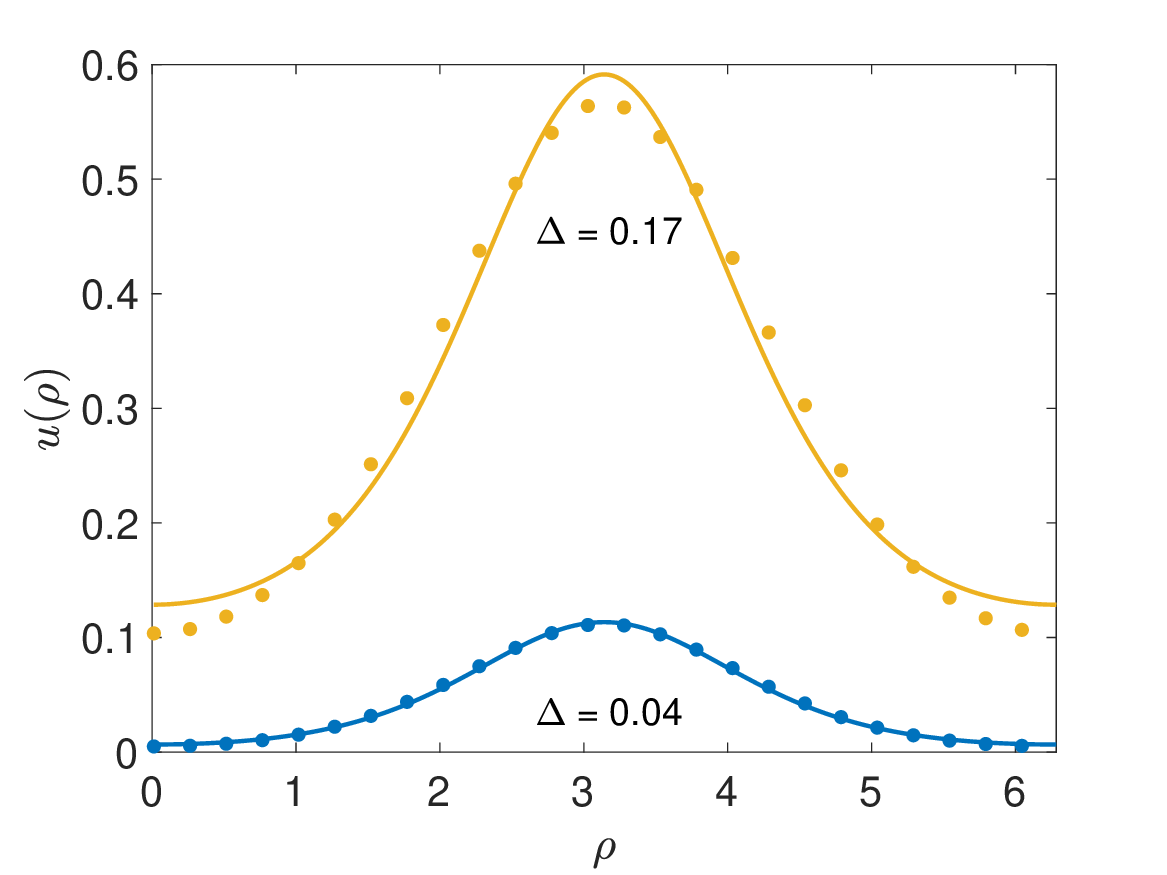}  &
   \rlap{\hspace*{5pt}\raisebox{\dimexpr\ht1-.1\baselineskip}{\bf (b)}}
 \includegraphics[height=6cm]{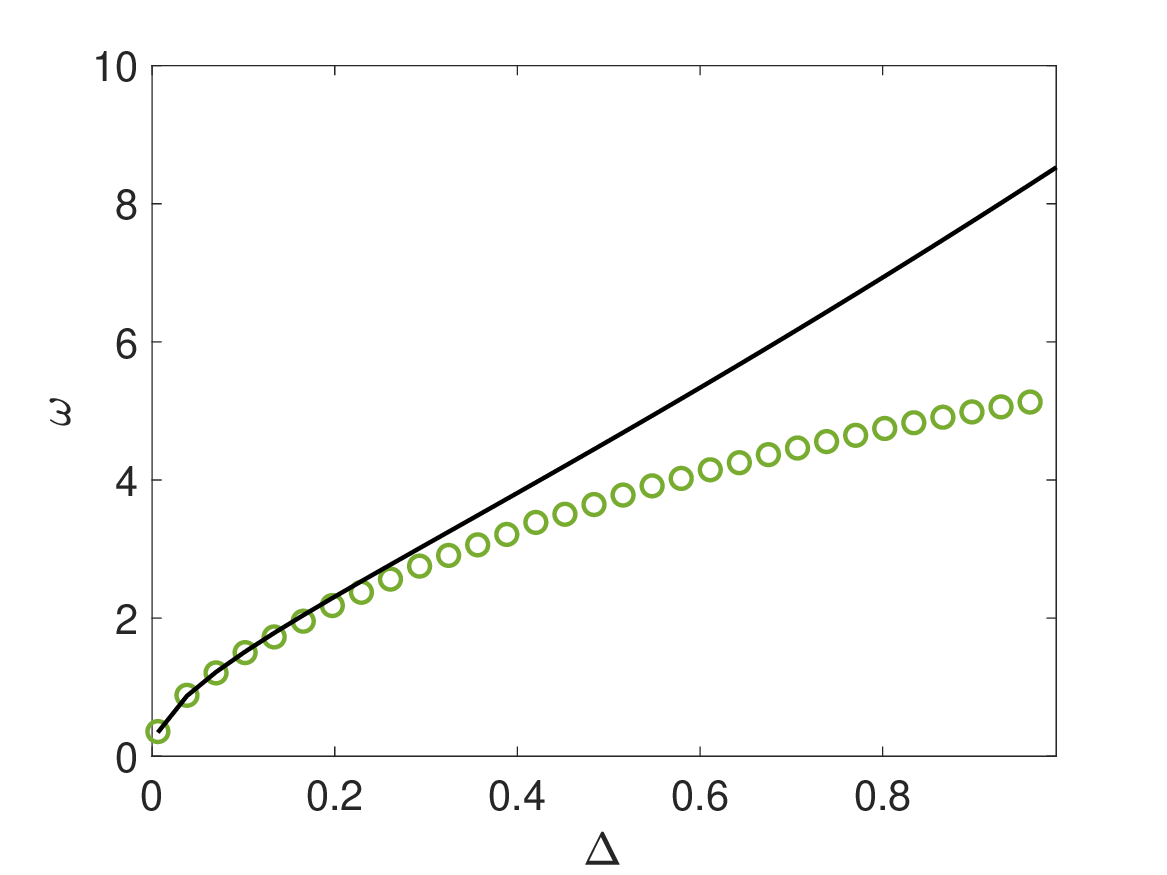} 
  \end{tabular}
  }
   \caption{Comparison of numerical lattice periodic traveling waves and the 
   KdV predictions with $a_1 = a_2=1$ and $a_3=0$. \textbf{(a)} The wave numbers are $k_1=k_2=1$.
   The solid line represents the KdV prediction and the markers
   represent the numerical lattice periodic wave. Two examples
   are shown with deviation in sound speed $\Delta = 0.04$
   (bottom)  and $\Delta = 0.17$ (top).
   \textbf{(b)} Frequency vs $\Delta$ for $k_1=k_2=1$ (green open circles)
   and the numerical solution and the KdV prediction (black line).
   \label{fig:periodic}
   }
\end{figure}

The KdV equation, Eq.~\eqref{eq:KdV},  has the following three parameter family of periodic solutions,
\begin{equation} \label{eq:kdvperiodic}
\begin{aligned}
A(X,T) &= \frac{\sqrt{a_1}}{a_2 r}\left(r_1 + r_2 - r_3 + 2(r_3 - r_1) \dn^2 \left(2\sqrt{\frac{(r_3 - r_1)r }{    \sqrt{a_1}(k_1^4 + k_2^4)  } }(X - V T) ; \mu \right) \right), 
\end{aligned}
\end{equation}
where
$$ V = \left( \frac{r_1 + r_2 + r_3}{3}\right), \qquad \mu = \frac{r_2-r_1}{r_3-r_1}. $$
The parameters are ordered $r_1 \leq r_2 \leq r_3$, and $\dn$ is one of the Jacobi elliptic functions with elliptic parameter $0 \leq \mu \leq 1$
 ($\sqrt{\mu}$ is the so-called modulus) \cite{NIST2010}.
Notice that in the limit $\mu \rightarrow 1$ one recovers the solitary waves of the model. In particular,
if one considers $r_1\rightarrow 0$ and $r_2\rightarrow r_3 = 3 \sigma/2 $
then one recovers Eq.~\eqref{eq:KdVsoliton}.

Returning to the original 2D lattice variables we have
\begin{equation} \label{eq:2Dper}
\begin{aligned}
u_{n,m}(t) &= \epsilon^2 \frac{\sqrt{a_1}}{a_2 r}\left(r_1 + r_2 - r_3 + 2(r_3 - r_1) \dn^2 \left(2\sqrt{\frac{(r_3 - r_1)r }{    \sqrt{a_1}(k_1^4 + k_2^4)  } }\epsilon( k_1n+k_2m - (\sqrt{a_1} r + V \epsilon^2 )t) ; \mu \right) \right), 
\end{aligned}
\end{equation}
where $V$ and $\mu$ are defined as before. Similar to the solitary wave case, where there was only a single effective parameter $\sigma \epsilon^2/r$, we have three effective parameters for this periodic wave
$$ R_1 = \frac{\epsilon^2 r_1}{r}, \quad R_2 = \frac{\epsilon^2 r_2}{r}, \quad R_3 = \frac{\epsilon^2 r_3}{r} $$
in which case we can write the periodic wave as
\begin{equation} \label{eq:2Dper_normal}
\begin{aligned}
u_{n,m}(t) &=  \frac{\sqrt{a_1}}{a_2}\left(R_1 + R_2 - R_3 + 2(R_3 - R_1) \dn^2 \left(2\sqrt{\frac{(R_3 - R_1) }{    \sqrt{a_1}(\cos^4(\theta) + \sin^4(\theta))  } }( \cos(\theta) n+ \sin(\theta)m - (\sqrt{a_1} + \Delta)t) ; \mu \right) \right), 
\end{aligned}
\end{equation}
where
$$ \Delta = \epsilon^2\frac{r_1 + r_2 + r_3}{3r} = \frac{R_1 + R_2 + R_3}{3}$$
which represents the deviation from the sound speed $c_s=\sqrt{a_1}$.
The amplitude of the periodic wave is 
$$a^+ =\frac{\max(u)-\min(u)}{2}= \frac{\sqrt{a_1}}{a_2} (R_2 - R_1). $$
Note, we include the $1/2$ factor in the definition of the amplitude
to be consistent with standard definitions of amplitude, e.g., for harmonic
waves \cite{GP73}.
If one defines 
$$\Theta = \frac{2\pi}{K(\mu)}  \sqrt{\frac{(R_3 - R_1) }{    \sqrt{a_1}(\cos^4(\theta) + \sin^4(\theta))  } } (\cos(\theta) n+ \sin(\theta)m - (\sqrt{a_1} + \Delta)t))    $$
where $K(\mu)$ and $E(\mu)$ are complete elliptic integrals of the first and second kind, respectively, then we can rewrite Eq.~\eqref{eq:2Dper_normal} as
\begin{equation} \label{eq:2Dparameterized}
  u_{n,m}(t) = u(\Theta) =U(\Theta) + \bar{u}     
\end{equation}
where the mean is 
\begin{equation} \label{eq:kdvmean}
\bar{u} = \frac{1}{2\pi}\int_0^{2\pi} u(\Theta)  d\Theta = \frac{\sqrt{a_1}}{a_2}   ( R_2 + R_1 - R_3
+ 2(R_3-R_1) E(\mu)/K(\mu)) 
\end{equation}
and 
$$ U(\Theta) =
 \frac{\sqrt{a_1}}{a_2} \left(2(R_3 - R_1) \dn^2 \left( \frac{K(\mu)}{\pi} \Theta; \mu \right) - 2(R_3-R_1) E(\mu)/K(\mu) \right)
$$
is a $2\pi$-periodic wave $U(\Theta) = U(\Theta + 2\pi)$ with zero mean
$$\frac{1}{2\pi}\int_0^{2\pi} U(\Theta)  d\Theta = 0. $$
The periodic wave shown in Eq.~\eqref{eq:2Dparameterized} approximates
 the one in Eq.~\eqref{eq:tw} with the following wave numbers
and frequency
\begin{subequations} \label{eq:perparm}
\begin{align}
   \hat{r} &=  \frac{2 \pi}{K(\mu)}   \sqrt{\frac{(R_3 - R_1) }{    \sqrt{a_1}(\cos^4(\theta) + \sin^4(\theta) ) } }, \\
    \hat{k}_1 &= \hat{r} \cos(\theta), \\ 
     \hat{k}_2 &= \hat{r} \sin(\theta), \\
    \omega &= \hat{r}(\sqrt{a_1} + \Delta).
\end{align}
\end{subequations}
The deviation from the linear frequency $\hat{r}\sqrt{a_1}$ is $\hat{r} \Delta$. 
Figure~\ref{fig:periodic} shows comparisons between
the KdV prediction, Eq.~\eqref{eq:2Dparameterized}, 
and numerical solutions of the fixed-point problem shown in 
Eq.~\eqref{eq:integral2} with $2\pi$ periodic boundary conditions.
We fix $\theta=\pi/4$ and $r_1 = 0.01$, $r_2=0.9$ and $r_3 = 1$ and then
 choose various values of $\epsilon$ to obtain $R_1$,$R_2$,$R_3$
(such that there is one parameter set for each value of $\epsilon$). Once
 the $(\theta,R_1,R_2,R_3)$ parameter set is chosen, the corresponding
wavenumbers $k_1,k_2$, and frequency $\omega$ 
are chosen according to Eq.~\eqref{eq:perparm} for the lattice solution.
The mean $\bar{u}$ is computed using Eq.~\eqref{eq:kdvmean} and
the corresponding amplitude parameter $R$ is computed using
Eq.~\eqref{eq:Rdef}. Now with
the parameter set $(R,k_1,k_2,\bar{u})$ we can implement
the fixed-point scheme to solve Eq.~\eqref{eq:integral2}
to obtain the corresponding numerical periodic wave
of the lattice. Figure~\ref{fig:periodic}(a) shows the KdV
periodic wave and corresponding numerical lattice solution
with $k_2/k_1=1$ for two different speed parameters
$\Delta = 0.04$ and $0.17$. Notice that, due to the
constraint of $2\pi$ periodicity of $U$, the formula of the KdV prediction of the
periodic wave is independent of the wavevector,
see Eq.~\eqref{eq:2Dparameterized}. However, which lattice solution
this is mapped to will depend on the wavevector, in virtue of Eq.~\eqref{eq:perparm}.
Figure~\ref{fig:periodic}(b) shows the speed-frequency relationship
for the KdV prediction (black line) and the numerical
solution for the $k_2/k_1=1$ case (green circles).
As before, the agreement is good for small deviations
from the sound speed (namely for $\Delta$ small).

\section{Dispersive Shock Waves} \label{sec:DSWs}

\begin{figure}[t!] %
    \centerline{
   \begin{tabular}{@{}p{0.34\linewidth}@{}p{0.34\linewidth}@{}p{0.34\linewidth}@{}}
     \rlap{\hspace*{5pt}\raisebox{\dimexpr\ht1-.1\baselineskip}{\bf (a)}}
 \includegraphics[height=4.2cm]{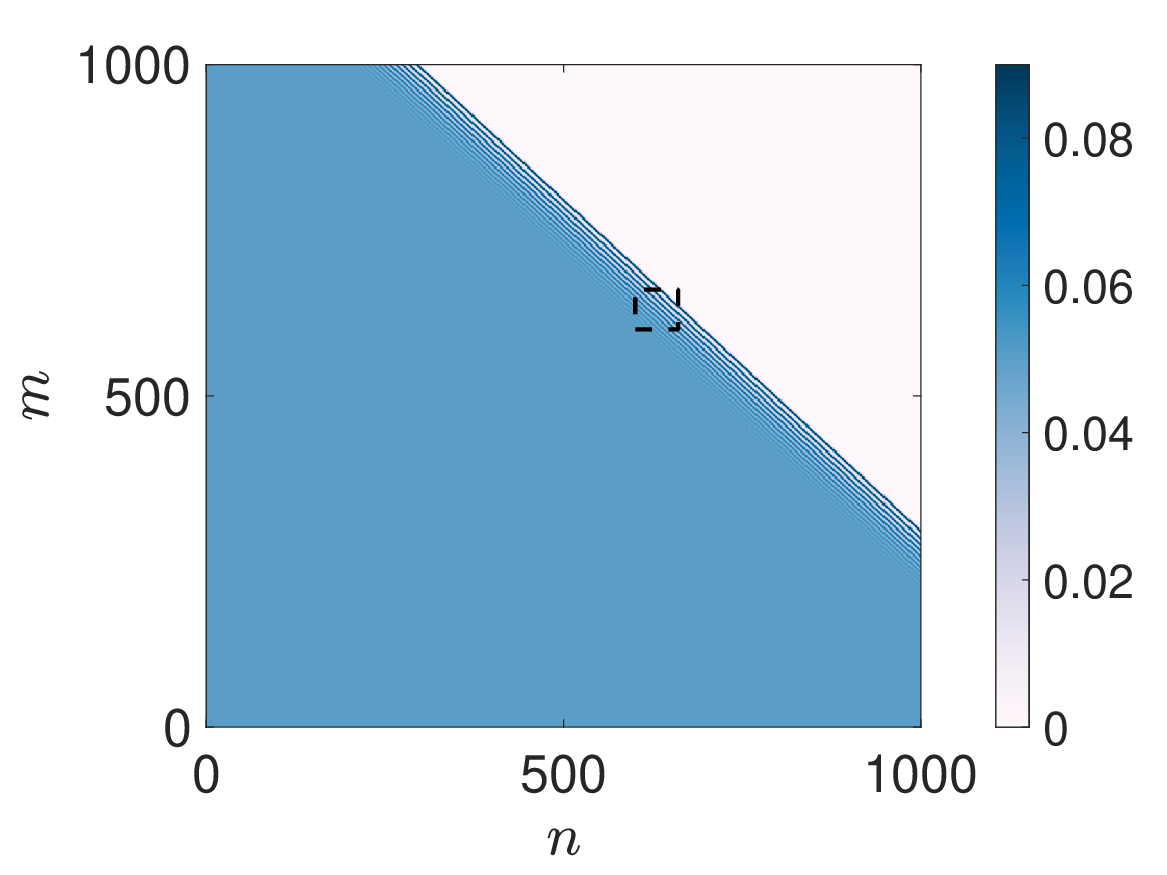}  &
   \rlap{\hspace*{5pt}\raisebox{\dimexpr\ht1-.1\baselineskip}{\bf (b)}}
 \includegraphics[height=4.2cm]{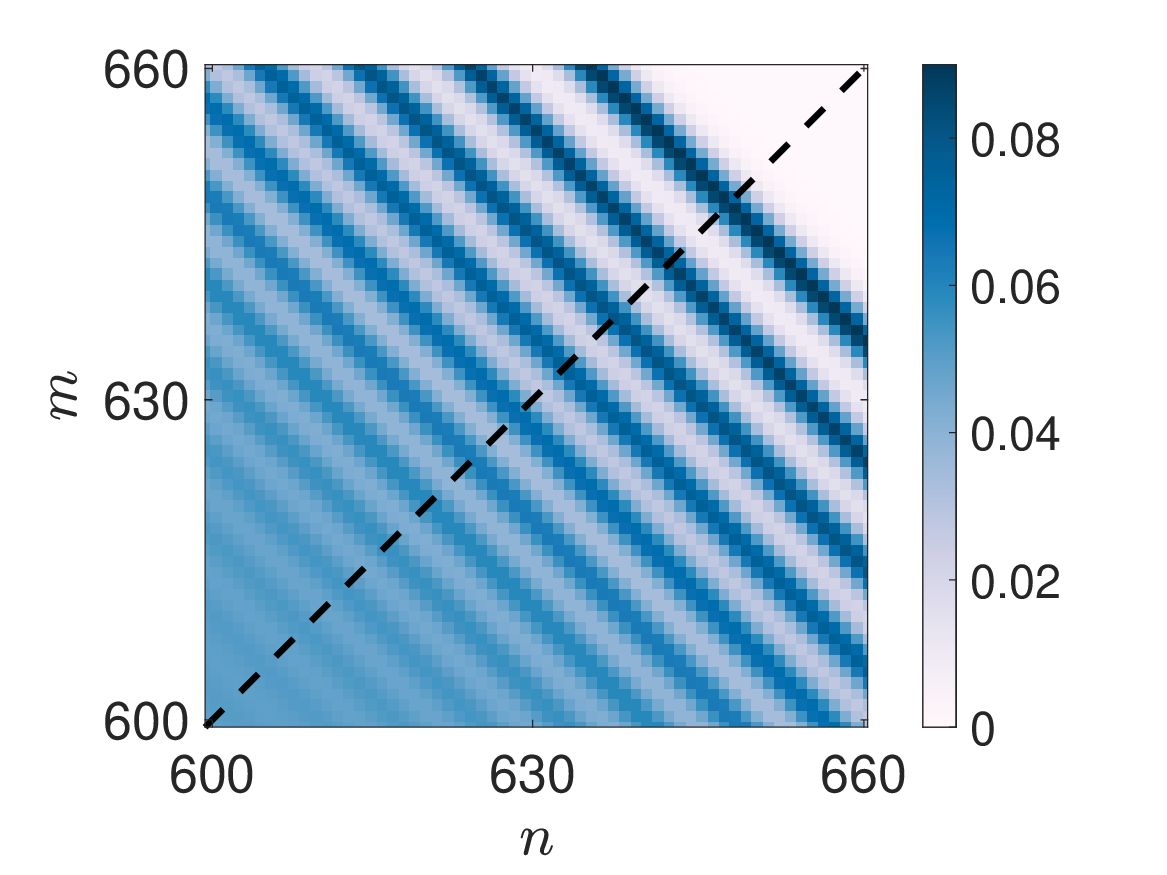} &
    \rlap{\hspace*{5pt}\raisebox{\dimexpr\ht1-.1\baselineskip}{\bf (c)}}
 \includegraphics[height=4.2cm]{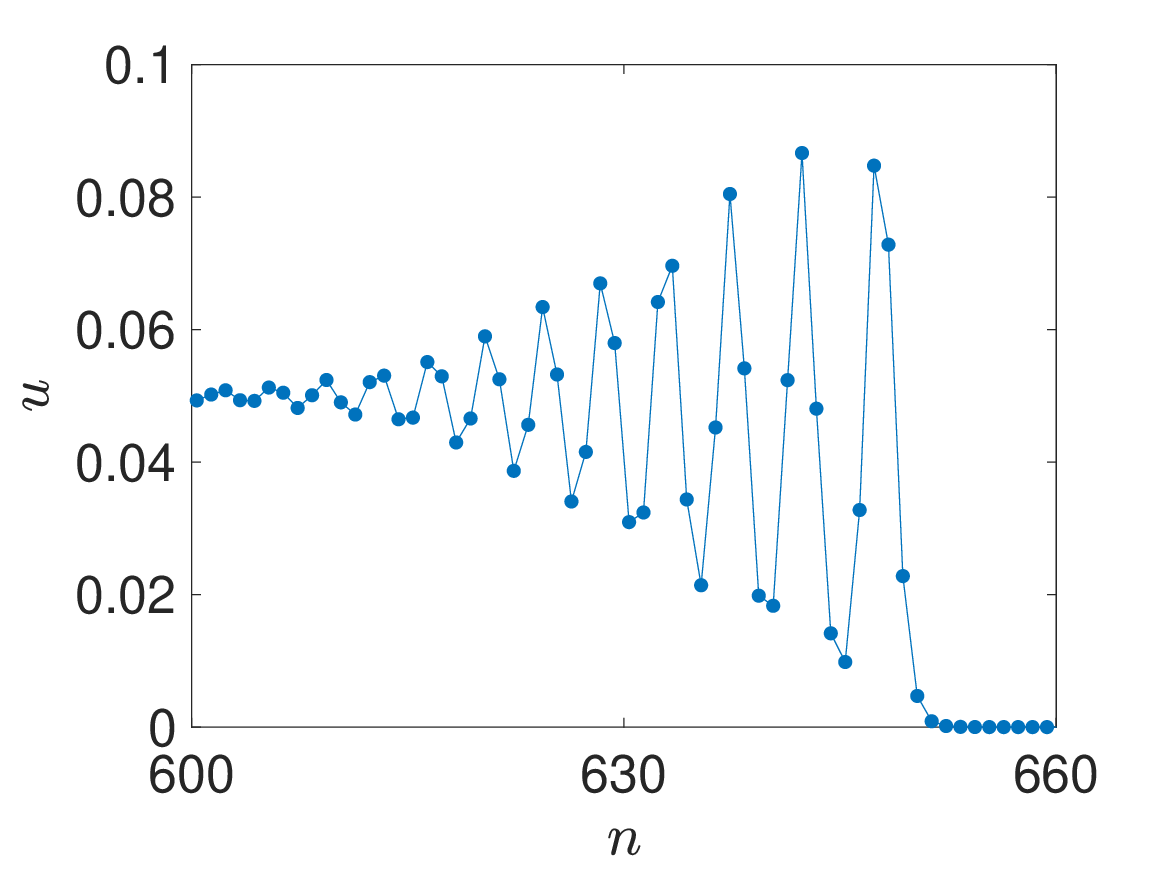} 
  \end{tabular}
  }
      \centerline{
   \begin{tabular}{@{}p{0.34\linewidth}@{}p{0.34\linewidth}@{}p{0.34\linewidth}@{}}
     \rlap{\hspace*{5pt}\raisebox{\dimexpr\ht1-.1\baselineskip}{\bf (d)}}
 \includegraphics[height=4.2cm]{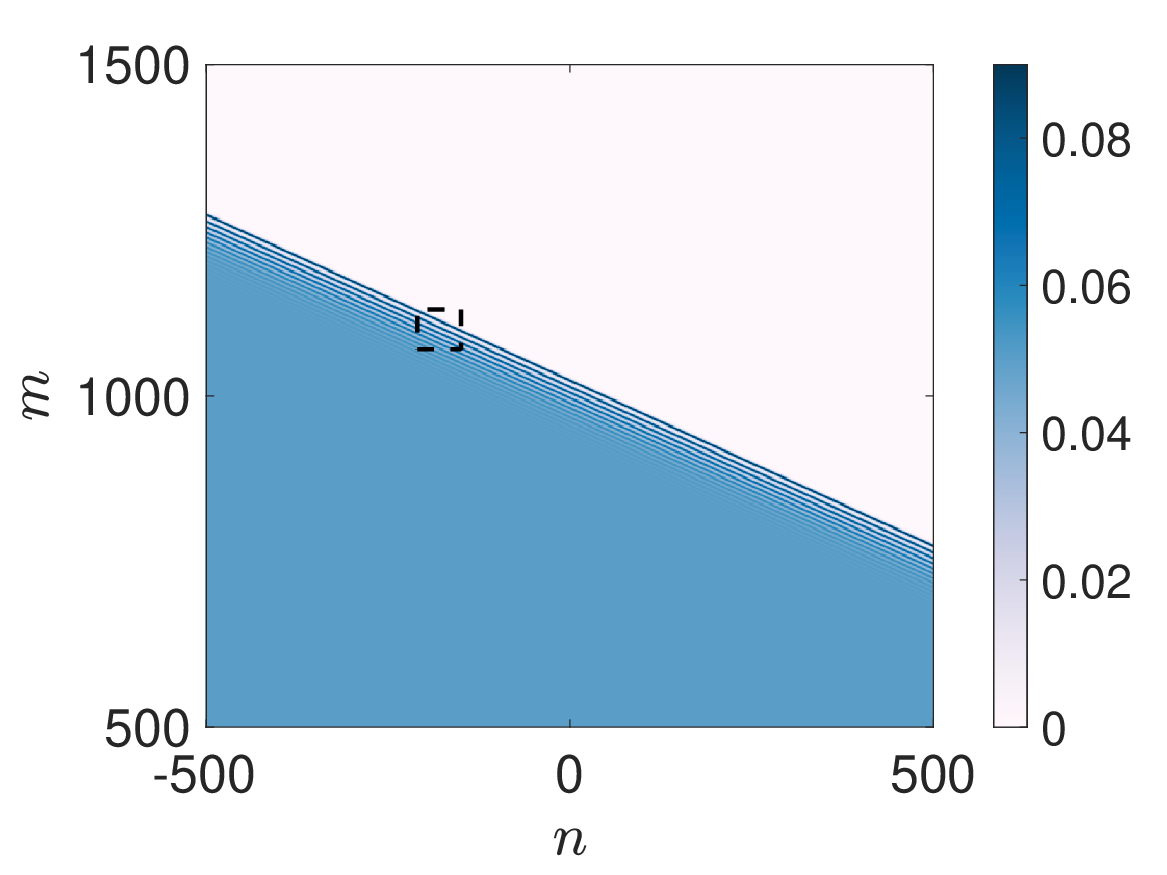}  &
   \rlap{\hspace*{5pt}\raisebox{\dimexpr\ht1-.1\baselineskip}{\bf (e)}}
 \includegraphics[height=4.2cm]{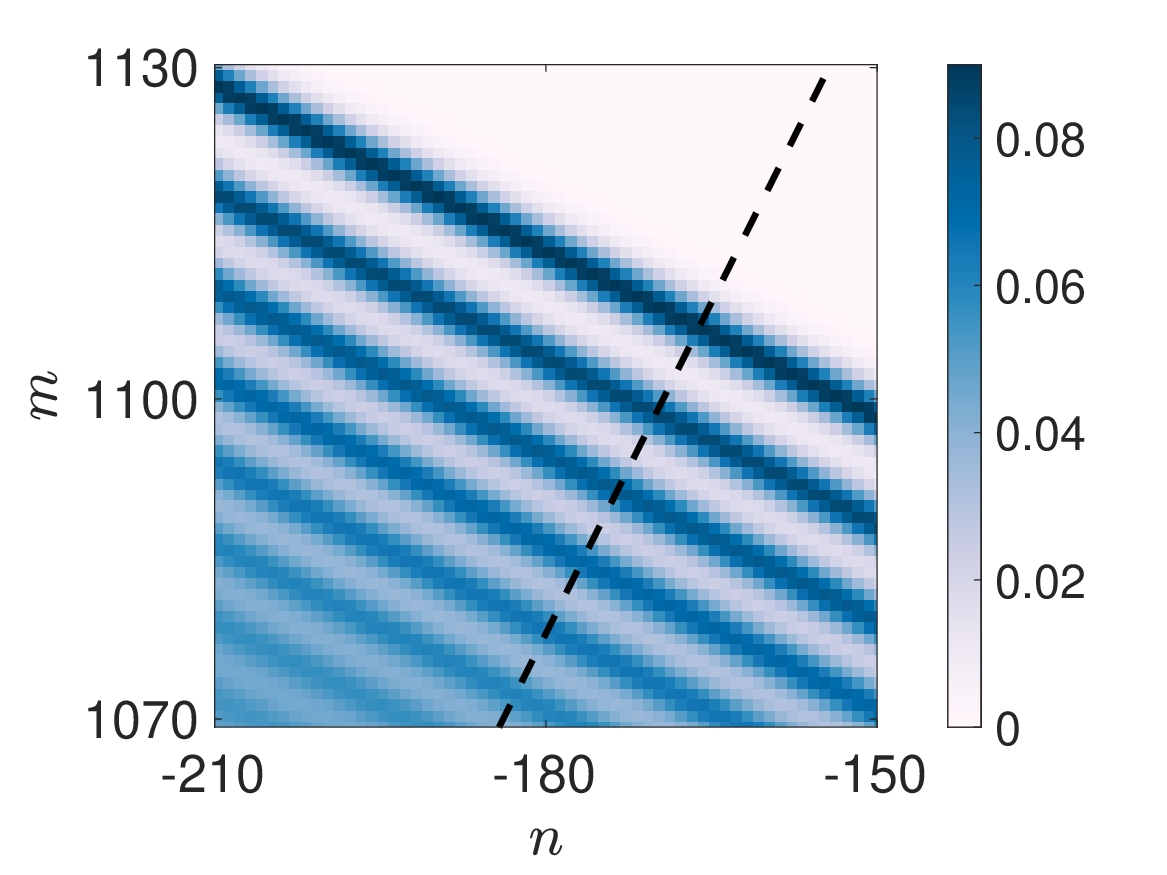} &
    \rlap{\hspace*{5pt}\raisebox{\dimexpr\ht1-.1\baselineskip}{\bf (f)}}
 \includegraphics[height=4.2cm]{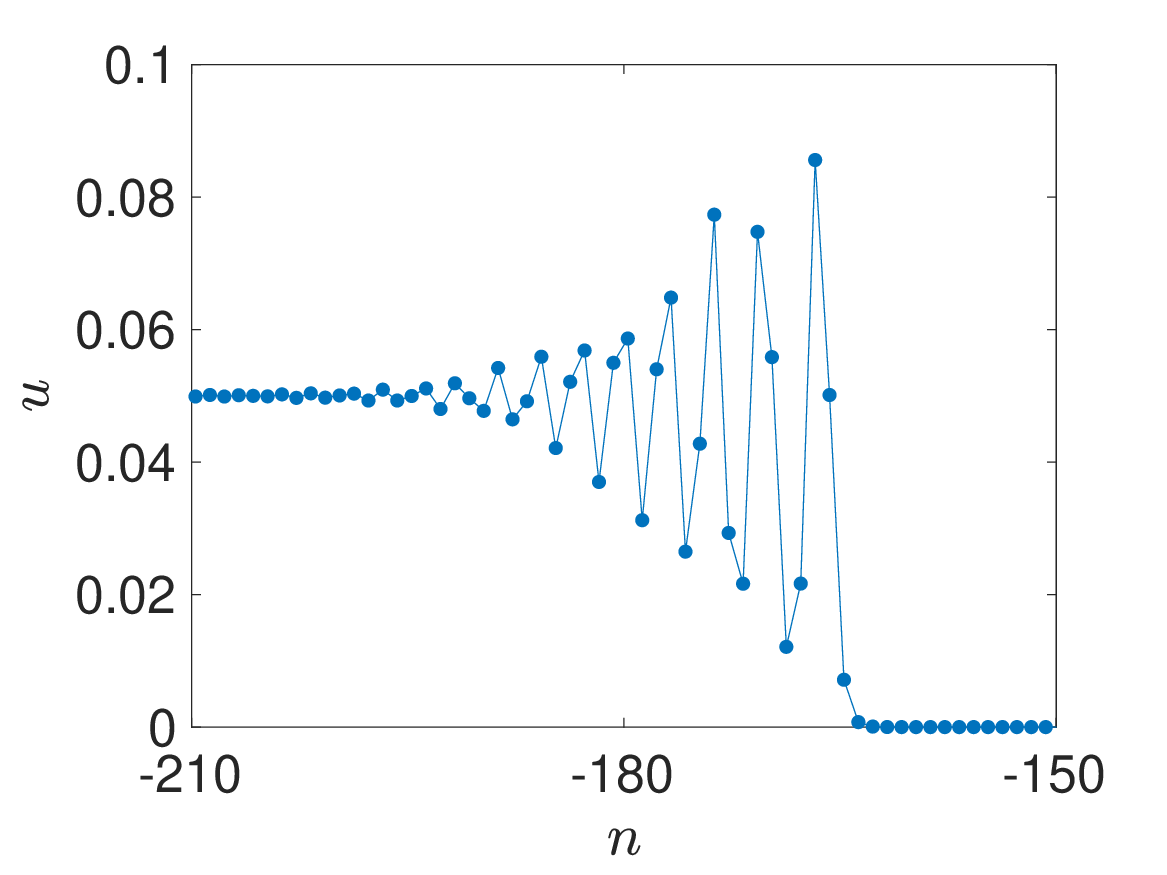} 
  \end{tabular}
  }
   \caption{\textbf{(a)} Intensity plot of solution at $t=894.4$ with
    $k_2/k_1=1$ and $\delta = 0.05$. Color intensity
     corresponds to $u_{n,m}$. 
     \textbf{(b)} Zoom of the boxed area of panel (a). At this scale,
     the 2D DSW structure can be seen. Data is extracted along the dashed line (with slope $k_2/k_1$) 
     which is perpendicular to the initial interface.
     \textbf{(c)} Spatial profile
     along the line of observation (the dashed line in panel (b)).
     \textbf{(d-f)} Same as top panels with $k_2/k_1=2$
   }
    \label{fig:DSWdensity}
\end{figure}

With a firm handle on traveling waves of our model, we are now ready to 
explore DSWs, such as the one shown in Fig.~\ref{fig:DSWexamples}(d).
We will first explain the details of the numerical computations
and explore various simulations. We then return to the KdV description
to obtain an analytical approximation of the DSWs which is then
compared to the simulations. After that,
we implement the DSW fitting method to describe the
trailing and leading edge of the DSWs.

\subsection{Numerical Computation of DSWs}

We consider a square lattice $(n,m)\in[-N,N]\times [-M,M]$, where $N,M\in\N$
are the respective lattice dimension parameters. We employ no flux boundary conditions, namely
$u_{-N-1,m}=u_{-N,m}$, $u_{N+1,m}=u_{N,m}$, $u_{n,-M-1}=u_{n,-M}$ and $u_{n,M+1}=u_{n,M}$.
The initial conditions used are of the form
\begin{equation} \label{eq:smoothIC}
\begin{aligned}
u_{n,m}(0) &= \epsilon^2 f(k_1 n + k_2 m) \\
\dot{u}_{n,m}(0) &= g(k_1 n + k_2 m)
\end{aligned}
\end{equation}
where $f(x)$ is smoothed variant of a step function. 
 Specifically,  we use
\begin{equation} \label{eq:f}
    f(x) =  \frac{\sqrt{a_1}}{ 2 a_2 r} \big( 1-\tanh(w  x) \big)\,,
\end{equation}
where $w$ is the smoothing parameter. For numerical computations we take
$w=0.5$. Figure~\ref{fig:DSWexamples}(c) shows a plot of the IC given in Eq.~\eqref{eq:smoothIC}
with system parameters $a_1=a_2=1$, $a_3=0$ and small parameter $\delta = \epsilon^2 = 0.2$ and wavevector parameters such that $r=1$ and $k_2/k_1=1$. In the figure one can
see the jump is smoothed out. The smoothness of $f$ is required in order to
compute an initial velocity 
that is consistent with the KdV ansatz.
In particular, since
\begin{eqnarray*} \label{eq:smoothvIC}
 \dot{u}_{n,m} &=&  \frac{d}{dt} \epsilon^2 A(X,T) \\
 &=& -c r \epsilon^3 A_X(X,T) + \epsilon^5 A_T \\
 &=& -\sqrt{a_1} r \epsilon^3 A_X(X,T) - \epsilon^5\left( \frac{\sqrt{a_1}}{24r}(k_1^4+k_2^4) A_{XXX} + \frac{a_2 r}{\sqrt{a_1}}AA_X       \right)
\end{eqnarray*}
which assumes that $A(X,T)$ satisfies the KdV equation, we have that
\begin{equation}
    g(x) =  - \epsilon^3 \sqrt{a_1} r f' - \epsilon^5 \left(   \frac{\sqrt{a}_1 }{24r}(k_1^4 + k_2^4) f''' + \frac{a_2 r}{\sqrt{a_1}} f f' \right).     \label{eq:initial_velocity} 
\end{equation}

In all our numerical simulations, we keep the domain in the KdV scaling fixed to
$X\in[-330,330]$ and $T\in[0,10]$. Assuming a square domain for the computation,
we have that $M=N=330\epsilon^{-1}$ and end simulation time $t_f = 10\epsilon^{-3}$. This implies
that as we consider smaller $\epsilon$ the spatial and temporal domain time will grow.
For the simulation, we employ the Verlet scheme \cite{Herrmann2018} and use Matlab's ODE45
function as well for a consistency check.

The DSW in Fig.~\ref{fig:DSWexamples}(d) is the result of simulating Eq.~\eqref{eq:model1D} with
initial condition given by Eq.~\eqref{eq:smoothIC}. Another example
simulation with $a_1=a_2=1$, $a_3=0$ and $k_2/k_1=1$ but now with $\delta=0.05$
is shown in Fig.~\ref{fig:DSWdensity}. In panel (a), an intensity plot is shown
at time $t_f = 10 \epsilon^{-3} \approx 894.4$, where color intensity corresponds
to the value of $u_{n,m}$. The solution seen is a 
line DSW, since it is constant along vectors orthogonal to $(k_1,k_2)$ and propagates 
along $(k_1,k_2)$. Panel (b) is a zoom of the boxed area, where the oscillatory 
structure can be better seen. It is also informative to extract
data along the propagation direction $(k_1,k_2)$. From a practical perspective,
this restricts our choices of $k_2/k_1$ to be
rational. For demonstration purposes, we fix $r=1$ and present results
for $k_2/k_1 =1$ and $k_2/k_1 =2$.  We extract data according to $(n, k_2/k_1 n)$ 
for $n=-N,\ldots,N$ resulting in a matrix of data with a single space
dimension and the time dimension. 
In panel (b) the dashed line is where the data is extracted,
which is then plotted in panel (c). In this view 
the DSW has the same character as the 1D DSW (compare to Fig.~\ref{fig:DSWexamples}(b)),
where the trailing and leading edge can be identified. The bottom
panels of Fig.~\ref{fig:DSWdensity} are the same as the top ones
but with $k_2/k_1 = 2$. The line DSWs in the top and bottom panels
are qualitatively similar. To uncover quantitative differences,
we will employ an analytical approximation using the KdV equation,
which we consider next.

\subsection{KdV Description of DSWs}
We now return to the KdV equation given in Eq.~\eqref{eq:KdV}, turning our attention to its DSW solutions.
For the readers' convenience, we
summarize the relevant result here (see \cite{Mark2016,Kamchatnov} for more thorough derivations). DSWs in the KdV equation arise from Riemann initial data, such as
\begin{equation}\label{eq:Astep}
A(X,0) = \left\{
\begin{split}
\frac{\sqrt{a_1}}{a_2 r}, \quad & X < 0 \\
0, \quad & X > 0.
\end{split}
\right.
\end{equation}
The scaling factor $\frac{\sqrt{a_1}}{a_2 r}$ is arbitrary and chosen for notational
convenience here. To obtain an analytical approximation of the DSW
one assumes the periodic wave (e.g., the $r_j$ of Eq.~\eqref{eq:kdvperiodic}) as varying slowly with respect to $X,T$ and averages three of the conserved quantities of the KdV equation over a period to obtain the Whitham modulation equations \cite{Whitham74}. These equations
are then solved subject to relevant boundary conditions for a DSW. In particular,
one seeks self-similar solutions, $r_j=r_j(X/T)$,
resulting in modulation equations of the form $(S_j - X/T)r'(X/T)=0$, 
where the characteristic speeds $S_j$ are nonlinear functions of $r_1,r_2$ and $r_3$. In this self-similar frame, $r_j$ is either constant, or $X/T$ is equivalent to the characteristic speed $S_j$.
Assuming step initial data of the form of Eq.~\eqref{eq:Astep}
one finds  $r_1 = 0, r_2 = \mu, r_3 = 1$, see the classic work of~\cite{GP73} for details,
where this analytical solution of the Whitham modulation equations was
first derived. Returning to the formula for the periodic wave, Eq.~\eqref{eq:kdvperiodic},
we have the following asymptotic DSW solution,
\begin{equation} \label{eq:kdv_dsw}
\begin{aligned}
A(X,T) &= \frac{\sqrt{a_1}}{a_2 r}\left(\mu-1 + 2 \dn^2 \left(2\sqrt{\frac{r }{    \sqrt{a_1}(k_1^4 + k_2^4)  } }(X - \frac{\mu+1}{3} T) ; \mu \right) \right), 
\end{aligned}
\end{equation}
where the parameter $\mu$ is related to the self-similar variable
in the following way
\begin{equation} \label{eq:v2}
\frac{X}{T} = S(\mu) := \frac{1+\mu}{3} - \frac{2}{3} \frac{\mu(1-\mu) K(\mu)}{E(\mu) - (1-\mu)K(\mu)},
\end{equation}
with $S(\mu)$ the second characteristic velocity of the KdV Whitham equations.
In Eq.~\eqref{eq:kdv_dsw}, the limit $\mu\rightarrow 0$ corresponds to the trailing, harmonic wave edge, 
while the limit $\mu\rightarrow 1$ corresponds to the leading, solitary wave edge.
Note that $S$ has the following limiting values,
\begin{equation}
 \lim_{\mu\rightarrow 0} S(\mu) =s^- = -1, \qquad \lim_{\mu\rightarrow 1} S(\mu) = s^+ = 2/3
\end{equation}

\begin{figure}[t!] %
    \centerline{
   \begin{tabular}{@{}p{0.5\linewidth}@{}p{0.5\linewidth}@{}}
     \rlap{\hspace*{5pt}\raisebox{\dimexpr\ht1-.1\baselineskip}{\bf (a)}}
 \includegraphics[height=6cm]{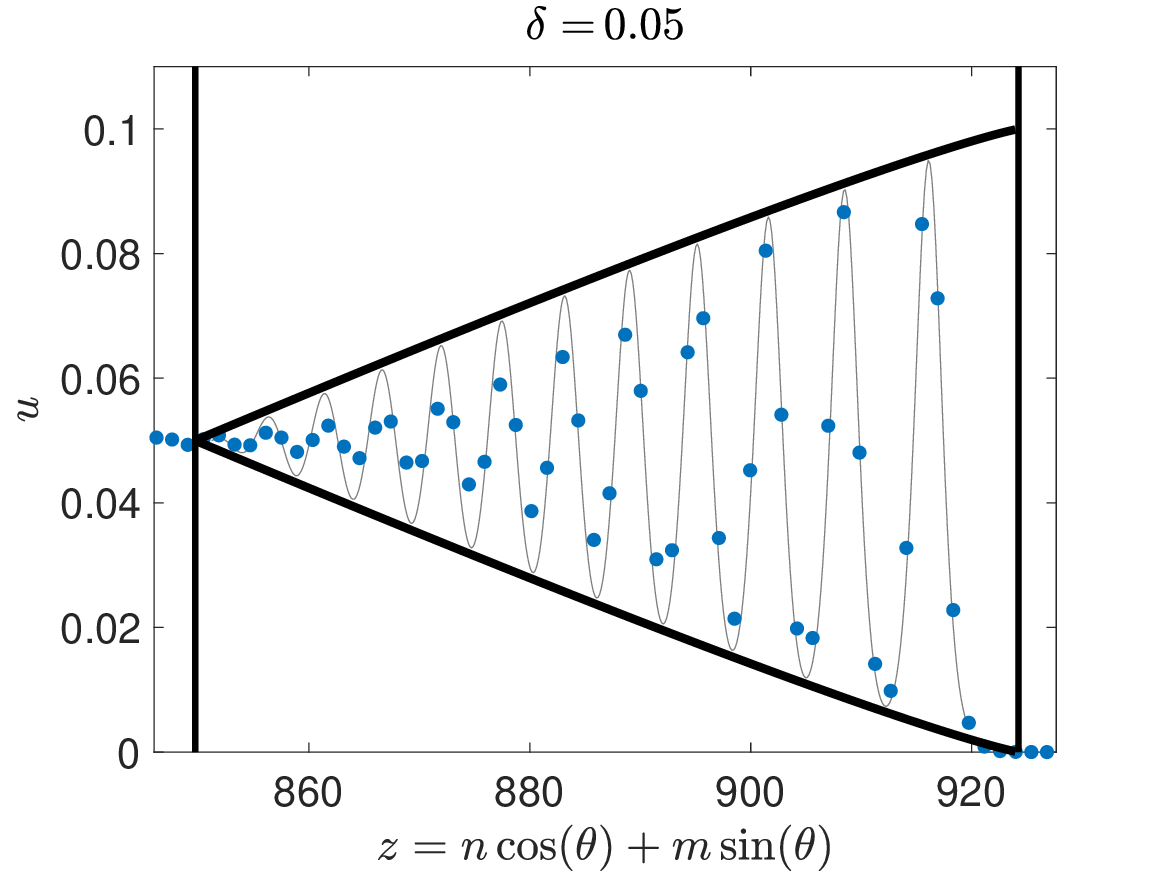}  &
   \rlap{\hspace*{5pt}\raisebox{\dimexpr\ht1-.1\baselineskip}{\bf (b)}}
 \includegraphics[height=6cm]{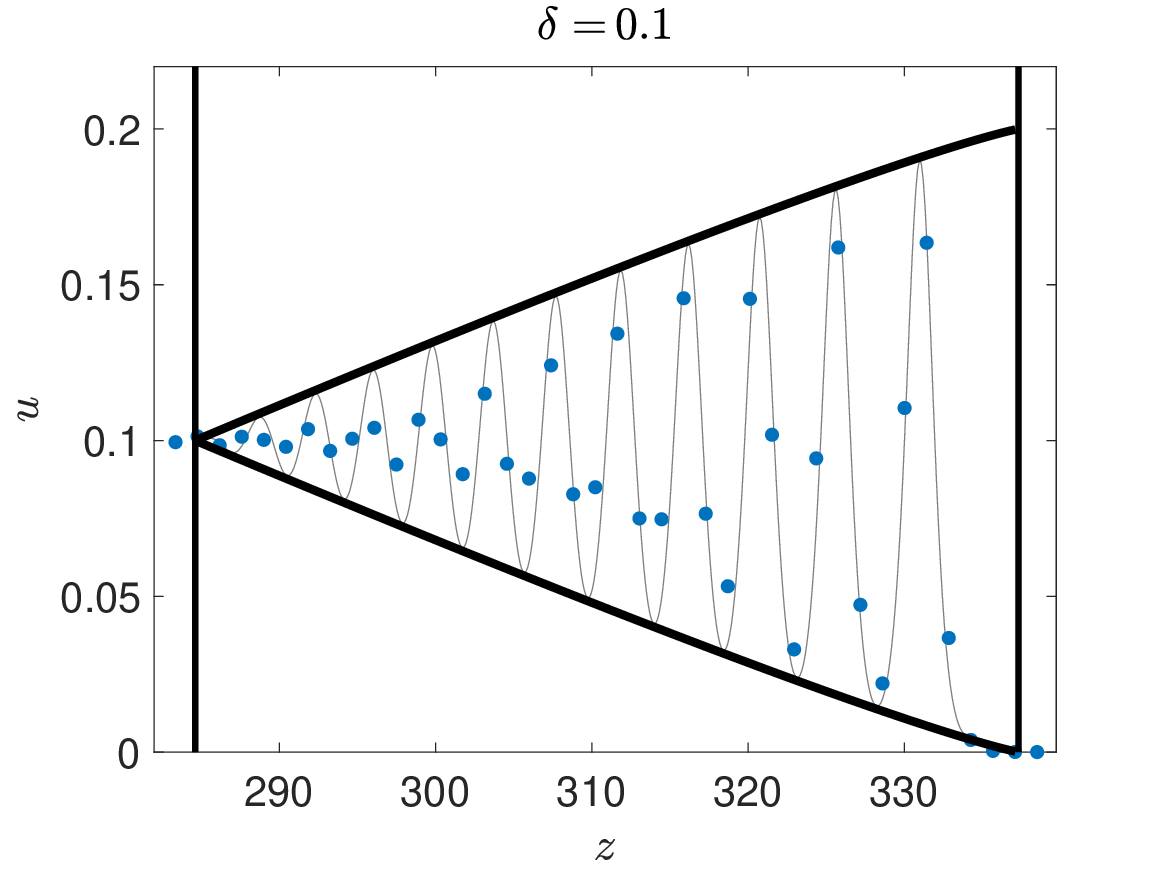} 
  \end{tabular}
  }
      \centerline{
   \begin{tabular}{@{}p{0.5\linewidth}@{}p{0.5\linewidth}@{}}
     \rlap{\hspace*{5pt}\raisebox{\dimexpr\ht1-.1\baselineskip}{\bf (c)}}
 \includegraphics[height=6cm]{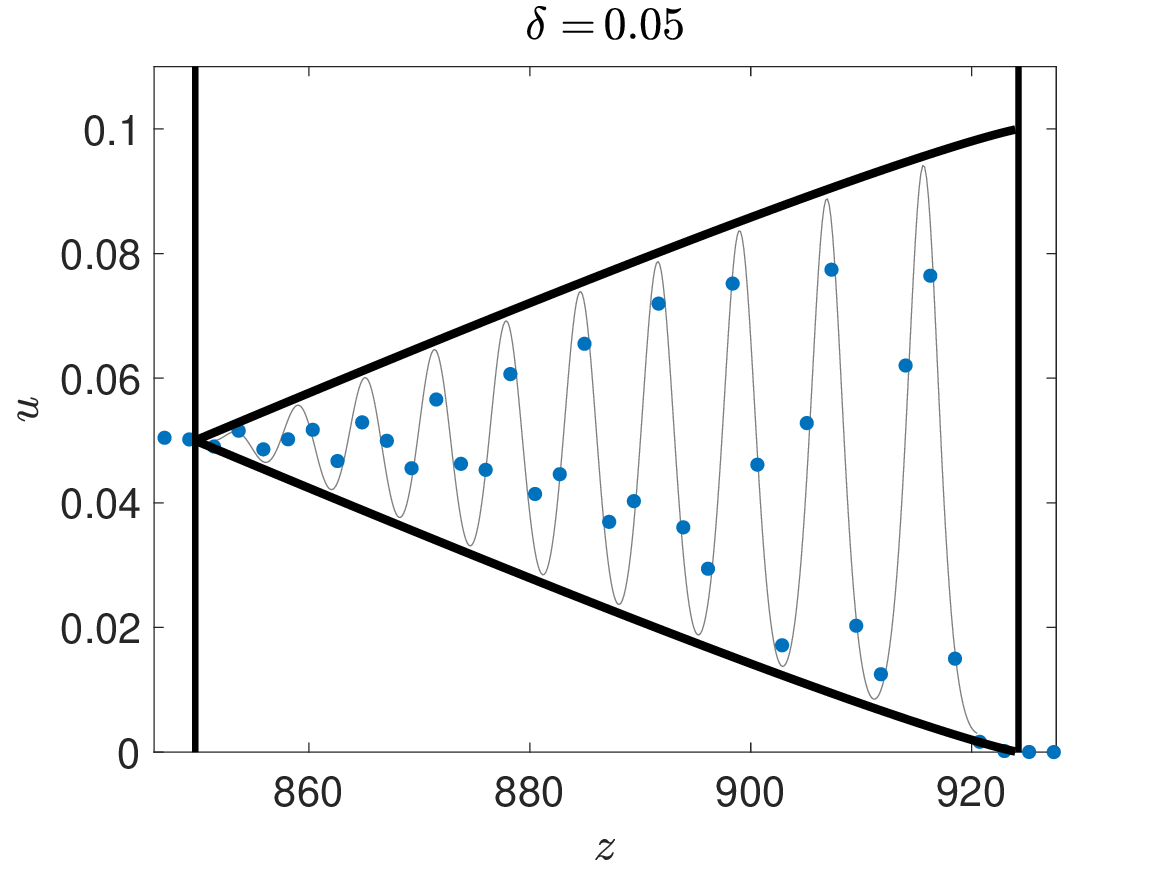}  &
   \rlap{\hspace*{5pt}\raisebox{\dimexpr\ht1-.1\baselineskip}{\bf (d)}}
 \includegraphics[height=6cm]{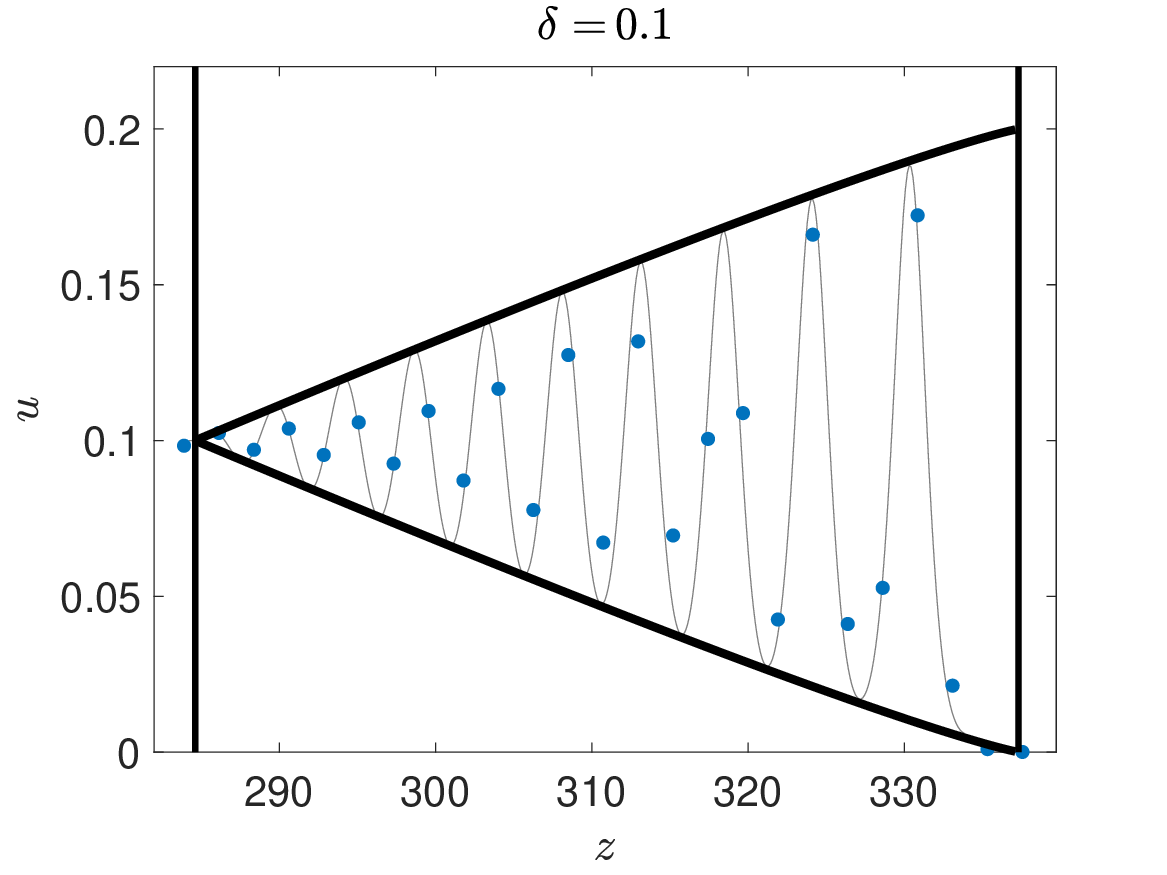} 
  \end{tabular}
  }
  \caption{Comparison of lattice DSW solution (markers) and the KdV approximation (lines).
  The data for the lattice solution is extracted along the line of observation
  (namely along the vector $(k_1,k_2)$) and is plotted against the coordinate
  $z = k_1 n + k_2 m$. 
In all panels, the variables in the KdV scaling
    are fixed to $T=10$ and $X\in[-350,350]$. The vertical lines
    are the predictions of the trailing edge ($\rho = t \,s^- $) and leading edge ($\rho= t\,s^+ $).
    The sloped lines are the prediction of the envelopes.
    \textbf{(a)} The small parameter is $\delta = 0.05$, $k_2/k_1=1$ and the time in the original lattice scaling is $t \approx 900$. The best-fit phase shift is the linear interpolation
    between $\theta_\ell = 1.5$ and $\theta_r = 1.5$.
    \textbf{(b)} The small parameter is $\delta = 0.1$, $k_2/k_1=1$ and the time in the original lattice scaling is $t \approx 315$. The best-fit phase shift is the linear interpolation
    between $\theta_\ell = 0.5$ and $\theta_r = 0.5$.
\textbf{(c)} Same as (a) with  $k_2/k_1=2$, $\theta_\ell = 1.5$ and $\theta_r = 2.5$.
\textbf{(d)} Same as (b) with  $k_2/k_1=2$, $\theta_\ell = -0.4$ and $\theta_r = 1$.
In all panels $r=1$. 
}
  \label{fig:compareDSW}
\end{figure}

Returning to the original
lattice variables, namely by using Eq.~\eqref{eq:ansatz1D}, we obtain the following approximate 2D lattice DSW,
\begin{equation}
\begin{aligned}
u_{n,m}(t) = \epsilon^2 \frac{\sqrt{a_1}}{a_2 r}\left(\mu-1 + 2 \dn^2 \left(2\sqrt{\frac{r }{    \sqrt{a_1}(k_1^4 + k_2^4)  } }(\epsilon( k_1 n + k_2m - crt) - \frac{\mu+1}{3} \epsilon^3 t) ; \mu \right) \right).
\end{aligned}
\end{equation}
Introducing the variable $\delta =  \epsilon^2/r$
the previous expression becomes
\begin{equation} \label{eq:kdvapprox}
\begin{aligned}
u_{n,m}(t) &= \delta \frac{\sqrt{a_1}}{a_2}
\left(  \mu - 1 + 2\dn^2 \left(2 \sqrt{ \frac{\delta}{\sqrt{a_1}(\cos^4(\theta)+\sin^4(\theta))}     } ( \cos(\theta) n + \sin(\theta) m - (\sqrt{a_1} + \frac{\mu+1}{3} \delta          )t) ; \mu \right) \right), 
\end{aligned}
\end{equation}
where the parameter $\mu$ is related to the variables $n,m,t$ in the following way
\begin{equation} \label{eq:v2again}
\frac{\cos(\theta) n + \sin(\theta) m }{t} = \sqrt{a_1} + S(\mu)\,\delta 
\end{equation}
Since the expression on the left-hand side represents the speed of
propagation in the observation direction (namely along the vector $(k_1,k_2)$) we can determine
the trailing and leading edge speeds by simply computing
the limit as $\mu\rightarrow 0$ and $\mu\rightarrow 1$
respectively, leading to
\vspace*{-0.4ex}
\begin{equation} \label{eq:speeds}
s^- = \sqrt{a_1} -\,\delta , \qquad s^+ = \sqrt{a_1} + \frac{2}{3}\,\delta 
\end{equation}
such that the leading edge speed $s^+$ is supersonic and the trailing edge speed $s^-$ is subsonic.  On the other hand 
we have the trailing edge
wavenumber $k_-$ and
leading edge amplitude $a_+$ given by
\begin{equation} \label{eq:wavenumber}
k^- = 4\sqrt{\frac{\delta}{\sqrt{a_1}(\cos^4(\theta)+\sin^4(\theta))}}  , \qquad a^+ = \frac{\max(u)-\min(u)}{2} =\delta\frac{\sqrt{a_1}}{a_2} 
\end{equation}
where we see that the only characteristic depending
on $\theta$ is the wavenumber. Note, the DSW
in Eq.~\eqref{eq:kdvapprox} corresponds
to the Riemann data in Eq.~\eqref{eq:step2}.

\begin{figure}[t!] %
\kern-\medskipamount
\centerline{
   \begin{tabular}{@{}p{0.5\linewidth}@{}p{0.5\linewidth}@{}}
     \rlap{\hspace*{5pt}\raisebox{\dimexpr\ht1-.1\baselineskip}{\bf (a)}}
 \includegraphics[height=6cm]{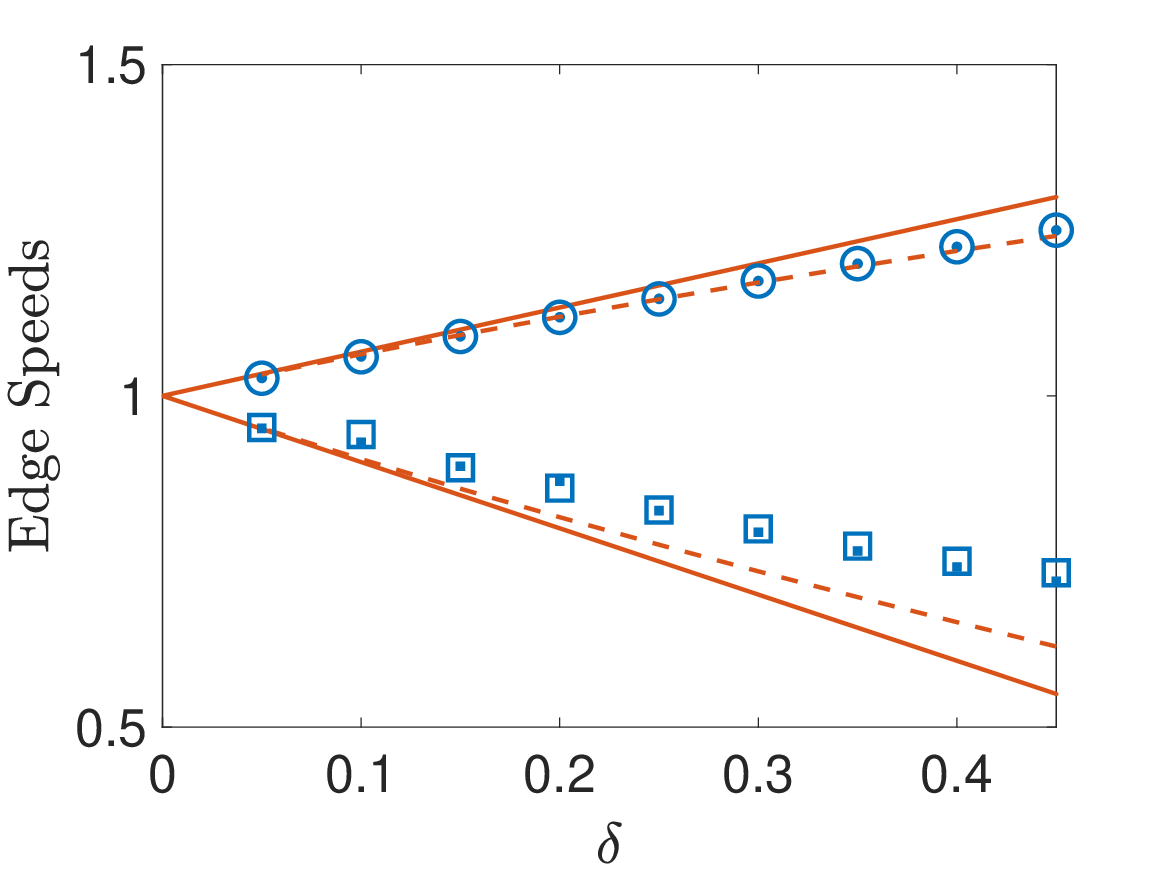}  &
   \rlap{\hspace*{5pt}\raisebox{\dimexpr\ht1-.1\baselineskip}{\bf (b)}}
 \includegraphics[height=6cm]{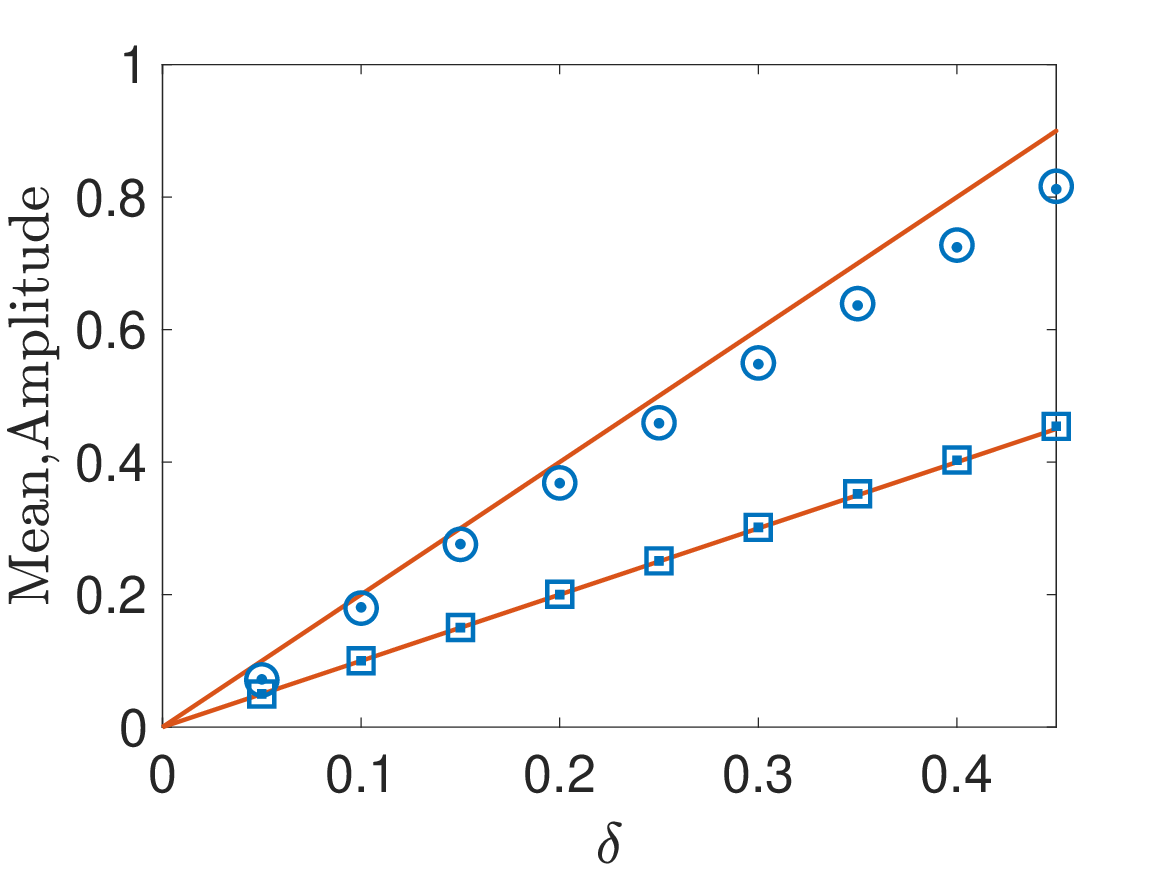} 
  \end{tabular}}
     \caption{\textbf{(a)} Trailing and leading edge speeds  of the DSW.  Blue markers indicate simulation results: circles for the leading edge and squares for the trailing edge. Empty and filled markers correspond to $k_2/k_1 = 2$ and $k_2/k_1 = 1$, respectively. The solid red lines represent the KdV predictions for the leading (line with positive slope) and trailing edge (line with negative slope) speeds. Recall that they are independent of $\theta$. The near overlap of filled and empty markers suggests that the speeds in the 2D lattice are also (nearly) independent of $\theta$. The dashed lines are the predictions of
     the edge speeds from the DSW fitting method, described in Sec.~\ref{sec:DSWfit}.
    \textbf{(b)} Trailing edge mean and leading edge amplitude. Blue markers indicate simulation results: circles for the leading edge amplitude and squares for the trailing mean. Empty and filled markers once again correspond to $k_2/k_1 = 2$ and $k_2/k_1 = 1$, respectively. The solid and dashed red lines represent the KdV predictions for the leading edge amplitude and trailing edge mean, respectively, which are independent of $\theta$. The near overlap of filled and empty markers suggests that the speeds in the 2D lattice are also practically independent of $\theta$. 
    }
   \label{fig:speeds}
\bigskip
    \centerline{
   \begin{tabular}{@{}p{0.5\linewidth}@{}p{0.5\linewidth}@{}}
     \rlap{\hspace*{5pt}\raisebox{\dimexpr\ht1-.1\baselineskip}{\bf (a)}}
 \includegraphics[height=6cm]{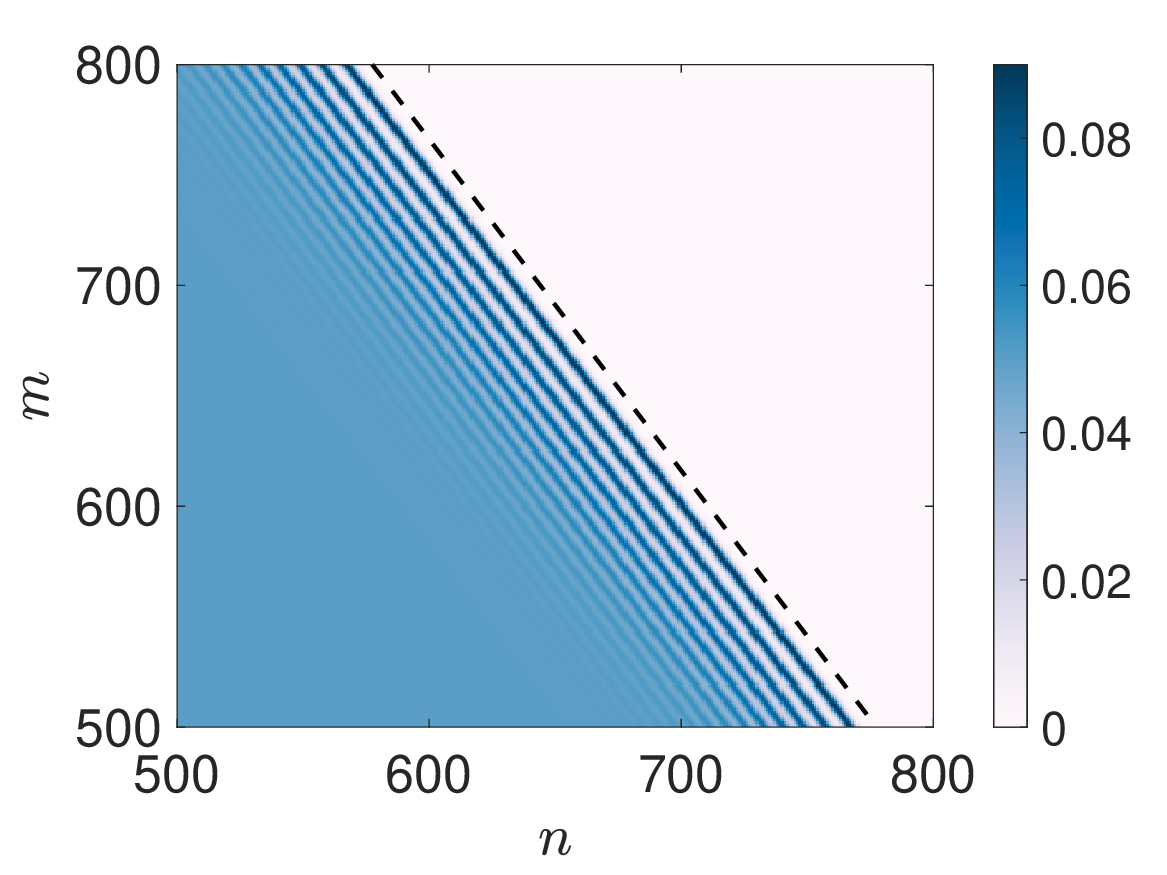}  &
   \rlap{\hspace*{5pt}\raisebox{\dimexpr\ht1-.1\baselineskip}{\bf (b)}}
 \includegraphics[height=6cm]{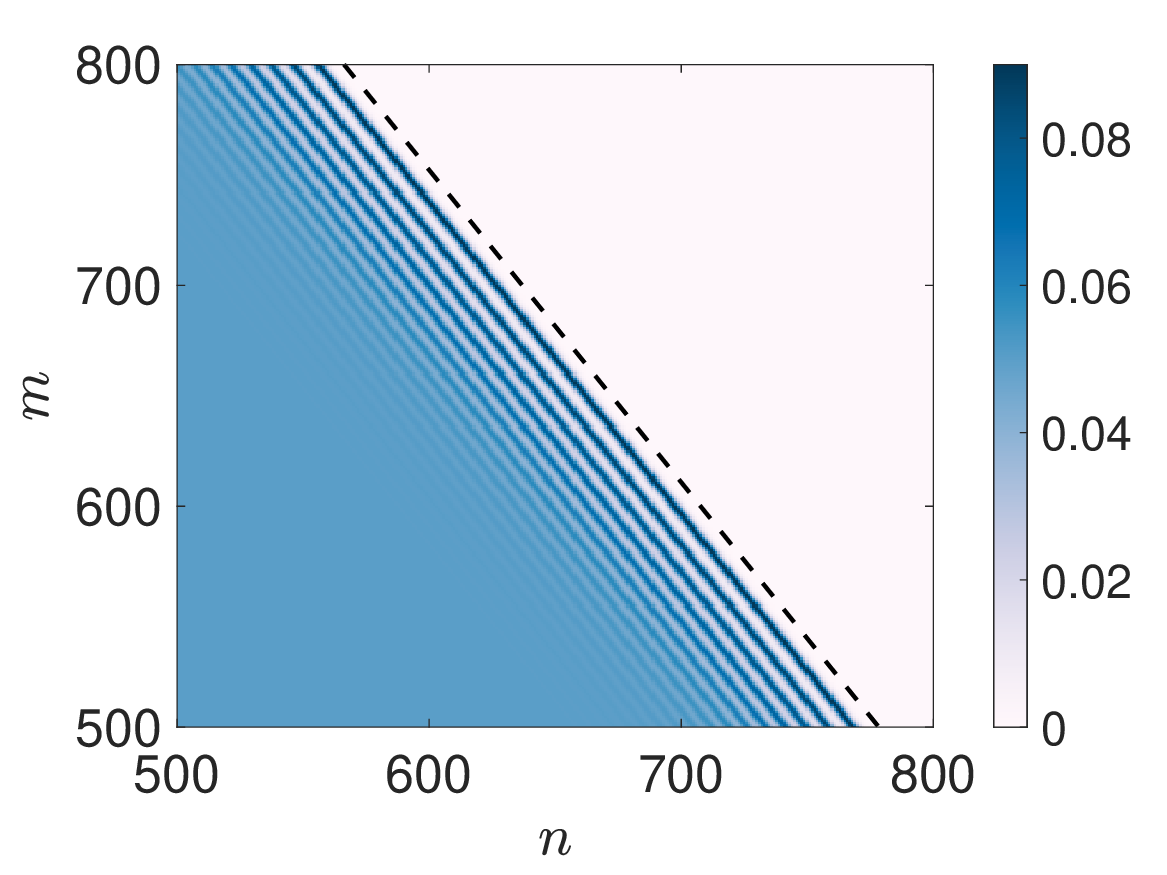} 
  \end{tabular}}
  \kern-\medskipamount
     \caption{Simulations at irrational angles.\textbf{(a)} Intensity plot of solution at $t=894.4$ with
    $k_2/k_1=\sqrt{2}$ and $\delta = 0.05$. Color intensity
     corresponds to $u_{n,m}$.  \textbf{(b)} Same as (a) but with  
     $k_2/k_1=(1+\sqrt{5})/2$. The dashed sloped line is the prediction
     of the leading edge location based on the KdV equation. 
    }
   \label{fig:irat}
\end{figure}

Notice that $f(x)$, defined in Eq.~\eqref{eq:f}, corresponds precisely
to a smoothed variant of the step function used to initialize
the KdV DSW, see Eq.~\eqref{eq:Astep}. In particular $f(x)$ tends
to $A(X,0)$ as the smoothing parameter $w\rightarrow \infty$.
In simulations we have chosen $a_1=a_2=r=1$,  and thus we have that 
$\delta = \epsilon^2$. Upon choosing
a wave vector $(k_1,k_2)$ we can easily compare
the KdV prediction in Eq.~\eqref{eq:kdvapprox}
to simulations of Eq.~\eqref{eq:model} with initial
values given in Eq.~\eqref{eq:smoothIC}.

Figure~\ref{fig:compareDSW} shows a comparison of the KdV prediction (smooth curves)
and numerical simulations (markers) for two jump heights ($\delta=0.05$ and $\delta=0.1$)
and two wavevectors ($k_2/k_1=1$ and $k_2/k_1=2$). Note that
the markers of panel Fig.~\ref{fig:compareDSW}(a) correspond to the solution shown
in Fig.~\ref{fig:DSWexamples}(c) and Fig.~\ref{fig:DSWdensity}(c).
The markers of Fig.~\ref{fig:compareDSW}(c) correspond to 
Fig.~\ref{fig:DSWdensity}(f). While Fig.~\ref{fig:DSWdensity}
shows the DSW profile versus the $n$ index, Fig.~\ref{fig:compareDSW}
shows the DSW versus the coordinate $z = \cos(\theta)n + \sin(\theta)m $,
which is most natural for comparisons with the KdV approximation, see
Eq.~\eqref{eq:kdvapprox}.
The gray line is
the KdV prediction given by Eq.~\eqref{eq:kdvapprox},
and the black vertical lines are the predicted
locations of the trailing and leading edge $s^-t_f$ and 
$s^+t_f$, respectively. The sloped lines are simply
the lines connecting the predicted amplitude to the
location of the trailing edge, which act as an
envelope prediction.

In general, there will be a phase mismatch between the theoretical prediction
and the lattice DSW, the source of which can be explained in several ways. First, there is a small delay for the DSW to form, since the initial condition is smooth. 
Second, the KdV DSW formula is only an asymptotic 
approximation and is missing a phase term (which cannot be captured by 1st order Whitham theory), and finally the KdV equation itself was derived as an approximate model, so even exact solutions will lead to deviation to the lattice solution. 
{We account for phase mismatch, stemming from any of the three possibilities
just described}, in a way similar to that detailed in \cite{Ari2024}, which we summarize here for the reader's convenience. A phase shift $\theta_0(n)$ is applied to the theoretical prediction by determining the best-fit phase at both the trailing and leading edges of the DSW, yielding phase shifts $\theta_\ell$ and $\theta_r$, respectively. The overall phase shift $\theta_0(n)$ is then defined via linear interpolation between these two values. In practice, we used $n_\ell = s^- t_f + 5$ and $n_r = s^+ t_f - 5$. Incorporating the phase in this way leads to good agreement between the KdV prediction and the full numerical DSW profile, as 
seen in Fig.~\ref{fig:compareDSW}. Note that we do not adjust the trailing or leading edges, nor do we modify the envelope predictions—so the solid black lines represent the original theoretical prediction without any empirical phase correction.

We next investigate the edge characteristics of the DSWs as a function
of the jump height $\delta$.
Numerically, to identify the location of the leading edge, we simply find 
the maximum of the profile, which also yields the amplitude of the leading
edge. To identify the location of the trailing edge,
we find the best fit line to the first three local maxima (that is,
near the leading edge) and find when this intersects the
best fit line of the first three local minima. Once that is located
we take a small temporal window around the trailing edge location
and compute the mean at that location.
The speed of the edges (trailing or leading) are estimated by
simply computing the spatial locations as a function 
of time and finding the slope of the best fit line of those points.
The procedure is similar to that detailed in \cite{Su25}. The speeds reported
here will always be assumed to be speeds along the $z$ coordinate.

Figure~\ref{fig:speeds} shows a comparison of the trailing edge speed and mean, and
the leading edge speed and amplitude. This is shown for both simulation (markers)
and KdV theory (lines) for the two example cases: $k_2/k_1=2$ (empty in markers) and $k_2/k_1=2$ (filled in markers).
The KdV prediction has the correct asymptotic behavior, but does worse
in general as the jump height is increased, as expected. It is clear
from the formula for the speeds, amplitude and mean, that they
are (practically) independent of $\theta$. We observe this too for the full 2D simulation,
which is evident by the fact the filled in markers fall into the empty ones.

Thus far, we have considered rational slopes. This allowed us to extract
data from the simulation for better comparison to the theory (i.e. to
generate Fig.~\ref{fig:compareDSW}). The overall dynamics for
irrational slopes, however, is qualitatively similar to the
cases we presented; see Fig.~\ref{fig:irat} for example.
In other words, we observed nothing special about 
rational slopes.

\subsection{ Application of DSW fitting to the 2D FPUT lattice}  \label{sec:DSWfit}

The so-called DSW fitting method is an analytical framework for describing the leading
and trailing edge of dispersive shock waves, proposed by El in \cite{El2005}. The power of the method is that it does not rely {  on}
integrability, and thus {  is} a perfect candidate
for our study. The DSW fitting method is rooted in Whitham modulation theory, which treats the rapidly oscillating DSW region as a slowly varying, modulated wavetrain. Rather than solving the full modulation equations directly, the method involves solving an ODE model
relating the wavenumber $k$ to the mean $\bar{u}$. This ODE has the form
\begin{equation}\label{e:DSWfittingODE}
\frac{dk}{d\bar u} =
    \frac{\partial \omega/\partial \bar u}
{V(\bar{u})-\partial \omega/\partial k}, \qquad k(u^+) = 0,
\end{equation}
where, for example, in the case of the 1D FPUT equation $V(\bar{u}) =\sqrt{\Phi''(\bar u)} $. 
Equation~\eqref{e:DSWfittingODE} is used to describe the trailing edge of the DSW.
 To describe the leading edge of the DSW, a conjugate ODE is used, obtained by replacing $k$ and $\omega$ in Eq.~\eqref{e:DSWfittingODE} by conjugate variables 
$\tilde k$ and $\tilde\omega$ with $k = i\tilde k$ and $\tilde\omega$ 
defined via 
${  \tilde{\omega}(\tilde k,\bar u)} = - i \omega(i\tilde{k},\bar{u})$.
 We refer the reader to~\cite{El2005}  for details.
The DSW fitting method has been used successfully in a large
variety of 1D \cite{Mark2016} and 2D  \cite{Hoefer2017} continuum systems.
More recently the method has been successfully applied in discrete
settings \cite{Sprenger2024,yang2024regularizedcontinuummodeltraveling}.  
Here, we wish to extend the method  and apply it to  our 2D discrete lattice
along the propagation direction of the quasi-1D DSW.

\subsubsection{ General setting}

We will consider more general step initial data  of the form  
\begin{equation}  \label{eq:step4}
 u_{n,m}(0) = \left\{
\begin{split}
u^-, \quad & k_1 n + k_2 m < 0 \\
u^+, \quad & k_1 n + k_2 m > 0,
\end{split}
\right.
\end{equation}
%

It will be convenient to write $ k_1 = r\cos\theta, k_2 = r\sin\theta$ as before.
Note that the above initial condition is only a function of the similarity variable $\xi = \@k\cdot\@n$, where
$\@k = (k_1,k_2)$ and $\@n = (n,m)$, and is independent of the transverse variable $\eta = \@k^\perp\cdot\@n$,
where $\@k^\perp = (k_2,-k_1)^T$.
In what follows, we make the assumption that the same property remains true at all times $t>0$.
As Figs.~\ref{fig:DSWdensity} and~\ref{fig:irat} show,
this assumption is very well supported by the results of the numerical simulations.
(Note $\xi$ varies along the direction $\@k$, perpendicular to the wave fronts, whereas
$\eta$ varies along the direction $\@k^\perp$, parallel to the wave fronts.)
Crucially, this assumption implies that we can effectively reduce the solution of the above initial value problem to a 1D problem concerning propagation along the single variable~$r$, whereas the angle $\theta$ is kept constant. {For that reason, $\theta$ will be treated as a constant
parameter throughout this section.}

In light of the above observations,  linearizing Eq.~\eqref{eq:model} about the constant background $\bar u$  then  gives the dispersion relation
\begin{equation}\label{e:omega2D}
\omega^2(r,\bar u;\theta)  = 4\,\Phi''(\bar u)\, S(r;\theta),
\end{equation}
where
\begin{equation}\label{e:Sdef}
S(r;\theta) : = \sin^2\Big(\frac{r\cos\theta}{2}\Big) + \sin^2\Big(\frac{r\sin\theta}{2}\Big).
\end{equation}
Note that~\eqref{e:omega2D} has the same form  as Eq.~\eqref{eq:disp} with $\bar{u}=0$.
Importantly, the above assumption means that, in Eq.~\eqref{e:omega2D}, $\theta$ simply plays a role as a constant parameter. We will assume that for
the remainder of this section.

We start by taking the ODE stemming from the DSW fitting method for the 1D FPUT equation
(see Eq.~\eqref{e:DSWfittingODE})
and replacing 
the wavenumber $k$ with the wavenumber magnitude $r$. 
{In this set-up, the wavenumber
magnitude $r$ is treated as function 
of $\bar{u}$ with parameter $\theta$,
namely, $r=r(\bar u;\theta)$. Including $\theta$ as a parameter will allow us
to capture angle dependent effects (otherwise, $r$
would be independent of $\theta$).
}
Next, we replace the 1D dispersion relation with 
the 2D dispersion relation~\eqref{e:omega2D}, 
leading to the following initial value problems
\begin{equation}
\label{eq:drdubar_def}
\frac{\partial r}{\partial \bar u}
=
\frac{\partial \omega/\partial \bar u}
{\sqrt{\Phi''(\bar u)}-\partial \omega/\partial r}, \qquad r(u^+;\theta) = 0
\end{equation}
\begin{equation}
\label{eq:dtilder_dubar_def}
\frac{\partial\tilde r}{\partial\bar u}
=
\frac{\partial \tilde\omega / \partial \bar u}
{\sqrt{\Phi''(\bar u)}-\partial \tilde\omega / \partial \tilde r},
\qquad
\tilde r(u^-;\theta) = 0 .
\end{equation}
where the conjugate dispersion relation is given by
$$
{  \tilde{\omega}(\tilde r,\bar u;\theta)} = - i \omega(i\tilde{r},\bar{u};\theta). 
$$
 Explicitly, we have 
\begin{equation*}
\omega(r,\bar u;\theta)
=
2\sqrt{\Phi''(\bar u)}\, \sqrt{S(r;\theta)}, \qquad 
{  \tilde{\omega}(\tilde r,\bar u;\theta)} = 2\sqrt{\Phi''(\bar u)}\,\sqrt{\smash{\tilde S}(\tilde r;\theta)},
\end{equation*}
 where $S(r;\theta)$ is given by~\eqref{e:Sdef}, and where for convenience we also defined  
\begin{equation*}
\tilde S(\tilde r;\theta) = \sinh^2\!\Big(\frac{\tilde r\cos\theta}{2}\Big) + \sinh^2\!\Big(\frac{\tilde r\sin\theta}{2}\Big).
\end{equation*}

Computing the relevant partial derivatives in Eq.~\eqref{eq:drdubar_def}
and Eq.~\eqref{eq:dtilder_dubar_def} and simplifying yields,
\bse
\begin{equation}\label{eq:2Dfit}
\frac{\partial r}{ \partial \bar u}
=
\frac{\Phi'''(\bar u)}{\Phi''(\bar u)}
\,
\frac{S(r;\theta)}
{\,\sqrt{S(r;\theta)}-T(r;\theta)}, \qquad r(u^+;\theta) = 0,
\end{equation}
\begin{equation} \label{eq:2Dfit2}
\frac{\partial\tilde r}{d\bar u}
=
\frac{\Phi'''(\bar u)}{\Phi''(\bar u)}
\,
\frac{\tilde S(\tilde r;\theta)}
{\sqrt{\tilde S(\tilde r;\theta)}-\tilde T(\tilde r;\theta)}, \qquad \tilde{r}(u^-;\theta) = 0,
\end{equation}
\ese
where
\begin{gather*}
T(r;\theta) =
    \cos\theta\,
        \sin\!\Big(\tfrac12{r\cos\theta}\Big)
        \cos\!\Big(\tfrac12{r\cos\theta}\Big)
    + \sin\theta\,
        \sin\!\Big(\tfrac12{r\sin\theta}\Big)
        \cos\!\Big(\tfrac12{r\sin\theta}\Big), \\
\tilde T(\tilde r;\theta) =
    \cos\theta\,
        \sinh\!\Big(\tfrac12{\tilde r\cos\theta}\Big)
        \cosh\!\Big(\tfrac12{\tilde r\cos\theta}\Big)
    + \sin\theta\,
        \sinh\!\Big(\tfrac12{\tilde r\sin\theta}\Big)
        \cosh\!\Big(\tfrac12{\tilde r\sin\theta}\Big).
\end{gather*}
We were unable to obtain an explicit solution of the above initial value problem for general values of~$\theta$.
For special values of $\theta$, however,  these problems  can be solved  in closed form, as we show next. 

\subsubsection{The $\theta=0$ case}

We first consider $\theta=0$, which corresponds to the
1D situation. In this case, we have that $S(r;0) = \sin^2(r/2)$,
$T(r;0) =\sin(r/2)\cos(r/2) $ and $\tilde{S}(\tilde{r};0) = \sinh^2(\tilde{r}/2)$, 
$\tilde{T}(\tilde{r};0) =\sinh(\tilde{r}/2)\cosh(\tilde{r}/2)$.
Plugging this into Eq.~\eqref{eq:2Dfit} and Eq.~\eqref{eq:2Dfit2} yields

\begin{eqnarray}
\frac{\partial r}{\partial\bar u}
&=&
\frac{\Phi'''(\bar u)}{\Phi''(\bar u)}
\,
\frac{\sin(\tfrac r2)}{1-\cos(\tfrac r2)}
=
\frac{\Phi'''(\bar u)}{\Phi''(\bar u)}
\,
\cot\!\Big(\tfrac r4\Big), \qquad r(u^+;0) = 0 \\
\frac{\partial \tilde r}{\partial\bar u}
&=&
\frac{\Phi'''(\bar u)}{\Phi''(\bar u)}
\,
\frac{\sinh(\tfrac{\tilde r}{2})}{1-\cosh(\tfrac{\tilde r}{2})}
=
-\frac{\Phi'''(\bar u)}{\Phi''(\bar u)}
\,
\coth\!\Big(\tfrac{\tilde r}{4}\Big), \qquad \tilde r(u^-;0) = 0
\end{eqnarray}
These equations can be solved explicitly. 
The solutions are
\begin{equation}
r(\bar u;0)
=
4\,
\arccos\!\Bigg[
\left(\frac{\Phi''( u^+)}{\Phi''(\bar u)}\right)^{1/4}
\Bigg].
\label{eq:diag_explicit}
\end{equation}
and for the conjugate problem, 
\begin{equation}
\tilde r(\bar u;0)
=
4\,
\mathrm{arccosh}\!\Bigg[
\left(\frac{\Phi''(u^-)}{\Phi''(\bar u)}\right)^{1/4}
\Bigg].
\label{eq:diag_conjugate}
\end{equation}
{
Evaluating  $r(\bar{u};0)$ at the trailing edge (where $\bar{u}=u^-$) 
will give us an expression for the trailing edge wavenumber.
Since the expression will depend on $\theta$ we denote the trailing edge wavenumber
$r^-(\theta)$ for general $\theta$, and for the
particular choice of $\theta=0$ we simply write $r^-(0)$. In particular,
we have
\begin{equation}
r^-(0) := r(u^-;0)
=
4\arccos\!\left[
\left(\frac{\Phi''(u^+)}{\Phi''(u^-)}\right)^{1/4}
\right].
\label{eq:k-} 
\end{equation}
Note this formula coincides exactly with formula derived in the 1D case \cite{Su26},
as expected.
Likewise, we evaluate $\tilde{r}(\bar{u};0)$ at the leading edge (where $\bar{u}=u^+$)
to obtain
\begin{equation}
r^+(0) := \tilde{r}(u^+ ;0)
=
4\mathrm{arccosh}\!\left[
\left(\frac{\Phi''(u^-)}{\Phi''(u^+)}\right)^{1/4}
\right].
\label{eq:k+} 
\end{equation}
}
The velocity at the trailing edge (which we call $s^-(\theta)$)
and the velocity at the leading edge (which we call $s^+(\theta)$) are 
$$ s^-(0) = \partial_r\omega(r^-{ (0)},u^-;0) =\sqrt{\Phi''(u^-)}\cos(r^-(0)/2) $$

$$ s^+(0)  = \frac{\tilde{\omega}(r^+(0),u^+;0)}{r^+(0)} = 2 \sqrt{\Phi''(u^+)}\frac{\sinh(r^+(0)/2)}{r^+(0)} $$
%
If we pick $u^+ = 0$, $u^- = \delta \frac{\sqrt{a_1}}{a_2}$ 
(so that the initial conditions in Eq.~\eqref{eq:Astep} correspond to Eq.~\eqref{eq:step2} )
and we pick the potential as in Eq.~\eqref{eq:def_phi} and expand $r^-(0)$, $s^-(0)$ and $s^+(0)$ about $\delta = 0$ we obtain
$$r^-(0)\approx 4 \sqrt{\frac{\delta}{\sqrt{a_1}}} $$ 
$$ s^-(0) \approx \sqrt{a_1} - \delta$$
$$ s^+(0) \approx \sqrt{a_1} + \frac{2 }{3}\delta$$
%
%
which are exactly the KdV predictions. In particular,
see the expression for $k^-$ in Eq.~\eqref{eq:wavenumber}
with $\theta=0$ and the expression for $s^+$ and $s^-$ in Eq.~\eqref{eq:speeds}
(which are independent of $\theta$).
In other words, the DSW fitting
prediction and the KdV prediction match to leading order in the limit of small jump heights ($\delta\rightarrow0$).

\subsubsection{The $\theta=\pi/4$ case}
We now turn our attention to the case $\theta=\pi/4$, which corresponds
to some of the simulations we conducted (with $k_2/k_1=1$).
Since $\cos\theta=\sin\theta=\frac{1}{\sqrt{2}}$ we have the following
simplifications
\begin{align*}
S(r;\pi/4)
&=
2\sin^2\!\left(\frac{r}{2\sqrt{2}}\right), \qquad
T(r;\pi/4)
=
\sqrt{2}
\sin\!\left(\frac{r}{2\sqrt{2}}\right)\cos\!\left(\frac{r}{2\sqrt{2}}\right). \\
\tilde{S}(\tilde{r};\pi/4)
&=
2\sinh^2\!\left(\frac{\tilde{r}}{2\sqrt{2}}\right), \qquad
\tilde{T}(\tilde{r};\pi/4)
=
\sqrt{2}
\sinh\!\left(\frac{\tilde{r}}{2\sqrt{2}}\right)\cosh\!\left(\frac{\tilde{r}}{2\sqrt{2}}\right). 
\end{align*}
%
Substituting these expressions into Eq.~\eqref{eq:2Dfit}
and Eq.~\eqref{eq:2Dfit2}, we obtain 
\begin{equation}
\frac{\partial r}{\partial \bar u}
=
\sqrt{2}\,
\frac{\Phi'''(\bar u)}{\Phi''(\bar u)}
\cot\!\Big(\frac{r}{4\sqrt{2}}\Big),
\qquad
r(u^+;\pi/4)=0.
\label{eq:diag_fit}
\end{equation}
and 
\begin{equation}
\frac{\partial \tilde r}{\partial \bar u}
=
-\sqrt{2}\,
\frac{\Phi'''(\bar u)}{\Phi''(\bar u)}
\coth\!\Big(\frac{\tilde r}{4\sqrt{2}}\Big),
\qquad
\tilde r(u^-;\pi/4)=0.
\end{equation}
These ODEs  are identical to the ones derived in the $\theta=0$ case
after a change of variable $r\rightarrow r/\sqrt{2}$  and
$\tilde{r}\rightarrow \tilde{r}/\sqrt{2}$. Thus, we have
that 

\begin{equation} \label{eq:rmpi4}
r^-(\pi/4) := r(u^-;\pi/4)
=
4\sqrt{2}\arccos\!\left[
\left(\frac{\Phi''(u^+)}{\Phi''(u^-)}\right)^{1/4}
\right],
\end{equation}
\begin{equation} \label{eq:rppi4}
r^+(\pi/4) := \tilde{r}(u^+;\pi/4)
=
4\sqrt{2}\mathrm{arccosh}\!\left[
\left(\frac{\Phi''(u^-)}{\Phi''(u^+)}\right)^{1/4}
\right].
\end{equation}
The velocity at the trailing edge ($s^-(\theta)$)
and leading edge ($s^+(\theta)$) are
$$ s-+(\pi/4)  =\sqrt{\Phi''(u^-)}\cos\left(\frac{r^-}{2\sqrt{2}}\right) $$

$$ s^+(\pi/4)  =2 \sqrt{2} \sqrt{\Phi''(u^+)}\frac{\sinh\left(\frac{r^+(\pi/4)}{2\sqrt{2}}\right)}{r^+(\pi/4)} $$

Substituting the expression from Eq.~\eqref{eq:rmpi4} and Eq.~\eqref{eq:rppi4} into
the above expression we see
that the factors of $\sqrt{2}$ cancel, leading to the conclusion that 

$$ s^+(\pi/4)  = s^+(0) $$ 

$$ s^-(\pi/4)  = s^-(0) $$ 

In other words, the DSW fitting method also predicts the leading and trailing edge
speeds are independent of the angle $\theta$.  The dashed lines of 
Fig.~\ref{fig:speeds}(a) show the prediction of the leading and trailing edge
speeds, comparing favorably to the speeds extracted from the full 2D simulations.
In the figure, one can see that the KdV prediction (solid lines) and
DSW fitting prediction (dashed lines) agree as $\delta\rightarrow 0$.
For both the leading and trailing edge speed, the DSW fitting prediction
is more accurate than the KdV prediction.

For the trailing edge wavenumber, we have that

$$ r^-(\pi/4)  = \sqrt{2} \, r^-(0) $$ 
In other words, the trailing edge wavenumber $r^-$ is indeed dependent on $\theta$.
The expansion of $r^-(\pi/4)$ is
$$ r^-(\pi/4) \approx  4\sqrt{2}  \sqrt{\frac{\delta}{\sqrt{a_1}}} $$ 
which once again coincides with the KdV approximation of the trailing edge wavenumber for $\theta = \pi/4$.

\subsubsection{The general $\theta$ case}

The above calculations suggest that the prediction of the leading and trailing 
edge speeds are independent of $\theta$ (consistent with the KdV prediction, but also numerical observations). The trailing edge wavenumber, on the other hand, was dependent
on $\theta$.  For general values of $\theta$ there is no explicit solution to Eq.~\eqref{eq:2Dfit}, and so to predict the trailing edge wavenumber from the DSW
fitting approach, we must solve Eq.~\eqref{eq:2Dfit} numerically. We did this
using the \url{ode45} routine in Matlab. Recall that the initial value is
$r(u^+) = 0$ and for the simulations we have $u^+=0$. However, the
right-hand side of the ODE is undefined at $r=0$. Thus, for the simulation
we use $r(\bar{u}=0;\theta) = 10^{-6}$ { (using
$r(\bar{u}=0;\theta) = 10^{-5}$ and $r(\bar{u}=0;\theta) = 10^{-7}$
produced similar results)}.
We simulated until $\bar{u} = u^- = \delta\sqrt{a_1}/a_2$,
and extracted the final $r$ value as the prediction for the trailing edge
wavenumber magnitude. The simulations produced the correct result
for the $\theta=0$ and $\theta = \pi/4$ cases, and so we used $r(\bar{u}=0;\theta) = 10^{-6}$
for simulations with other values of $\theta$ as well. 

A plot of
$r(\bar{u}=\delta;\theta)$ is shown in Fig.~\ref{fig:wavenumbers} (dashed line)
for various $\theta$. For comparison, we included the KdV prediction as well:
$$4\sqrt{\frac{\delta}{\sqrt{a_1}(\cos^4(\theta) + \sin^4(\theta))}} $$
shown as the solid line. The trailing edge wavenumber extracted from the 2D
simulation is shown as markers.

To numerically estimate the wavenumber magnitude $r$, we used the following procedure. We first
estimated the trailing edge location (in the same way detailed in
the previous section). Suppose the location of the trailing edge
was for $n=n_0$ and $m=m_0$.  Call this data
$$u_0(t) = u_{n_0,m_0}(t).$$
We then found the distance (in time) between the last two local maxima
of $u_0(t)$, which serves as an estimate for the period $T$.
Note that the local maxima were found by first identifying the maximum
values of the signal, and then performing a quadratic fit for three neighboring
points. We used that quadratic fit to then obtain a better estimate of the local
maximum. Then we computed $\omega = 2\pi/T$. If we assume
that in a small spatial-temporal window  the DSW is exactly a periodic wave,
 we have that 
$$ u_{n_0,m_0}(t) = U(r\cos(\theta) n_0 +r\sin(\theta) m_0 - \omega t)$$
where we just discussed how to compute $\omega$ and we seek an estimate
of $r$. To do this, we then extract data at the neighboring point, namely
$$u_1(t) = u_{n_0+1,m_0+1}(t).$$
We then compute the distance (in time) of the last local max of $u_0(t)$ and $u_1(t)$,
call this number $\ell$. Under the assumption of the DSW being a perfect periodic
wave (in the small space-time window), we have that 
\begin{eqnarray*}
 u_{n_0+1,m_0+1}(t) &=& U(r\cos(\theta) (n_0+1) +r\sin(\theta)(m_0+1) - \omega t)   \\
 &=&U(r\cos(\theta)n_0 +r\sin(\theta)m_0 -  \omega t + r(\cos(\theta)+ \sin(\theta)) ) 
\end{eqnarray*}
An observed phase shift between $u_0$ and $u_1$ will be of the form $ r(\cos(\theta)+ \sin(\theta))/\omega$.
Thus, we have that
$$r = \frac{\ell \omega}{\cos(\theta) + \sin(\theta)} $$
where $\ell$ is the apparent (numerically observed) phase shift between $u_1$ and $u_0$ and
$\omega$ is the angular frequency of $u_0$. This is how
$r$ was estimated numerically for Fig.~\ref{fig:wavenumbers}.
{Note, we chose the second site $u_1=u_{n_0+1,m_0+1}$ to estimate
 $\ell$. For a perfect periodic wave the choice of $u_1$ does not matter, but for slowly modulated waves, the choice
could lead to small differences. We repeated the above
calculation with $u_1 =u_{n_0,m_0+1}$ and $u_1=u_{n_0+1,m_0}$ (and adjusting
the formula for $\ell$ accordingly)
which led to only small differences (on the order of $10^{-3}$).}

\begin{figure} 
    \centerline{
 \includegraphics[height=7cm]{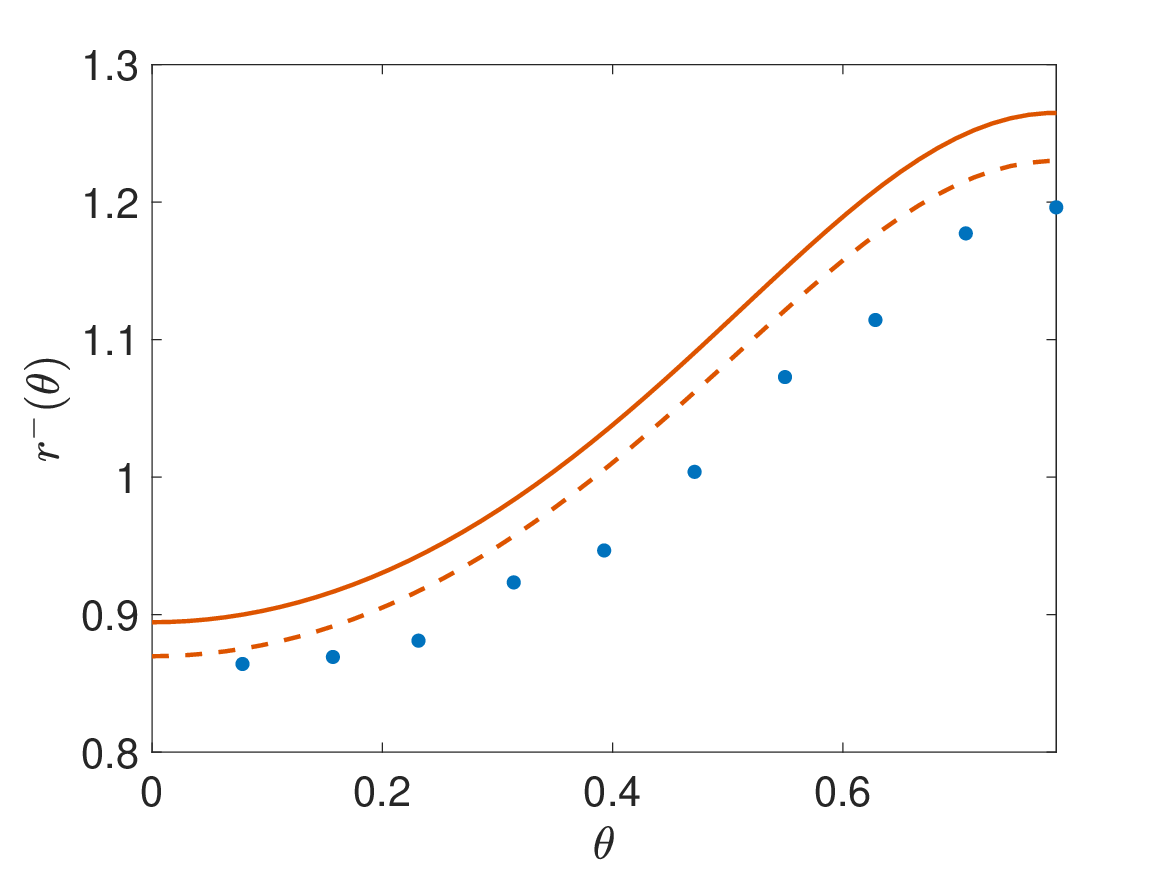}  
  }
   \caption{Prediction of the trailing edge wavenumber $r^-$ as a function
   of $\theta$ from the DSW fitting method (dashed line), KdV approximation
   (solid line) and from the full 2D simulation (markers) for $\delta = 0.05$.
   }
   \label{fig:wavenumbers}
\end{figure}

\section{Conclusions}
\label{s:conclusions}

In this work, we have investigated traveling waves and dispersive shock waves in a two-dimensional FPUT lattice. By employing a variational framework, we established the existence of both periodic and solitary traveling waves for convex interaction potentials and demonstrated that the convexity assumption can be removed for unimodal profiles. This approach also yielded an effective numerical algorithm for computing traveling wave solutions, which enabled a systematic exploration of both solitary and periodic waves. The resulting numerical computations were compared with analytical approximations derived from the Korteweg–de Vries (KdV) equation, showing consistency between the two
approaches.

We then focused on the formation and evolution of dispersive shock waves. In particular, we examined line DSWs generated by jump initial data, characterized by invariance in one spatial direction and propagation in the orthogonal direction. Our results indicate that while the detailed shape of these waves depends on their direction of propagation, key quantities such as speed and amplitude remain independent of direction. A comprehensive numerical study of line DSWs was carried out and systematically compared with KdV-based predictions over a range of jump heights. In the small-amplitude regime, our numerical results exhibit good agreement with KdV theory. To further refine these predictions, we applied the DSW fitting method to characterize the leading and trailing edges, obtaining close agreement between the theoretical predictions and the numerical simulations. Both the KdV approximation and the DSW fitting approach were shown to converge in the limit of vanishing jump height.

We close this section by briefly discussing several natural directions for further work suggested by the present study.

\textit{Stability of traveling waves and line DSWs.}
The variational existence results of Section~\ref{sec:TWs} are not accompanied by stability results, and the line DSWs of Section~\ref{sec:DSWs} were analyzed only along the propagation direction. In closely related continuum settings, line solutions are well known to be susceptible to transverse modulational instabilities.
For example, this is reflected in the KP-I/KP-II dichotomy and in the transverse stability analyses based on multi-dimensional Whitham theory carried out for the KP~\cite{Ablowitz2017}, 2D Benjamin--Ono~\cite{PRE2017v96p032225}, 2D NLS~\cite{2DNLS_mod2022,2DNLS_mod2023}, and 2D Zakharov--Kuznetsov~\cite{2DZK_mod2023} equations. Whether the line traveling waves and line DSWs constructed here are spectrally and dynamically stable to transverse perturbations of the 2D lattice equation~\eqref{eq:model} is an open and, in our view, central question. The variational machinery developed in Section~\ref{sec:TWs} appears to be a natural starting point for such a stability analysis.

\textit{Genuinely 2D DSW geometries.}
By construction, the line DSWs considered here are essentially one-dimensional objects embedded in the 2D lattice. The most immediate genuinely-2D extension is the case of radially symmetric (cylindrical) jump initial data, the continuum analog of which is naturally tied to the cylindrical KdV reduction of the KP equation~\cite{PHYSD333p84} and the modified KdV variant~\cite{2DmKdV_mod2023}. A radial DSW in the present 2D FPUT lattice would test whether the $\theta$-dependent features identified here (e.g.~the trailing-edge wavenumber in Fig.~\ref{fig:wavenumbers}) integrate correctly as the propagation direction varies continuously around an expanding front.  Another interesting direction would be the study of wedge-type initial conditions similarly to what was recently done for the KP equation in \cite{PRL2025v135p067201}.

\textit{Oblique DSW interactions, Mach reflection, and corner data.}
Beyond radial geometries, a rich set of multidimensional phenomena---oblique soliton interactions, Mach expansion of bent solitons, and reflections of DSWs at corners---has been studied in the KP setting using reductions of the 2D Whitham equations~\cite{JFM2021v909pA24,NLTY2021v34p3583,KPII_Mach2022}. Each of these has a natural analog in the 2D FPUT lattice: two line DSWs propagating at different angles, a line DSW initialized with a kink or corner in its interface, and the interaction of an oblique line soliton (cf.~Fig.~\ref{fig:solitary}) with a line DSW. The numerical and asymptotic methods of Section~\ref{sec:DSWs} should extend to these settings essentially without modification.

\textit{Full 2D Whitham modulation theory for 2D lattices.}
The KdV-based approximation employed in Section~\ref{sec:DSWs} captures the small-amplitude regime accurately but, as noted around Fig.~\ref{fig:compareDSW}, misses an $O(1)$ phase term that first-order Whitham theory cannot supply. A natural sharpening of the present work would be the derivation of the \emph{full} Whitham modulation system for the 2D lattice~\eqref{eq:model}, in the spirit of what was done for KP in~\cite{Ablowitz2017} and for the 1D FPUT and Toda lattices in~\cite{Venakides99,DHM06,blochkodama,Gino2024}. Such a system would in principle yield angle-dependent characteristic speeds rather than the degenerate $\theta$-independent leading-order speeds obtained here, and could be tested against the numerical edge speeds reported in Fig.~\ref{fig:speeds}.

\textit{Integrable discrete models in 2D.}
In the 1D setting, integrable approximations based on the Toda lattice have proven to be powerful quantitative tools for describing DSWs, capturing discreteness effects that elude continuum reductions such as KdV~\cite{Ari2024,Gino2024}. A natural 2D analog is the 2D Toda lattice, whose soliton structure has been investigated in~\cite{jpa2010}.  A quantitative study of 2D DSWs in the 2D Toda lattice---paralleling the 1D programs of~\cite{Ari2024,Gino2024} and providing a non-perturbative point of comparison for the present results---would be a worthwhile project. The discrete KP hierarchy provides a further integrable structure whose relevance to DSWs in lattices of FPUT type appears, to the best of our knowledge, not to have been explored.

\textit{Higher-order and quasi-continuum models.}
A complementary route to capturing discreteness effects that elude the leading-order KdV description is the construction of higher-order continuum or quasi-continuum models, as recently pursued in the 1D context in~\cite{Sprenger2024,yang2024regularizedcontinuummodeltraveling,CHONG2022133533}. Extending such reductions to the 2D lattice is a natural next step and is expected to improve agreement at moderate jump heights, where Fig.~\ref{fig:speeds} shows the KdV prediction deteriorating relative to the DSW fitting prediction.

\textit{Other lattice geometries and vector lattices.}
The square geometry of~\eqref{eq:model} is a natural starting point, but is by no means generic. Triangular and hexagonal lattices, with their differing coordination numbers and anisotropic dispersion relations, are expected to give rise to qualitatively different angular dependences of the DSW characteristics. More broadly, many physical 2D lattices---mechanical, granular, photonic---are vector-valued, with multiple degrees of freedom per site. The variational framework of Section~\ref{sec:TWs} and the DSW fitting framework of Section~\ref{sec:DSWfit} should both extend, with appropriate modifications, to these settings.

\textit{Irrational propagation directions.}
Finally, our quantitative comparisons between simulation and theory were carried out for rational slopes $k_2/k_1$, which permit clean extraction of one-dimensional profiles along the propagation direction. The simulations at irrational slopes shown in Fig.~\ref{fig:irat} are qualitatively similar to their rational counterparts, but a careful analysis of how the lattice's quasi-periodic sampling along irrational directions affects the small-amplitude tails of the DSW, the structure of the trailing oscillations, and the validity of the long-wave approximation has not been attempted here. Such an analysis, possibly with number-theoretic flavor, may reveal subtle effects that are invisible in the qualitative picture of Fig.~\ref{fig:irat}.

 In conclusion, we hope that the present study will provide a useful  analytical and numerical framework for understanding traveling waves and dispersive shock dynamics in scalar two-dimensional FPUT lattices. Our results highlight the effectiveness of combining variational methods, asymptotic analysis, and numerical simulations, and they lay the groundwork for future investigations of more general nonlinear wave phenomena.

\section*{Acknowledgments}
This paper is dedicated to the late Michael Herrmann, whose insightful discussions on periodic waves and dispersive shock waves were invaluable to the authors. Michael, our esteemed collaborator and dear friend, continues to be missed.

C.C. gratefully acknowledges financial support for this project by the Fulbright U.S. Scholars Program, which is sponsored by the U.S. Department of State and German-American Fulbright Commission. Its contents are solely the responsibility of the authors and do not necessarily represent the official views of the Fulbright Program, the Government of the United States, or the German-American Fulbright Commission.

This material is also based upon work supported by the US National Science Foundation under Grant DMS-2107945 (C.C.) and  PHY-2408988 (PGK). 

W.R. acknowledges funding by the Deutsche Forschungsgemeinschaft (DFG, German Research Foundation) – Project-ID 258734477 – SFB 1173.

This research was partly conducted while P.G.K. was  visiting the Okinawa Institute of Science and
Technology (OIST) through the Theoretical Sciences Visiting Program (TSVP), the University of
Sydney through the visitor program of the Sydney Mathematical Research Institute (SMRI) and the Department of Mechanical Engineering at Seoul National 
University through a Fulbright Fellowship. Their support is gratefully acknowledged.
Finally, this work was also  supported by a grant from the Simons Foundation [SFI-MPS-SFM-00011048, P.G.K]. 

\section*{Data Availability}
The data and code that support the findings of this study are available from the corresponding author upon reasonable request. All code was written exclusively by Christopher Chong.

\section*{Appendix} \label{sec:appendix}
{ 
In order to handle the existence of solitary waves in the case where $a_1>0$ and $I=\R$ we rely on some arguments from \cite{herrmann_peridynamics} and adapt them to our case. If we denote by 
$$
\Phi_0(u) = \frac{a_1}{2} u^2
$$
the quadratic part of $\Phi$ then next to the functional 
$$
J(z)= \int_\R \Phi(-Az)\,d\rho
$$
we can also consider the purely quadratic functional 
$$
J_0(z) = \int_\R \Phi_0(-Az)\,d\rho.
$$
Further let 
$$
j(R) = \sup_{C} J(z), \quad  j_0(R) = \sup_{C} J_0(z) \quad \mbox{ where } \quad C:= \{z\in L^2(I)\times L^2(I): \|z\|_2 \leq R\}.
$$

The first lemma concerns the relation between $j(R)$ and $j_0(R)$. It is a consequence of the potential being strictly superquadratic.

\begin{lemma}
    Let $a_1, a_2>0$ and $a_3 \geq 0$. Then for every $R>0$ we have $j_0(R)<j(R)$.
\end{lemma}

\begin{proof}
    Let us first give an upper bound for $j_0(R)$. Consider the interval $I_L := [-\frac{L}{2}, \frac{L}{2}]$ and its characteristic function $\chi_L$. Recall that the Fourier transform of $\chi_L$ is given by $L\sinc(L\xi)$ with $\sinc(s)=\frac{\sin \pi s}{\pi s}$ and where $|\sinc(s)|\leq 1$. Then, using $\chi_{k_i} = \chi_{[-\frac{k_i}{2},\frac{k_i}{2}]}$ for $i=1,2$ together with Plancherel's theorem and $(k_1,k_2)=r(\cos\theta,\sin\theta)\in (0,\pi)\times (0,\pi)$, we get 
    \begin{align*}
        \int_\R (-Az)^2\,dx & = \int_\R (\chi_{k_1}\ast z_1+\chi_{k_2}\ast z_2)^2\,dx \\
        &= \int_\R (\hat\chi_{k_1} \hat z_1 + \hat\chi_{k_2}\hat z_2)^2\,d\xi \\
        &= \int_\R (k_1 \sinc(k_1\xi)\hat z_1(\xi) + k_2 \sinc(k_2\xi) \hat z_2(\xi))^2 \,d\xi \\
        & \leq \int_\R (k_1 |\hat z_1| + k_2 |\hat z_2)^2 \,d\xi \\
        & = \int_\R k_1^2 |\hat z_1|^2 +k_2^2 |\hat z_2|^2 + 2k_1k_2 |\hat z_1| |\hat z_2|\,d\xi \\
        & \leq \int_\R k_1^2\left(1+ \frac{\sin^2\theta}{\cos^2\theta}\right) |\hat z_1|^2 + k_2^2\left(1+\frac{\cos^2\theta}{\sin^2\theta}\right) |\hat z_2|^2 \,d\xi \\
        & = r^2 R^2.
    \end{align*}
    The fact that this upper bound is actually sharp is a consequence of the following computation. Let $z_1=\frac{R\cos\theta}{\sqrt{L}} \chi_L$ and $z_2 = \frac{R\sin\theta}{\sqrt{L}} \chi_L$. Then 
    $$
    \int_{\rho-\frac{k_i}{2}}^{\rho+\frac{k_i}{2}} \chi_L(s)\,ds = \left\{ \begin{array}{ll}  k_i & \mbox{ if } |\rho|<\frac{L-k_i}{2}, \\ 
    0 & \mbox{ if } |\rho|>\frac{L+k_i}{2}.\end{array} \right.
    $$
    Therefore, with $k=\max\{k_1,k_2\}$,
    \begin{align*}
        J_0(z)=\int_\R (-Az)^2\,dx & \geq \left(\frac{k_1R\cos\theta}{\sqrt{L}}+\frac{k_2R\cos\theta}{\sqrt{L}}\right)^2(L-k) = r^2R^2 - O(L^{-1})
    \end{align*}
     as $L\to \infty$ and hence $j_0(R)=r^2R^2$. Similarly for $l\in \N, l\geq 3$, $\int_\R (-A z)^l\,dx \geq O( L^{1-\frac{l}{2}})$ as $L\to \infty$. This leads to the conclusion that 
     $$
     j(R)-j_0(R) = j(R)-J_0(z)-O(L^{-1}) \geq J(z)-J_0(z)-O(L^{-1}) \geq a_2 O(L^{-1/2}) + a_3 O(L^{-1})- O(L^{-1})>0 
     $$
     for large enough $L$ since $a_2>0$. Note that the estimate (in its current form) also shows that $a_2=0$ cannot easily be compensated for by assuming $a_3>0$, unless further refined estimates are produced.
\end{proof}

Based on this lemma we can now show the main result.

\begin{lemma}
     Let $a_1, a_2>0$ and $a_3 \geq 0$. Then $j(R)$ is attained by a pair of unimodal functions.
\end{lemma}

\begin{proof} For a maximizing sequence $(z_j)_{j\in \N}$ we can assume by Lemma~\ref{lem:prop_op_2} and Lemma~\ref{lem:rearrange} that both $z_j^1$ and $z_j^2$ are unimodal, that $z_j\stackrel{n\to \infty}{\rightharpoonup} z_\infty$ weakly in $L^2(\R)\times L^2(\R)$ and $Az_j\stackrel{j\to\infty}{\to} Az_\infty$ strongly in $L^p(\R)\times L^p(\R)$ for $p\in (2,\infty]$. In order to show that $J(z_j)\to  J(z_\infty)=\max_C J$ is suffices to show that $z_j\to z_\infty$ strongly in $L^2(\R)\times L^2(\R)$. Since $z_j^1, z_j^2$ are both unimodal they satisfy 
\begin{equation}  \label{eq:pointwise}
0\leq z_j^i(x)\leq \frac{R}{\sqrt{|x|}} \mbox{ for } i=1,2 \mbox{ and all } j\in\N.    
\end{equation} 
Furthermore $\|z_j\|_{L^2}, \|z_\infty\|_{L^2}\leq R$. The proof will be finished once we have shown that $\|z_\infty\|_{L^2}=R$. Next, for any $X>0$, we define the splitting 
$$
z_j= \widetilde z_j + \overline{z}_j \quad \widetilde z_j := \chi_{[-X,X]}z_j, \quad \overline{z}_j := \chi_{\R\setminus [-X,X]} z_j.
$$
Due to the splitting we have that $\widetilde z_j\stackrel{j\to \infty}{\rightharpoonup} \widetilde z_\infty:= \chi_{[-X,X]} z_\infty$ and that $A\widetilde z_j \to A\widetilde z_\infty$ uniformly as $j\to \infty$, cf. Lemma~\ref{lem:prop_op_2}(iv). Since $A\tilde z_j$ is supported inside $[-X-k,X+k]$ we get that $A\widetilde z_j\to A\widetilde z$ in $L^p(\R)$ as $j\to \infty$ for any $p\in [1,\infty]$. We can therefore write 
$$
J(\widetilde z_\infty)=J(\widetilde z_j)+ o_j(1)
$$
where $o_j(1)\to 0$ as $j\to \infty$ when $X$ is kept fixed. Next we split the real line into three sets 
$$
\R = \bigcup_{i=1}^3 I_i(X), \quad \mbox{ with } I_1(X)=\left\{x\in\R: |x| \leq X-\frac{k}{2}\right\}, \quad I_2(X) = \left\{x\in \R: \bigl||x|-X\bigr| \leq \frac{k}{2}\right\}, \quad I_3(X)=\left\{x\in \R: |x|\geq  X+\frac{k}{2}\right\}.
$$
We will now determine estimates for $\int_{I_i(X)} \Phi(-Az_j)\,dx$ for $i=1, 2, 3$. First, using \eqref{eq:pointwise} we get that 
$$
\int_{I_2(X)} \Phi(-Az_j)\,dx \leq o_X(1) \mbox{ uniformly in $j\in \N$ as $X\to \infty$.}
$$
For $x\in I_1(X)$ we have $-X\leq x-\frac{k}{2} \leq x+ \frac{k}{2}\leq X$, i.e., $Az_j = A\widetilde z_j$ since $A\overline{z}_j=0$. Therefore 
$$
\int_{I_1(X)} \Phi(-Az_j)\,dx = \int_{I_1(X)} \Phi(-A\widetilde z_j)\,dx \leq \int_\R \Phi(-A\widetilde z_j)\,dx
$$
Finally, for $x\in I_3(X)$ we have $x-\frac{k}{2}>X$ if $x>0$ and $x+\frac{k}{2}<-X$ if $x<0$. In any case we have $A z_j= A\overline{z}_j$ since this time $A\widetilde z_j=0$. Therefore
$$
\int_{I_3(X)} \Phi(-Az_j)\,dx = \int_{I_3(X)} \Phi(-A\overline{z}_j)\,dx\leq \int_{\R} \Phi(-A\overline{z}_j)\,dx = \int_\R \Phi_0(-A\overline{z}_j)\,dx + o_X(1)
$$
uniformly in $j$ since $|A\overline{z}_j|\leq \frac{2kR}{\sqrt{|X|-k/2}}$. Altogether this shows that 
\begin{equation} \label{eq:altogether}
J(z_j) \leq J(\widetilde z_j)+J_0(\overline{z}_j)+ o_X(1) = J(\widetilde z_\infty)+ J_0(\overline{z}_j) + o_j(1)+o_X(1).
\end{equation}
Furthermore, since $\|\widetilde z_\infty\|_2\leq R$ we see that 
$$
\frac{R^2}{\|\widetilde z_\infty\|_2^2}J(\widetilde z_\infty) \leq J\left(\frac{R}{\|\widetilde z_\infty\|_2}\widetilde z_\infty\right) \leq j(R)
$$
and 
$$
\frac{R^2}{\|\overline{z}_j\|_2^2} J_0(\overline{z}_j) = J_0\left(\frac{R}{\|\overline{z}_j\|_2}\overline{z}_j\right) \leq j_0(R)
$$
so that 
$$
J(\widetilde z_\infty) \leq \frac{\|\widetilde z_\infty\|_2^2}{R^2}j(R), \quad J_0(\overline{z}_j) \leq \frac{\|\overline{z}_j\|_2^2}{R^2} j_0(R) 
$$
where these estimates also hold in case $\widetilde z_\infty=0$ or $\overline{z}_j=0$. If we insert this into \eqref{eq:altogether} we obtain 
$$
j(R)+ o_j(1)= J(z_j) \leq \frac{\|\widetilde z_\infty\|_2^2}{R^2}j(R)+\frac{\|\overline{z}_j\|_2^2}{R^2} j_0(R) + o_j(1)+o_X(1)
$$
and hence (sending $j\to \infty$)  
\begin{equation} \label{eq:final}
j(R) \leq \|\widetilde z_\infty\|_2^2\frac{j(R)}{R^2}+ \limsup_{j\in\N} \|\overline{z}_j\|_2^2 \frac{j_0(R)}{R^2}+o_X(1).
\end{equation}
Since $\|\widetilde z_\infty\|_2^2+ \limsup_{j\in \N} \|\overline{z}_j\|_2^2 \leq R^2$ this implies 
$$
\limsup_{j\in\N}\frac{\|\overline{z}_j\|_2^2}{R^2} (\underbrace{j(R)-j_0(R)}_{>0}) \leq o_X(1)
$$
and then \eqref{eq:final} implies $j(R) \leq \|\tilde z_\infty\|_2^2 \frac{j(R)}{R^2}+o_X(1)$, i.e., $\|z_\infty\|_2\geq R$ after having sent $X\to\infty$. This shows that $\|z_j\|_2\stackrel{j\to\infty}{\rightarrow}R=\|z_\infty\|_2$ which together with $z_j\stackrel{j\to\infty}{\rightharpoonup}z_\infty$ finishes the proof that $z_j\stackrel{j\to\infty}{\to} z$ strongly in $L^2(\R)\times L^2(\R)$ and thus we have completed the proof of the lemma.
\end{proof}
}

\bibliographystyle{unsrt}
\bibliography{Chong2025}

\end{document}